\documentclass{amsart}
\copyrightinfo{2018}{Nicholas Coxon}
\usepackage{amssymb,url,booktabs}
\usepackage[table]{xcolor}
\usepackage{tikz, pgfplots}
\pgfplotsset{compat=1.9}
\usetikzlibrary{
	calc,
	matrix,
	backgrounds
}
\newtheorem{theorem}{Theorem}[section]
\newtheorem{proposition}[theorem]{Proposition}
\newtheorem{lemma}[theorem]{Lemma}
\newtheorem{corollary}[theorem]{Corollary}
\theoremstyle{definition}
\newtheorem{definition}[theorem]{Definition}
\newtheorem{example}[theorem]{Example}
\theoremstyle{remark}
\newtheorem{remark}[theorem]{Remark}
\numberwithin{equation}{section}
\usepackage[noend]{algpseudocode}
\usepackage{algorithm}

\makeatletter
	\algnewcommand{\LineComment}[1]{\Statex\hskip\ALG@thistlm\hskip\ALG@tlm #1}
	\def\ALG@doentity
	{%
		\edef\ALG@thisblock{\csname ALG@currentblock@\theALG@nested\endcsname}%
		\expandafter\ifx\csname ALG@b@\ALG@L @\ALG@thisentity @\ALG@thisblock\endcsname\relax%
			\def\ALG@thisblock{0}%
		\fi%
		\ALG@getentitytext%
		\ifnum\ALG@thisblock=0\else\ALG@doend\fi%
		\ifx\ALG@text\ALG@x@notext%
		\else%
			\item%
		\fi%
		\noindent\hskip\ALG@tlm%
		\expandafter\ifnum0=\csname ALG@b@\ALG@L @\ALG@thisentity @\ALG@thisblock\endcsname\else%
			\ALG@dobegin%
		\fi%
		\def\ALG@entitiecommand{\ALG@displayentity}%
	}%
\makeatother
\newcommand{\dg}{d}
\newcommand{\fint}{g}
\newcommand{\fout}{h}
\newcommand{\ftay}{c}
\newcommand{\floor}[1]{\left\lfloor #1\right\rfloor}
\newcommand{\ceil}[1]{\left\lceil #1\right\rceil}
\newcommand{\N}{\mathbb{N}}

\newcommand{\F}{\mathbb{F}}
\newcommand{\flag}{b}
\newcommand{\set}{\leftarrow}
\newcommand{\bigO}{\mathcal{O}}
\newcommand{\Gray}{\Delta}
\newcommand{\leaves}{L}
\newcommand{\PushDown}{\varphi}
\newcommand{\basedeg}{t} 
\DeclareMathOperator{\image}{Im} 
\DeclareMathOperator{\trace}{Tr} 
\DeclareMathOperator*{\argmin}{arg\,min}
\let\originalleft\left
\let\originalright\right
\renewcommand{\left}{\mathopen{}\mathclose\bgroup\originalleft}
\renewcommand{\right}{\aftergroup\egroup\originalright}
\begin{document}
%
%
\title{Fast transforms over finite fields of characteristic two}
%
\author[Nicholas Coxon]{Nicholas Coxon}
\address{INRIA Saclay--\^{I}le-de-France \& Laboratoire d'Informatique, \'{E}cole polytechnique, 91128 Palaiseau Cedex, France}
\email{nicholas.coxon@inria.fr}
\date{\today}
\thanks{This work was supported by Nokia in the framework of the common laboratory between Nokia Bell Labs and INRIA}
%
\begin{abstract}
An additive fast Fourier transform over a finite field of characteristic two efficiently evaluates polynomials at every element of an $\mathbb{F}_2$-linear subspace of the field. We view these transforms as performing a change of basis from the monomial basis to the associated Lagrange basis, and consider the problem of performing the various conversions between these two bases, the associated Newton basis, and the ``novel'' basis of Lin, Chung and Han~(FOCS~2014). Existing algorithms are divided between two families, those designed for arbitrary subspaces and more efficient algorithms designed for specially constructed subspaces of fields with degree equal to a power of two. We generalise techniques from both families to provide new conversion algorithms that may be applied to arbitrary subspaces, but which benefit equally from the specially constructed subspaces. We then construct subspaces of fields with smooth degree for which our algorithms provide better performance than existing algorithms.
\end{abstract}
\maketitle
%
\section{Introduction}

Let $\F$ be a finite field of characteristic two, and $W=\{\omega_0,\dotsc,\omega_{2^n-1}\}$ be an $n$-dimensional $\F_2$-linear subspace of $\F$. Define polynomials
\begin{equation*}
	L_i
	=\prod^{2^n-1}_{\substack{j=0\\j\neq i}}
	\frac{x-\omega_j}
	{\omega_i-\omega_j},
	\quad
	N_i
	=\prod^{i-1}_{j=0}
	\frac{x-\omega_j}
	{\omega_i-\omega_j}
	\quad\text{and}\quad
	X_i
	=\prod^{n-1}_{k=0}
	\prod^{2^k[i]_k-1}_{j=0}
	\frac{x-\omega_j}
	{\omega_{2^k[i]_k}-\omega_j}
\end{equation*}
for $i\in\{0,\dotsc,2^n-1\}$, where $[{}\cdot{}]_k:\N\rightarrow\{0,1\}$ for $k\in\N$ such that
\begin{equation*}
	i=\sum_{k\in\N}2^k[i]_k
	\quad\text{for $i\in\N$}.
\end{equation*}
Let $\F[x]_\ell$ denote the space of polynomials with coefficients in $\F$ and degree less than $\ell$. Then $\{L_0,\dotsc,L_{2^n-1}\}$ is the Lagrange basis of $\F[x]_{2^n}$ associated with the enumeration of $W$. Similarly, $\{N_0,\dotsc,N_{2^n-1}\}$ is the associated Newton basis, normalised so that $N_i(\omega_i)=1$. The definition of the functions $[{}\cdot{}]_k$ implies that each of the polynomials $X_i$ has degree equal to $i$. Thus, $\{X_0,\dotsc,X_{2^n-1}\}$ is also a basis of $\F[x]_{2^n}$. This unusual basis was introduced by Lin, Chung and Han in 2014~\cite{lin2014}. Consequently, we refer to it as the Lin--Chung--Han basis, or simply the LCH basis. In this paper, we describe new fast algorithms for converting between the Lagrange, (normalised) Newton, Lin--Chung--Han and monomial bases for specially constructed subspaces.

Converting to the Lagrange basis from one of the three remaining bases corresponds to evaluating a polynomial at each element in $W$. An algorithm that efficiently performs this evaluation for polynomials written on the monomial basis is referred to as an additive fast Fourier transform (FFT). The designation as ``additive'' reflects the fact that a fast Fourier transform traditionally evaluates polynomials at each element of a cyclic group. To avoid confusion, we refer to such algorithms as multiplicative FFTs hereafter. Additive FFTs have been investigated as an alternative to multiplicative FFTs for use in fast multiplication algorithms for binary polynomials~\cite{vonzurGathen1996,brent2008,mateer2008,chen2017,chen2018,li2018}, and have also found applications in coding theory and cryptography~\cite{bernstein2013,bernstein2014,chou2017,bensasson2017}.

Additive FFTs first appeared in the 1980s with the of algorithm of Wang and Zhu~\cite{wang1988}, which was subsequently rediscovered by Cantor~\cite{cantor1989}. Applied to characteristic two finite fields, the Wang--Zhu--Cantor algorithm is restricted to extensions of degree equal to a power of two. This restriction is removed by the algorithm of von zur Gathen and Gerhard~\cite{vonzurGathen1996}, but at the cost of a higher algebraic complexity. For these and subsequent algorithms~\cite{gao2010,bernstein2013,bernstein2014}, the subspace $W$ is described by an ordered basis $\beta=(\beta_0,\dotsc,\beta_{n-1})$, which also defines the enumeration of the~space~by
\begin{equation}\label{eqn:binary-enumeration}
	\omega_i=\sum^{n-1}_{k=0}[i]_k\beta_k
	\quad\text{for $i\in\{0,\dotsc,2^n-1\}$}.
\end{equation}
For an arbitrary choice of basis, the algorithm of von zur Gathen and Gerhard performs $\bigO(\ell\log^2\ell)$ additions and multiplications, where $\ell=\left|W\right|=2^n$. The subsequent algorithm of Gao and Mateer~\cite{gao2010} achieves the same additive complexity while performing only $\bigO(\ell\log\ell)$ multiplications.

For extensions with degree equal to a power of two, faster algorithms are obtained through a special choice of subspace and its basis. The defining property of these spaces is the existence of a Cantor basis, i.e., a basis $\beta=(\beta_0,\dotsc,\beta_{n-1})$ such that
\begin{equation}\label{eqn:cantor-basis}
	\beta_0=1
	\quad\text{and}\quad
	\beta_i=\beta^2_{i+1}-\beta_{i+1}
	\quad\text{for $i\in\{0,\dotsc,n-2\}$}.
\end{equation}
For subspaces represented by a Cantor basis, the Wang--Zhu--Cantor algorithm performs $\bigO(\ell\log^{\log_23}\ell)$ additions and $\bigO(\ell\log\ell)$ multiplications. Gao and Mateer~\cite{gao2010} also contribute to this special case by providing an algorithm that achieves the same multiplicative complexity while performing only $\bigO(\ell(\log\ell)\log\log\ell)$ additions.

For the same subspace enumeration used in the additive FFTs, Lin, Chung and Han~\cite{lin2014} provide algorithms for converting between the Lagrange and LCH bases that perform $\bigO(\ell\log\ell)$ additions and multiplications. Lin, Chung and Han use their basis and conversion algorithms to provide fast encoding and decoding algorithms for Reed--Solomon codes. This application is further explored in the subsequent work of Lin, Al-Naffouri and Han~\cite{lin2016b} and Lin, Al-Naffouri, Han and Chung~\cite{lin2016a}, while Ben-Sasson et al.~\cite{bensasson2018} utilise the conversion algorithms within their zero-knowledge proof system. Lin, Al-Naffouri, Han and Chung~\cite{lin2016a} additionally consider the problem of converting between the LCH and monomial bases. For subspaces represented by an arbitrary choice of basis, they provide algorithms for converting between the two bases that perform $\bigO(\ell\log^2\ell)$ additions and $\bigO(\ell\log\ell)$ multiplications. For subspaces represented by a Cantor basis they provide algorithms that require only $\bigO(\ell(\log\ell)\log\log\ell)$ additions and perform no multiplications. The techniques used in both cases, as well as for the algorithms of Lin, Chung and Han, originate in the work of Gao and Mateer. This relationship becomes apparent when combining the algorithms to obtain additive FFTs, since one essentially obtains the algorithms of Gao and Mateer.

The techniques developed for additive FFTs have yet to be applied to conversions involving the Newton basis. In the realm of multiplicative FFTs, one has the algorithms of van der Hoeven and Schost~\cite{hoeven2013}, which convert between the monomial basis and the Newton basis associated with the radix-2 truncated Fourier transform points~\cite{hoeven2004,hoeven2005}. Fast conversion between the two basis is a necessary requirement of multivariate evaluation and interpolation algorithms~\cite{hoeven2013,coxon2017} and their application to systematic encoding of Reed--Muller and multiplicity codes~\cite{coxon2017}. For applications in coding theory, characteristic two finite fields are particularly interesting due to their fast arithmetic. However, the algorithms of van der Hoeven and Schost are not suited to such fields as they require the existence of roots of unity with order equal to a power of two. It is likely that this problem may be partially overcome by generalising their algorithm in a manner analogous to the generalisation of the radix-2 truncated Fourier transform~\cite{hoeven2004,hoeven2005} to mixed radices by Larrieu~\cite{larrieu2017}. Developing an algorithm based on the ideas of additive FFTs provides a second and more widely applicable solution to the problem.

In this paper, we describe new fast conversion algorithms between the Lin--Chung--Han basis and each of the Newton, Lagrange and monomial bases. These algorithms may in-turn be combined to obtain fast conversions algorithms between any two of the four bases. We once again represent subspaces by ordered bases, and use~\eqref{eqn:binary-enumeration} for their enumeration.

In Section~\ref{sec:preliminaries}, we show that if the defining basis $\beta=(\beta_0,\dotsc,\beta_{n-1})$ has dimension greater than one and satisfies
\begin{equation}\label{eqn:subfield-condition}
	1,
	\frac{\beta_1}{\beta_0},
	\dotsc,
	\frac{\beta_{\dg-1}}{\beta_0}
	\in\F_{2^\dg}
\end{equation}
for some $\dg\in\{1,\dotsc,n-1\}$, then each of the three conversions problems over the subspace generated by $\beta$ may be efficiently reduced to instances of the problem over the subspaces generated by $\alpha=(\beta_0,\dotsc,\beta_{\dg-1})$ and some vector $\delta\in\F^{n-\dg}$. One may always take $\dg$ equal to one, allowing the reductions to be applied regardless of the choice of~$\beta$. Consequently, fast conversions algorithms are obtained by recursively solving the smaller problems admitted by the reduction, and directly solving the problems for the base case of $n=1$.

Our basis conversion algorithms are described in Section~\ref{sec:algorithms}. Over the subspace generated by an $n$-dimensional basis $\beta$, the algorithms take as input the first $\ell$ coefficients on the input basis of a polynomial in $\F[x]_\ell\subseteq\F[x]_{2^n}$. The algorithms then return the first $\ell$ coefficients of the polynomial on the desired output basis. For conversion between the Lagrange and LCH bases, the first $\ell$ Lagrange basis polynomials do not form a basis of $\F[x]_\ell$ for $\ell<2^n$. Consequently, we embed these cases in the larger case of $\ell=2^n$, after-which we disregard unnecessary parts of the resulting computations so as not to incur a large penalty in complexity. This approach results in what is known as ``pruned'' or ``truncated'' algorithms in the literature on fast Fourier transforms. While truncated additive FFTs have been previously investigated~\cite{vonzurGathen1996,mateer2008,brent2008,bernstein2013,bernstein2014,chen2017,chen2018}, our algorithms for converting between the Lagrange and LCH bases are obtained as analogues of Harvey's ``cache-friendly'' truncated multiplicative FFTs~\cite{harvey2009}. As a consequence of this approach, the algorithms in fact solve slightly more general problems than those just described, allowing us to in-turn provide the slightly generalised additive FFTs required by the fast Hermite interpolation and evaluation algorithms of Coxon~\cite{coxon2018}.

Table~\ref{tab:arbitrary-complexity} provides bounds on the number of additions and multiplications performed by our algorithms for conversion in either direction between the LCH basis and each of the three remaining bases. These bounds omit the cost of a small precomputation requiring $\bigO(n^3)$ field operations. The table also provides bounds on the number of field elements that are required to be stored in auxiliary space by the algorithms. The bound for conversion with the Newton basis is new. For conversion with the Lagrange basis, we only have the algorithm of Lin, Chung and Han~\cite{lin2014} to compare with, and only for the case $\ell=2^n$. Their algorithm performs fewer additions in this case, but only after a much larger precomputation, of unanalysed complexity, that stores $2^n-1$ field elements in auxiliary space. For conversion with the monomial basis, we have the algorithms of Lin et al.~\cite{lin2016a} to compare with, but once again only for the case $\ell=2^n$. Our algorithms perform the same number of additions as their algorithms in this case, while performing fewer multiplications.

\begin{table}[h]
	\setlength{\belowcaptionskip}{0pt}
	\begin{tabular}{llll}
		\toprule
		Basis    & Additions                                         & Multiplications                                  & Auxiliary space  \\
		\midrule
		Newton   & $\ell\left(\ceil{\log_2\ell}-1\right)+1 $         & $\floor{\ell/2}\ceil{\log_2\ell}$                & $\bigO(n^2)$ \\[0.8ex]
		Lagrange & $\floor{\ell/2}\left(3\ceil{\log_2\ell}+1\right)$ & $\floor{\ell/2}\left(\ceil{\log_2\ell}+1\right)$ & $2^n-\ell+\bigO(n^2)$ \\[0.8ex]
		Monomial & $\floor{\ell/2}\binom{\ceil{\log_2\ell}}{2}$      & $3\floor{\ell/2}\ceil{\log_2\ell}+1$             & $\bigO(n)$   \\
		\bottomrule
	\end{tabular}
	\vspace{\abovecaptionskip}
	\caption{Algebraic and space complexities.}\label{tab:arbitrary-complexity}
\end{table}

The cost of the reductions used in our algorithms reduce with the size of the value~$\dg$ for which they are applied. For an arbitrary choice of the basis $\beta$, the condition~\eqref{eqn:subfield-condition} may only ever be satisfied by $\dg$ equal to one. This will also be the case if $\F_2$ is the only subfield of $\F$ with degree less than~$n$. The bounds in Table~\ref{tab:arbitrary-complexity} describe the complexity of our algorithms for this case, and thus represent their worst-case complexities. When the field has degree equal to a power of two and $\beta$ is a Cantor basis, the condition~\eqref{eqn:subfield-condition} is satisfied by $\dg$ equal to any power of two less than~$n$. Moreover, $\delta=(\beta_0,\dotsc,\beta_{n-\dg-1})$ for all such values of $\dg$, so that the recursive cases are themselves represented by Cantor bases. Consequently, it is possible to always take $\dg$ to be the largest power of two less than $n$. With this strategy, the algorithms for converting between the LCH and Newton bases perform only $(3\ell-2)\ceil{\log_2\ell}/4$ additions. The algorithms for converting between the LCH and Lagrange bases enjoy a similar reduction in the number of additions they perform, while the algorithms for converting between the LCH and monomials bases perform only  $\floor{\ell/2}\ceil{\log_2\ell}\ceil{\log_2\log_2\max(\ell,2)}$ additions. Moreover, for conversions between the LCH and monomial bases, all multiplications in the algorithms become multiplications by one, allowing them to be eliminated altogether. In this case, the algorithms reduce to those of Lin et al.~\cite{lin2016a}.

While the benefits of using a Cantor basis are clear, $\F$ will only admit a Cantor basis of a given dimension if its degree is divisible by a sufficiently large power of two. In Section~\ref{sec:construction}, we propose new basis constructions that provide benefits similar to those afforded by Cantor bases when the degree of the field contains a sufficiently large smooth factor, i.e., one that factors into a product of small primes. Such a factor ensures the presence of a tower of subfields, which we use in Section~\ref{sec:tower} to construct bases that reduce the number of field operations performed by the algorithms of Section~\ref{sec:algorithms} by allowing their reductions to be applied more frequently with values of $\dg$ greater than one. For towers containing quadratic extensions, we additionally show in Section~\ref{sec:quadratic} how to leverage freedom in the construction in order to eliminate some multiplications in the algorithms for conversion between the monomial and LCH bases, echoing the reduction in multiplications obtained for Cantor bases. Finally, in Section~\ref{sec:cantor}, we show how to take advantage of quadratic extensions in a different manner by generalising the construction of Cantor bases.

\section{Preliminaries}\label{sec:preliminaries}

For $\beta=(\beta_0,\dotsc,\beta_{n-1})\in\F^n$, define
\begin{equation*}
	\omega_{\beta,i}
	=\sum^{n-1}_{k=0}
	[i]_k\beta_k
	\quad\text{for $i\in\{0,\dotsc,2^n-1\}$}.
\end{equation*}
If the entries of $\beta$ are linearly independent over $\F_2$, then let $L_{\beta,i}$, $N_{\beta,i}$ and $X_{\beta,i}$ respectively denote the $i$th Lagrange, Newton and LCH basis polynomials associated with the enumeration $\{\omega_{\beta,0},\dotsc,\omega_{\beta,2^n-1}\}$ of the subspace it generates. That is, define
\begin{equation*}
	L_{\beta,i}
	=\prod^{2^n-1}_{\substack{j=0\\j\neq i}}
	\frac{x-\omega_{\beta,j}}{\omega_{\beta,i}-\omega_{\beta,j}},
	\qquad
	N_{\beta,i}
	=\prod^{i-1}_{j=0}
	\frac{x-\omega_{\beta,j}}{\omega_{\beta,i}-\omega_{\beta,j}}
\end{equation*}
and
\begin{equation*}
	X_{\beta,i}
	=\prod^{n-1}_{k=0}
	\prod^{2^k[i]_k-1}_{j=0}
	\frac{x-\omega_{\beta,j}}{\omega_{\beta,2^k[i]_k}-\omega_{\beta,j}}
\end{equation*}
for $i\in\{0,\dotsc,2^n-1\}$.

\subsection{Factorisations of basis polynomials}\label{sec:reduction}

The following lemma provides factorisations of the basis polynomials associated with a vector $\beta\in\F^n$. These factorisations in-turn provide the reductions employed in our basis conversion algorithms.

\begin{lemma}\label{lem:reduction} Let $n\geq 2$ and $\beta=(\beta_0,\dotsc,\beta_{n-1})\in\F^n$ have entries that are linearly independent over $\F_2$. For some $\dg\in\{1,\dotsc,n-1\}$ such that $\beta_i/\beta_0\in\F_{2^\dg}$ for $i\in\{0,\dotsc,\dg-1\}$, set $\alpha=(\beta_0,\dotsc,\beta_{\dg-1})$, $\gamma=(\beta_\dg,\dotsc,\beta_{n-1})$ and
\begin{equation*}
	\delta=\left(
		\left(\frac{\beta_\dg}{\beta_0}\right)^{2^\dg}-\frac{\beta_\dg}{\beta_0}
		,\dotsc,
		\left(\frac{\beta_{n-1}}{\beta_0}\right)^{2^\dg}-\frac{\beta_{n-1}}{\beta_0}
	\right).
\end{equation*}
Then $\delta$ has entries that are linearly independent over $\F_2$, and
\begin{equation*}
	\begin{aligned}
		L_{\beta,2^\dg i+j}
		&=
		L_{\delta,i}
		\left(
			\left(\frac{x}{\beta_0}\right)^{2^\dg}-\frac{x}{\beta_0}
		\right)
		L_{\alpha,j}\left(x-\omega_{\gamma,i}\right),\\
		N_{\beta,2^\dg i+j}
		&=
		N_{\delta,i}
		\left(
			\left(\frac{x}{\beta_0}\right)^{2^\dg}-\frac{x}{\beta_0}
		\right)
		N_{\alpha,j}\left(x-\omega_{\gamma,i}\right),\\
		X_{\beta,2^\dg i+j}
		&=
		X_{\delta,i}
		\left(
			\left(\frac{x}{\beta_0}\right)^{2^\dg}-\frac{x}{\beta_0}
		\right)
		X_{\alpha,j}\left(x\right)
	\end{aligned}
\end{equation*}
for $i\in\{0,\dotsc,2^{n-\dg}-1\}$ and $j\in\{0,\dotsc,2^\dg-1\}$.
\end{lemma}
\begin{proof} Let $n\geq 2$ and $\beta=(\beta_0,\dotsc,\beta_{n-1})\in\F^n$ have entries that are linearly independent over $\F_2$. Define $\alpha$, $\gamma$ and $\delta$ as per the lemma for some $\dg\in\{1,\dotsc,n-1\}$ such that $\beta_i/\beta_0\in\F_{2^\dg}$ for $i\in\{0,\dotsc,\dg-1\}$. Then
\begin{equation}\label{eqn:omega-subspace-sum}
	\omega_{\beta,2^\dg i+j}
	=
	\sum^{\dg-1}_{k=0}
	[j]_k\beta_k
	+\sum^{n-\dg-1}_{k=0}
	[i]_k\beta_{\dg+k}
	=
	\omega_{\alpha,j}
	+\omega_{\gamma,i}
\end{equation}
for $i\in\{0,\dotsc,2^{n-\dg}-1\}$ and $j\in\{0,\dotsc,2^\dg-1\}$. The choice of $\dg$ implies that
\begin{equation*}
	\left(
		\frac{\beta_k}
		{\beta_0}
	\right)^{2^\dg}
	-
	\frac{\beta_k}
	{\beta_0}
	=\prod_{\omega\in\F_{2^\dg}}
	\frac{\beta_k}
	{\beta_0}
	-\omega
	=0
	\quad\text{for $k\in\{0,\dotsc,\dg-1\}$}.
\end{equation*}
Thus,
\begin{equation}\label{eqn:omega-subspace-image}
	\left(\frac{\omega_{\beta,2^\dg i+j}}{\beta_0}\right)^{2^\dg}
	-\frac{\omega_{\beta,2^\dg i+j}}{\beta_0}
	=
	\sum^{n-\dg-1}_{k=0}
	[i]_k
	\left(
		\left(\frac{\beta_{\dg+k}}{\beta_0}\right)^{2^\dg}
		-\frac{\beta_{\dg+k}}{\beta_0}
	\right)
	=\omega_{\delta,i}
\end{equation}
for $i\in\{0,\dotsc,2^{n-\dg}-1\}$ and $j\in\{0,\dotsc,2^\dg-1\}$. It follows that
\begin{equation}\label{eqn:product-indentity}
	\prod^{2^\dg-1}_{j=0}
	x-\omega_{\beta,2^\dg i+j}
	=
	\beta^{2^\dg}_0
	\left(
		\left(\frac{x}{\beta_0}\right)^{2^\dg}
		-\left(\frac{x}{\beta_0}\right)
		-\omega_{\delta,i}
	\right)
\end{equation}
for $i\in\{0,\dotsc,2^{n-\dg}-1\}$, since the polynomials on either side of the equation are monic and split over $\F$ with identical roots. As the entries of $\beta$ are linearly independent over $\F_2$, comparing roots shows that the polynomials on the left-hand side of~\eqref{eqn:product-indentity} are distinct for different values of $i\in\{0,\dotsc,2^{n-\dg}-1\}$. Consequently, comparing the polynomials on the right-hand side of the equation shows that $\omega_{\delta,i}\neq\omega_{\delta,j}$ for distinct $i,j\in\{0,\dotsc,2^{n-\dg}-1\}$. Thus, the entries of $\delta$ are linearly independent over $\F_2$.

Collecting terms and substituting in \eqref{eqn:omega-subspace-sum}, \eqref{eqn:omega-subspace-image} and \eqref{eqn:product-indentity} shows that
\begin{align*}
	L_{\beta,2^\dg i+j}
	&=
	\left(
		\prod^{2^{n-\dg}-1}_{\substack{s=0\\ s\neq i}}
		\prod^{2^\dg-1}_{t=0}
		\frac{x-\omega_{\beta,2^\dg s+t}}
		{\omega_{\beta,2^\dg i+j}-\omega_{\beta,2^\dg s+t}}
	\right)
	\left(
		\prod^{2^\dg-1}_{\substack{t=0\\ t\neq j}}
		\frac{x-\omega_{\beta,2^\dg i+t}}
		{\omega_{\beta,2^\dg i+j}-\omega_{\beta,2^\dg i+t}}
	\right)\\
	&=
	\left(
		\prod^{2^{n-\dg}-1}_{\substack{s=0\\ s\neq i}}
		\frac{
			\left(x/\beta_0\right)^{2^\dg}
			-\left(x/\beta_0\right)
			-\omega_{\delta,s}
		}{
			\omega_{\delta,i}
			-\omega_{\delta,s}
		}
	\right)
	\left(
		\prod^{2^\dg-1}_{\substack{t=0\\ t\neq j}}
		\frac{x-\omega_{\gamma,i}-\omega_{\alpha,t}}
		{\omega_{\gamma,i}+\omega_{\alpha,j}-\omega_{\gamma,i}-\omega_{\alpha,t}}
	\right)\\
	&=
	\left(
		\prod^{2^{n-\dg}-1}_{\substack{s=0\\ s\neq i}}
		\frac{
			\left(x/\beta_0\right)^{2^\dg}
			-\left(x/\beta_0\right)
			-\omega_{\delta,s}
		}{
			\omega_{\delta,i}
			-\omega_{\delta,s}
		}
	\right)
	\left(
		\prod^{2^\dg-1}_{\substack{t=0\\ t\neq j}}
		\frac{x-\omega_{\gamma,i}-\omega_{\alpha,t}}
		{\omega_{\alpha,j}-\omega_{\alpha,t}}
	\right)\\
	&=
	L_{\delta,i}
	\left(
		\left(\frac{x}{\beta_0}\right)^{2^\dg}-\frac{x}{\beta_0}
	\right)
	L_{\alpha,j}\left(x-\omega_{\gamma,i}\right)
\end{align*}
for $i\in\{0,\dotsc,2^{n-\dg}-1\}$ and $j\in\{0,\dotsc,2^\dg-1\}$.

\noindent Similarly,
\begin{align*}
	N_{\beta,2^\dg i+j}
	&=
	\left(
		\prod^{i-1}_{s=0}
		\prod^{2^\dg-1}_{t=0}
		\frac{x-\omega_{\beta,2^\dg s+t}}
		{\omega_{\beta,2^\dg i+j}-\omega_{\beta,2^\dg s+t}}
	\right)
	\left(
		\prod^{j-1}_{t=0}
		\frac{x-\omega_{\beta,2^\dg i+t}}
		{\omega_{\beta,2^\dg i+j}-\omega_{\beta,2^\dg i+t}}
	\right)\\
	&=
	\left(
		\prod^{i-1}_{s=0}
		\frac{
			\left(x/\beta_0\right)^{2^\dg}
			-\left(x/\beta_0\right)
			-\omega_{\delta,s}
		}{
			\omega_{\delta,i}
			-\omega_{\delta,s}
		}
	\right)
	\left(
		\prod^{j-1}_{t=0}
		\frac{x-\omega_{\gamma,i}-\omega_{\alpha,t}}
		{\omega_{\alpha,j}-\omega_{\alpha,t}}
	\right)\\
	&=
	N_{\delta,i}
	\left(
		\left(\frac{x}{\beta_0}\right)^{2^\dg}-\frac{x}{\beta_0}
	\right)
	N_{\alpha,j}\left(x-\omega_{\gamma,i}\right)
\end{align*}
for $i\in\{0,\dotsc,2^{n-\dg}-1\}$ and $j\in\{0,\dotsc,2^\dg-1\}$. It follows immediately from the definition of the Newton and LCH bases that
\begin{equation*}
	X_{\beta,i}
	=\prod^{n-1}_{k=0}
	N_{\beta,2^k[i]_k}
	\quad\text{for $i\in\{0,\dotsc,2^n-1\}$}.
\end{equation*}
Hence,
\begin{align*}
	X_{\beta,2^\dg i+j}
	&=\left(
		\prod^{n-\dg-1}_{k=0}
		N_{\beta,2^{k+\dg}[i]_k}
	\right)
	\left(
		\prod^{\dg-1}_{k=0}
		N_{\beta,2^k[j]_k}
	\right)\\
	&=
	\left(
		\prod^{n-\dg-1}_{k=0}
		N_{\delta,2^k[i]_k}
		\left(
			\left(\frac{x}{\beta_0}\right)^{2^\dg}
			-\frac{x}{\beta_0}
		\right)
	\right)
	\left(
		\prod^{\dg-1}_{k=0}
		N_{\alpha,2^k[j]_k}\left(x-\omega_{\gamma,0}\right)
	\right)\\
	&=
	X_{\delta,i}
	\left(
		\left(\frac{x}{\beta_0}\right)^{2^\dg}-\frac{x}{\beta_0}
	\right)
	X_{\alpha,j}(x)
\end{align*}
for $i\in\{0,\dotsc,2^{n-\dg}-1\}$ and $j\in\{0,\dotsc,2^\dg-1\}$.
\end{proof}

To illustrate how we can utilise Lemma~\ref{lem:reduction}, let us consider the problem of converting polynomials from the Lagrange basis to the LCH basis. The factorisation of the Lagrange basis polynomials provided by the lemma includes a shift of variables for one factor. Consequently, a recursive approach is facilitated by considering the more general problem of converting from a basis of shifted Lagrange polynomials to the LCH basis. An instance of this new problem is defined by a vector $\beta=(\beta_0,\dotsc,\beta_{n-1})\in\F^n$ with linearly independent entries over $\F_2$, a shift parameter $\lambda\in\F$, and coefficients $f_0,\dotsc,f_{2^n-1}\in\F$. The goal is then to compute $\fout_0,\dotsc,\fout_{2^n-1}\in\F$ such that
\begin{equation}\label{eqn:setup-LX}
	\sum^{2^n-1}_{i=0}\fout_iX_{\beta,i}=\sum^{2^n-1}_{i=0}f_iL_{\beta,i}\left(x-\lambda\right).
\end{equation}
If $n=1$, then
\begin{equation*}
	L_{\beta,0}=\frac{x}{\beta_0}+1,
	\quad
	L_{\beta,1}=\frac{x}{\beta_0},
	\quad
	X_{\beta,0}=1
	\quad\text{and}\quad
	X_{\beta,1}=\frac{x}{\beta_0}.
\end{equation*}
Thus, one can simply compute $\fout_0=f_0-(\lambda/\beta_0)(f_0+f_1)$ and $\fout_1=f_0+f_1$. If $n\geq 2$, then the following consequence of Lemma~\ref{lem:reduction} decomposes the length $2^n$ problem into $2^{n-\dg}$ problems of length $2^\dg$, and $2^\dg$ problems of length $2^{n-\dg}$, for $\dg\in\{1,\dotsc,n-1\}$ such that $\beta_i/\beta_0\in\F_{2^\dg}$ for $i\in\{0,\dotsc,\dg-1\}$. After choosing such a value of $\dg$, for which one always has the possibility of taking $\dg=1$, the smaller instances of the problem admitted by the decomposition can be solved recursively.

\begin{lemma}\label{lem:setup-LX-reduction} Suppose that $\beta=(\beta_0,\dotsc,\beta_{n-1})\in\F^n$ has $n\geq 2$ linearly independent entries over $\F_2$, and let $\alpha$, $\gamma$ and $\delta$ be defined as per Lemma~\ref{lem:reduction} for some $\dg\in\{1,\dotsc,n-1\}$ such that $\beta_i/\beta_0\in\F_{2^\dg}$ for $i\in\{0,\dotsc,\dg-1\}$. Suppose that $f_0,\dotsc,f_{2^n-1},\lambda,\fint_0,\dotsc,\fint_{2^n-1},\fout_0,\dotsc,\fout_{2^n-1}\in\F$ satisfy
\begin{equation}\label{eqn:setup-LX-rows}
	\sum^{2^\dg-1}_{j=0}
	\fint_{2^{\dg}i+j}
	X_{\alpha,j}
	=
	\sum^{2^{\dg}-1}_{j=0}
	f_{2^{\dg}i+j}
	L_{\alpha,j}\left(x-\lambda-\omega_{\gamma,i}\right)
	\quad\text{for $i\in\{0,\dotsc,2^{n-\dg}-1\}$},
\end{equation}
and
\begin{equation}\label{eqn:setup-LX-cols}
	\sum^{2^{n-\dg}-1}_{i=0}
	\fout_{2^{\dg}i+j}
	X_{\delta,i}
	=
	\sum^{2^{n-\dg}-1}_{i=0}
	\fint_{2^\dg i+j}
	L_{\delta,i}
	\left(x-\eta\right)
	\quad\text{for $j\in\{0,\dotsc,2^\dg-1\}$},
\end{equation}
where $\eta=(\lambda/\beta_0)^{2^\dg}-(\lambda/\beta_0)$. Then \eqref{eqn:setup-LX} holds.
\end{lemma}
\begin{proof} Suppose that $\beta=(\beta_0,\dotsc,\beta_{n-1})\in\F^n$ has $n\geq 2$ linearly independent entries over $\F_2$, and let $\alpha$, $\gamma$ and $\delta$ be defined as per Lemma~\ref{lem:reduction} for some $\dg\in\{1,\dotsc,{n-1}\}$ such that $\beta_i/\beta_0\in\F_{2^\dg}$ for $i\in\{0,\dotsc,\dg-1\}$. Suppose that $f_0,\dotsc,f_{2^n-1},\lambda$, $\fint_0,\dotsc,\fint_{2^n-1},\fout_0,\dotsc,\fout_{2^n-1}\in\F$ satisfy~\eqref{eqn:setup-LX-rows} and~\eqref{eqn:setup-LX-cols}. Then Lemma~\ref{lem:reduction} implies that
\begin{multline*}
	\sum^{2^n-1}_{i=0}f_iL_{\beta,i}\left(x-\lambda\right)
	=
	\sum^{2^{n-\dg}-1}_{i=0}
	\left(
		\sum^{2^\dg-1}_{j=0}
		f_{2^\dg i+j}
		L_{\alpha,j}\left(x-\lambda-\omega_{\gamma,i}\right)
	\right)\\
	\times
	L_{\delta,i}
	\left(
		\left(\frac{x}{\beta_0}\right)^{2^\dg}
		-\frac{x}{\beta_0}
		-\eta
	\right),
\end{multline*}
where $\eta=(\lambda/\beta_0)^{2^\dg}-(\lambda/\beta_0)$. By substituting in~\eqref{eqn:setup-LX-rows} and~\eqref{eqn:setup-LX-cols}, it follows that
\begin{equation*}
	\begin{split}
	\sum^{2^n-1}_{i=0}f_iL_{\beta,i}\left(x-\lambda\right)
	&=
	\sum^{2^{n-\dg}-1}_{i=0}
	\left(
		\sum^{2^\dg-1}_{j=0}
		\fint_{2^\dg i+j}
		X_{\alpha,j}\left(x\right)
	\right)
	L_{\delta,i}
	\left(
		\left(\frac{x}{\beta_0}\right)^{2^\dg}
		-\frac{x}{\beta_0}
		-\eta
	\right)\\
	&=
	\sum^{2^\dg-1}_{j=0}
	\left(
		\sum^{2^{n-\dg}-1}_{i=0}
		\fint_{2^\dg i+j}
		L_{\delta,i}
		\left(
			\left(\frac{x}{\beta_0}\right)^{2^\dg}
			-\frac{x}{\beta_0}
			-\eta
		\right)
	\right)
	X_{\alpha,j}\left(x\right)\\
	&=
	\sum^{2^\dg-1}_{j=0}
	\left(
		\sum^{2^{n-\dg}-1}_{i=0}
		\fout_{2^\dg i+j}
		X_{\delta,i}
		\left(
			\left(\frac{x}{\beta_0}\right)^{2^\dg}
			-\frac{x}{\beta_0}
		\right)
	\right)
	X_{\alpha,j}\left(x\right).
	\end{split}
\end{equation*}
Hence, Lemma~\ref{lem:reduction} implies that~\eqref{eqn:setup-LX} holds.
\end{proof}

\subsection{Reduction trees}

We use full binary trees to encode the values of $\dg$ for which the reductions provided by Lemma~\ref{lem:reduction} are applied when converting between two of the bases.

\begin{definition} A full binary tree is a tree that contains a unique vertex of degree zero or two, while all other vertices have degree one or three. Given a full binary tree $(V,E)$ with unique vertex $r\in V$ of degree zero or two, we use the following nomenclature:
\begin{itemize}
	\item the vertex $r$ is called the root of the tree,
	\item a vertex of degree at most one is called a leaf,
	\item a vertex of degree greater than one is called an internal vertex,
	\item if $v\in V\setminus\{r\}$ and $v=v_0,v_1,\dotsc,v_n=r$ is a path, then $v$ is called a child of $v_1$, and a descendant of $v_i$ for $i\in\{1,\dotsc,n\}$,
	\item the set of leaves that are descended from or equal to $v\in V$ is denoted $\leaves_v$,
	\item the subtree of $(V,E)$ rooted on $v\in V$ is the full binary tree $(V',E')$ such that $V'$ consists of $v$ and all its descendants, and $E'$ consists of all edges $\{u,v\}\in E$ such that $u,v\in V'$.
\end{itemize}
\end{definition}

\begin{example} The graph shown in Figure~\ref{fig:binary-tree} is a full binary tree with root vertex~$v_0$, internal vertices $v_0$, $v_1$, $v_2$ and $v_6$, and leaves $v_3$, $v_4$, $v_5$, $v_7$ and $v_8$. The vertex $v_1$ has children $v_2$ and $v_5$, while its descendants are $v_2$, $v_3$, $v_4$ and~$v_5$. Consequently, the subtree rooted on $v_1$ consists of those vertices and edges contained in the dotted box. Finally, $\leaves_{v_0}=\{v_3,v_4,v_5,v_7,v_8\}$, $\leaves_{v_1}=\{v_3,v_4,v_5\}$, $\leaves_{v_2}=\{v_3,v_4\}$, $\leaves_{v_6}=\{v_7,v_8\}$ and $\leaves_{v_i}=\{v_i\}$ for $i\in\{3,4,5,7,8\}$.
\begin{figure}[h]
	\begin{tikzpicture}
		\tikzset{vertex/.style={circle,fill=black,inner sep=0pt,minimum size=3pt,label=left:{\small#1}}}
		\tikzset{rvertex/.style={circle,fill=black,inner sep=0pt,minimum size=3pt,label=right:{\small#1}}}
		
		\def\xr{1.5}
		\def\xa{1.25}
		\def\xb{1}
		\def\yr{1.25}
		\def\ya{1.25}
		\def\yb{1.25}
		
		\node [vertex = $v_0$] (v0) at (0,0) {};
		\node [vertex = $v_1$] (v1) at ($(v0)+(-\xr,-\yr)$) {};
		\node [vertex = $v_2$] (v2) at ($(v1)+(-\xa,-\ya)$) {};
		\node [vertex = $v_3$] (v3) at ($(v2)+(-\xb,-\yb)$) {};
		\node [vertex = $v_4$] (v4) at ($(v2)+( \xb,-\yb)$) {};
		\node [vertex = $v_5$] (v5) at ($(v1)+( \xa,-\ya)$) {};
		\node [rvertex = $v_6$] (v6) at ($(v0)+( \xr,-\yr)$) {};
		\node [rvertex = $v_7$] (v7) at ($(v6)+(-\xa,-\ya)$) {};
		\node [rvertex = $v_8$] (v8) at ($(v6)+( \xa,-\ya)$) {};
		
		\draw (v0) -- (v1);
		\draw (v1) -- (v2);
		\draw (v2) -- (v3);
		\draw (v2) -- (v4);
		\draw (v1) -- (v5);
		\draw (v0) -- (v6);
		\draw (v6) -- (v7);
		\draw (v6) -- (v8);
		
		\draw[dotted] ($(v1)+(\xa,0)+(0.2,0.2)$) rectangle ($(v3)+(-0.6,-0.2)$);
	\end{tikzpicture}
	\caption{A full binary tree.}
	\label{fig:binary-tree}
\end{figure}
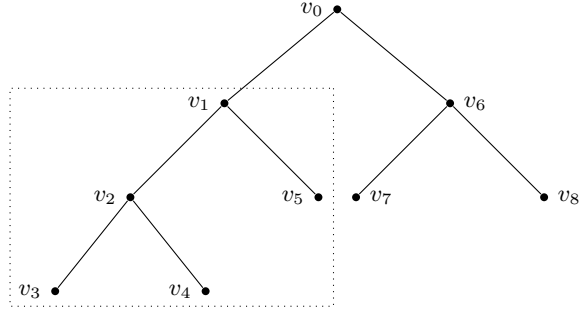
\end{example}

Each internal vertex of a full binary tree has exactly two children. As it is necessary for us to distinguish between the children of internal vertices in our algorithms, we assume that each full binary tree $(V,E)$ comes equipped with a partition $\{E_\alpha,E_\delta\}$ of $E$ such that if $v\in V$ has children $v_0$ and $v_1$, then $\{v,v_i\}\in E_\alpha$ if and only if $\{v,v_{1-i}\}\in E_\delta$. Then for each internal vertex $v\in V$, we denote by $v_\alpha$ the child of $v$ such that $\{v,v_\alpha\}\in E_\alpha$, and by $v_\delta$ the child of $v$ such that $\{v,v_\delta\}\in E_\delta$. The partition additionally induces a vertex labelling $\dg:V\rightarrow\N$ given by
\begin{equation*}
	\dg(v)
	=\begin{cases}
		0                               & \text{if $v$ is a leaf vertex},     \\
		\left|\leaves_{v_\alpha}\right| & \text{if $v$ is an internal vertex}.
	\end{cases}
\end{equation*}

\begin{example} We obtain one such partition $\{E_\alpha,E_\delta\}$ for the tree shown in Figure~\ref{fig:binary-tree} by letting $E_\alpha$ contain the leftmost (as shown in the figure), and $E_\delta$ the rightmost, of the two edges that connect each internal vertex with its children. Then $\dg(v_0)=\left|\leaves_{v_1}\right|=3$, $\dg(v_1)=\left|\leaves_{v_2}\right|=2$, $\dg(v_2)=\left|\leaves_{v_3}\right|=1$, $\dg(v_6)=\left|\leaves_{v_7}\right|=1$ and $\dg(v_i)=0$ for $i\in\{3,4,5,7,8\}$.
\end{example}

We use the vertex labelling to encode the values of $\dg$ for which the reductions provided by Lemma~\ref{lem:reduction} are applied. A reduction tree is a full binary tree that fulfils this role for some subspace basis~$\beta\in\F^n$. To allow us to formally define this notion, we now introduce maps that send $\beta$ to each of the vectors $\alpha$ and $\delta$ defined in Lemma~\ref{lem:reduction}. For $\beta=(\beta_0,\dotsc,\beta_{n-1})\in\F^n$ with $n\geq 2$ linearly independent entries over $\F_2$ and $\dg\in\{1,\dotsc,n-1\}$, define $\alpha(\beta,\dg)=(\beta_0,\dotsc,\beta_{\dg-1})$ and
\begin{equation*}
	\delta\left(\beta,\dg\right)
	=\left(
		\left(
			\frac{\beta_\dg}{\beta_0}
		\right)^{2^\dg}
		-\frac{\beta_\dg}{\beta_0}
		,\dotsc,
		\left(
			\frac{\beta_{n-1}}{\beta_0}
		\right)^{2^\dg}
		-\frac{\beta_{n-1}}{\beta_0}
	\right).
\end{equation*}
When $\dg$ has the additional property that $\beta_i/\beta_0\in\F_{2^\dg}$ for $i\in\{0,\dotsc,\dg-1\}$, Lemma~\ref{lem:reduction} implies that the vector $\delta(\beta,\dg)$ has linearly independent entries over~$\F_2$.

\begin{definition}\label{def:reduction-tree} Let $\beta\in\F^n$ have entries that are linearly independent over $\F_2$, and $(V,E)$ be a full binary tree with root vertex $r\in V$. Then $(V,E)$ is a reduction tree for $\beta$ if it has $n$ leaves, and the following conditions hold if $n>1$:
\begin{enumerate}
	\item\label{item:reduction-tree-quotients} $\beta_i/\beta_0\in\F_{2^{\dg(r)}}$ for $i\in\{0,\dotsc,\dg(r)-1\}$,
	\item\label{item:reduction-tree-alpha} the subtree of $(V,E)$ rooted on $r_\alpha$ is a reduction tree for $\alpha(\beta,\dg(r))$, and
	\item\label{item:reduction-tree-delta} the subtree of $(V,E)$ rooted on $r_\delta$ is a reduction tree for $\delta(\beta,\dg(r))$.
\end{enumerate}
\end{definition}

If $v$ is an internal vertex of a full binary tree, then
\begin{equation*}
	\left|\leaves_{v_\alpha}\right|
	=\dg(v)
	\quad\text{and}\quad
	\left|\leaves_{v_\delta}\right|
	=\left|\leaves_v\right|-\left|\leaves_{v_\alpha}\right|
	=\left|\leaves_v\right|-\dg(v),
\end{equation*}
since $\{\leaves_{v_\alpha},\leaves_{v_\delta}\}$ is a partition of $\leaves_v$. Therefore, given a basis vector $\beta$ and one of its reduction trees $(V,E)$, there exists vectors $\beta_v=(\beta_{v,0},\dotsc,\beta_{v,\left|\leaves_v\right|-1})\in\F^{\left|\leaves_v\right|}$ for $v\in V$, each with linearly independent entries over $\F_2$, such that
\begin{equation*}
	\frac{\beta_{v,0}}{\beta_{v,0}},
	\dotsc,
	\frac{\beta_{v,\dg(v)-1}}{\beta_{v,0}}
	\in\F_{2^{\dg(v)}},\quad
	\alpha\left(\beta_v,\dg(v)\right)
	=\beta_{v_\alpha}
	\quad\text{and}\quad
	\delta\left(\beta_v,\dg(v)\right)
	=\beta_{v_\delta}
\end{equation*}
for all internal $v\in V$, and $\beta_v=\beta$ if $v$ is the root of the tree. These vectors define instances of each of the basis conversions problem over the subspaces they generate. If $v$ is an internal vertex, then Lemma~\ref{lem:reduction} allows instances of the conversions problems over the subspace generated by $\beta_v$ to be reduced to instances over the subspaces generated by $\beta_{v_\alpha}$ and $\beta_{v_\delta}$. If $v$ is in a leaf, then $\beta_v$ is $1$-dimensional and the corresponding conversion problems may be solved directly. Consequently, the existence of a reduction tree allows the basis conversion problems to be recursively solved with recursion depth equal to that of the tree, i.e., equal to the length of the longest path between its root vertex and one of its leaves.

As the first condition of Definition~\ref{def:reduction-tree} is trivially satisfied if $\dg(r)=1$, the existence of a reduction tree for an arbitrarily chosen subspace basis is guaranteed.

\begin{proposition}\label{prop:trivial-reduction-tree} Let $\beta\in\F^n$ have entries that are linearly independent over $\F_2$. Then every full binary tree with $n$ leaves and $\image(\dg)\subseteq\{0,1\}$ is a reduction tree~for~$\beta$.
\end{proposition}

It is straightforward to prove Proposition~\ref{prop:trivial-reduction-tree} by induction on $n$, or to obtain the proposition as a consequence of Theorem~\ref{thm:sufficient} in Section~\ref{sec:tower}. Consequently, we omit its proof. The choice of reduction trees provided by the proposition captures the strategy used in recent algorithms~\cite{gao2010,lin2014,lin2016a} for an arbitrary choice of subspace basis. Accordingly, we use such trees as our baseline for comparison. The reduction strategy use by recent algorithms specific to Cantor bases~\cite{gao2010,lin2016a} is captured by reduction trees such that $\dg(v)=2^{\ceil{\log_2\left|\leaves_v\right|}-1}$ for all internal vertices. We characterise the reduction trees of Cantor bases in the following proposition.

\begin{proposition}\label{prop:cantor-trees} Suppose that $\beta\in\F^n$ is a Cantor basis. Then a full binary tree is a reduction tree for $\beta$ if and only if it has $n$ leaves and $\image(\dg)\subseteq\{0,2^0,2^1,2^2,\dotsc\}$.
\end{proposition}

We use the following properties of Cantor bases to prove Proposition~\ref{prop:cantor-trees}.

\begin{lemma}[Properties of Cantor bases]\label{lem:cantor} Suppose that $\beta=(\beta_0,\dotsc,\beta_{n-1})\in\F^n$ is a Cantor basis. Then
\begin{enumerate}
	\item\label{cantor-li} $\beta_0,\dotsc,\beta_{n-1}$ are linearly independent over $\F_2$,
	\item\label{cantor-subfields} $\beta_0,\dotsc,\beta_{i-1}\in\F_{2^{2^k}}$ if $i\leq2^k$ for some $k\in\N$, and
	\item\label{cantor-delta}$\beta^{2^{2^k}}_i-\beta_i=\beta_{i-2^k}$ if $i\geq 2^k$ for some $k\in\N$.
\end{enumerate}
\end{lemma}

Lemma~\ref{lem:cantor} is proved by Gao and Mateer~\cite[Apendix~A]{gao2010}, and also obtained as a special case of Lemma~\ref{lem:gen-cantor} in Section~\ref{sec:cantor}.

\begin{proof}[Proof of Proposition~\ref{prop:cantor-trees}] A full binary tree with one leaf has $\image(\dg)=\{0\}$. Thus, the proposition holds for $n=1$, since a full binary tree is a reduction tree for a $1$-dimensional Cantor basis if and only if it has one leaf. Therefore, let $n\geq 2$, and suppose that the proposition is true for all smaller values of $n$. Suppose that $\beta=(\beta_0,\dotsc,\beta_{n-1})\in\F^n$ is a Cantor basis, and let $(V,E)$ be a full binary tree with root vertex $r\in V$. Then $(V,E)$ is a reduction tree for $\beta$ if and only if it has $n$ leaves, $\beta_0,\dotsc,\beta_{\dg(r)-1}\in\F_{2^{\dg(r)}}$ (recall that $\beta_0=1$ for a Cantor basis), the subtree rooted on $r_\alpha$ is a reduction tree for $\alpha(\beta,\dg(r))$, and the subtree rooted on $r_\delta$ is a reduction tree for $\delta(\beta,\dg(r))$. If the tree has $n$ leaves, then ${1\leq\dg(r)<n}$ and properties~\eqref{cantor-li} and~\eqref{cantor-subfields} of Lemma~\ref{lem:cantor} imply that $\beta_0,\dotsc,\beta_{\dg(r)-1}\in\F_{2^{\dg(r)}}$ if and only if $\dg(r)$ is a power of two. In this case, property~\eqref{cantor-delta} of the lemma implies that $\alpha(\beta,\dg(r))=(\beta_0,\dotsc,\beta_{\dg(r)-1})$ and $\delta(\beta,\dg(r))=(\beta_0,\dotsc,\beta_{n-\dg(r)-1})$. Finally, if the tree has $n$ leaves, then the subtree rooted on $r_\alpha$ has $\dg(r)$ leaves, and the subtree rooted on $r_\delta$ has $\left|\leaves_r\right|-\dg(r)=n-\dg(r)$ leaves. Therefore, the induction hypothesis implies that $(V,E)$ is a reduction tree for $\beta$ if and only if it has $n$ leaves, $\dg(r)$ is a power of two, and the subtrees rooted on $r_\alpha$ and $r_\delta$ both satisfy $\image(\dg)\subseteq\{0,2^0,2^1,2^2,\dotsc\}$. That is, if and only if the tree has $n$ leaves and $\image(\dg)\subseteq\{0,2^0,2^1,2^2,\dotsc\}$. Hence, the proposition follows by induction.
\end{proof}

\section{Conversion algorithms}\label{sec:algorithms}

In this section, we describe how to convert between the polynomial bases once~$\beta$ and a suitable reduction tree have been chosen. Accordingly, we now fix a vector $\beta=(\beta_0,\dotsc,\beta_{n-1})\in\F^n$ that has linearly independent entries over $\F_2$, and a reduction tree $(V,E)$ for $\beta$. The algorithms of this section are then presented for this arbitrary selection. Recall that we only provide algorithms for converting between the LCH basis and each of the Newton, Lagrange and monomial bases. The algorithms that we propose for each of these problems require the precomputation of a small number of constants associated with the basis and the reduction tree. Consequently, we begin the section by discussing these precomputations and their cost.

\subsection{Precomputations}

For the remainder of the section, we use the shorthand notations $\dg_v=\dg(v)$ and $n_v=\left|\leaves_v\right|$ for $v\in V$. Define $\beta_v=(\beta_{v,0},\dotsc,\beta_{v,n_v-1})\in\F^{n_v}$ for $v\in V$ recursively as follows: if $v$ is the root of the tree, then $\beta_v=\beta$; if $v$ is an internal vertex, then $\beta_{v_\alpha}=\alpha(\beta_v,\dg_v)$ and $\beta_{v_\delta}=\delta(\beta_v,\dg_v)$. Then Definition~\ref{def:reduction-tree} and Lemma~\ref{lem:reduction} imply that the vectors $\beta_v$ for $v\in V$ each have entries that are linearly independent over $\F_2$. Finally, let $\gamma_v=(\beta_{v,d_v},\dotsc,\beta_{v,n_v-1})$ for all internal $v\in V$. 

Given $\beta_v$ for some internal $v\in V$, repeated squaring allows $\delta(\beta_v,\dg_v)$ to be computed with $(n_v-\dg_v)(\dg_v+1)$ multiplications and $n_v-\dg_v$ additions. If $r\in V$ is the root vertex of the tree, then $V\setminus\leaves_r$ is the set of internal vertices in $V$, and
\begin{equation}\label{eqn:sum-identity}
	\sum_{v\in V\setminus\leaves_r}
	\dg_v
	\left(n_v-\dg_v\right)
	=
	\sum_{v\in V\setminus\leaves_r}
	\left|\leaves_{v_\alpha}\right|
	\left|\leaves_{v_\delta}\right|
	=
	\binom{\left|\leaves_r\right|}{2}
	=
	\binom{n}{2}.
\end{equation}
Thus, the vectors $\beta_v$ for $v\in V$ can be computed with $\bigO(n^2)$ field operations. The algorithms we propose for converting between the monomial and LCH bases only require the precomputation and storage of $\beta_{v_\delta,0}$ or $1/\beta_{v_\delta,0}$ for all internal $v\in V$. As a full binary tree with $n$ leaves has $n-1$ internal vertices, the algorithms require the storage of $\bigO(n)$ field elements, which can be computed with $\bigO(n^2)$ field operations.

For conversions involving the Lagrange or Newton bases, we generalise the problems to include a shift parameter, as we did at the end of Section~\ref{sec:reduction} for Lagrange to LCH conversion. As a result, we require additional machinery to allow computations related to the shift parameter to be handled in a time and space-friendly manner. Define maps $\PushDown_v:\leaves_v\times\F\rightarrow\F$ for $v\in V$ recursively as follows: for $u\in\leaves_v$ and $\lambda\in\F$,
\begin{equation*}
	\PushDown_v(u,\lambda)
	=\begin{cases}
		\lambda/\beta_{v,0} 
		&\text{if $u=v$},\\
		\PushDown_{v_\alpha}(u,\lambda)
		&\text{if $u\neq v$ and $u\in\leaves_{v_\alpha}$},\\
		\PushDown_{v_\delta}
		\left(
			u,
			\left(\lambda/\beta_{v,0}\right)^{2^{\dg_v}}
			-\lambda/\beta_{v,0}
		\right)
		&\text{if $u\neq v$ and $u\in\leaves_{v_\delta}$}.
	\end{cases}
\end{equation*}
For internal $v\in V$, define
\begin{equation*}
	\sigma_{v,i}=\sum^i_{j=0}\beta_{v,\dg_v+j}
	\quad\text{for $i\in\{0,\dotsc,n_v-\dg_v-1\}$}.
\end{equation*}
Then the algorithms we propose for conversions involving the Newton or Lagrange bases require the precomputation and storage of $\PushDown_v(u,\sigma_{v,0}),\dotsc,\PushDown_v(u,\sigma_{v,n_v-\dg_v-1})$ for all internal $v\in V$ and $u\in\leaves_{v_\alpha}$. The identity~\eqref{eqn:sum-identity} implies that $n(n-1)/2$ field elements are stored as result of this requirement. Moreover, it is straightforward to show that the precomputation can be performed with $\bigO(n^3)$ field operations.

\begin{remark}\label{rmk:cantor-precomputations} If $\beta$ is a Cantor basis, then Proposition~\ref{prop:cantor-trees} and property~\eqref{cantor-delta} of Lemma~\ref{lem:cantor} imply that $\beta_v=(\beta_0,\dotsc,\beta_{n_v-1})$ for $v\in V$. It follows that $\beta_{v_\delta,0}=\beta_0=1$ for all internal $v\in V$. Thus, no precomputations are required for converting between the LCH and monomial bases. For $v\in V$ and $u\in\leaves_v$, the map $\PushDown_v(u,{}\cdot{}):\F\rightarrow\F$ is $\F_2$-linear, while Proposition~\ref{prop:cantor-trees} and properties~\eqref{cantor-subfields} and~\eqref{cantor-delta} of Lemma~\ref{lem:cantor} imply that
\begin{equation*}\label{eqn:pushdown-cantor}
	\PushDown_v\left(u,\beta_i\right)
	=\begin{cases}
		\beta_i & \text{if $u=v$},\\
		\PushDown_{v_\alpha}\left(u,\beta_i\right)
		& \text{if $u\neq v$ and $u\in\leaves_{v_\alpha}$},\\
		\PushDown_{v_\delta}
		\left(
			u,
			\beta_{i-\dg_v}
		\right)
		& \text{if $u\neq v$, $u\in\leaves_{v_\delta}$ and $i\geq\dg_v$},\\
		0 & \text{if $u\neq v$, $u\in\leaves_{v_\delta}$ and $i<\dg_v$},\\
	\end{cases}
\end{equation*}
for $i\in\{0,\dotsc,n-1\}$. Consequently, the precomputations for conversions involving the Newton or Lagrange bases require fewer operations.
\end{remark}

\subsection{Conversion between the Newton and Lin--Chung--Han bases}\label{sec:NX}

The factorisations of the Newton basis polynomials provided by Lemma~\ref{lem:reduction} include a shift of variables for one factor. Consequently, as we did for Lagrange basis to LCH basis conversion, we consider the more general problem of converting between a basis of shifted Newton polynomials and the LCH basis. Our conversions algorithms are then based on the following analogue of Lemma~\ref{lem:setup-LX-reduction}.

\begin{lemma}\label{lem:NX-reduction} Let $v\in V$ be an internal vertex and $\ell\in\{1,\dotsc,2^{n_v}\}$. Suppose that $f_0,\dotsc,f_{\ell-1},\lambda,\fint_0,\dotsc,\fint_{\ell-1},\fout_0,\dotsc,\fout_{\ell-1}\in\F$ satisfy
\begin{equation}\label{eqn:NX-rows}
	\sum^{\min\left(\ell-2^{\dg_v}i,2^{\dg_v}\right)-1}_{j=0}
	\fint_{2^{\dg_v}i+j}
	X_{\beta_{v_\alpha},j}
	=
	\sum^{\min\left(\ell-2^{\dg_v}i,2^{\dg_v}\right)-1}_{j=0}
	f_{2^{\dg_v}i+j}
	N_{\beta_{v_\alpha},j}\left(x-\lambda-\omega_{\gamma_v,i}\right)
\end{equation}
for $i\in\{0,\dotsc,\ceil{\ell/2^{\dg_v}}-1\}$, and
\begin{equation}\label{eqn:NX-cols}
	\sum^{\ceil{(\ell-j)/2^{\dg_v}}-1}_{i=0}
	\fout_{2^{\dg_v}i+j}
	X_{\beta_{v_\delta},i}
	=
	\sum^{\ceil{(\ell-j)/2^{\dg_v}}-1}_{i=0}
	\fint_{2^{\dg_v}i+j}
	N_{\beta_{v_\delta},i}\left(x-\eta\right)
\end{equation}
for $j\in\{0,\dotsc,\min(2^{\dg_v},\ell)-1\}$, where $\eta=(\lambda/\beta_{v,0})^{2^{\dg_v}}-\lambda/\beta_{v,0}$. Then
\begin{equation}\label{eqn:NX}
	\sum^{\ell-1}_{i=0}\fout_iX_{\beta_v,i}
	=
	\sum^{\ell-1}_{i=0}f_iN_{\beta_v,i}\left(x-\lambda\right).
\end{equation}
\end{lemma}

The proof of Lemma~\ref{lem:NX-reduction} is omitted since it follows along similar lines to that of Lemma~\ref{lem:setup-LX-reduction}. Using the lemma, we obtain Algorithms~\ref{alg:N2X} and~\ref{alg:X2N} for conversion between the Newton and LCH bases. The algorithms each operate on a vector $(a_0,\dotsc,a_{\ell-1})$ of field elements that initially contains the coefficients of a polynomial on the input basis, and overwrite its entries with the coefficients of the polynomial on the output basis. The subvectors of this vector that appear in the algorithms would be represented in practice by auxiliary variables, e.g., by a pointer to their first element and a stride parameter. The map $\Gray:\N\rightarrow\N$ that appears in the algorithms is defined by $i\mapsto\min\left\{k\in\N\mid [i]_k=0\right\}$. By noting that $\Delta(0),\Delta(1),\dotsc,\Delta(2^k-2)$ is the transition sequence of the $k$-bit binary reflected Gray code, it is possible to successively compute the terms of the sequence at the cost of a small constant number of operations per element by the algorithm of Bitner, Ehrlich and Reingold~\cite{bitner1976} (see also~\cite{knuth2005}).

\begin{algorithm}[h]
	\caption{$\mathsf{N2X}(v,(\PushDown_v(u,\lambda))_{u\in\leaves_v},\ell,(a_0,\ldots,a_{\ell-1}))$}\label{alg:N2X}
	\begin{algorithmic}[1]
		\Require a vertex $v\in V$, the vector $(\PushDown_v(u,\lambda))_{u\in\leaves_v}\in\F^{n_v}$ for some $\lambda\in\F$, $\ell\in\{1,\dotsc,2^{n_v}\}$, and $a_i=f_i\in\F$ for $i\in\{0,\dotsc,\ell-1\}$.
		\Ensure $a_i=\fout_i\in\F$ for $i\in\{0,\dotsc,\ell-1\}$ such that~\eqref{eqn:NX} holds.
		\If{$v$ is a leaf}\label{N2X:base}
			\If{$\ell=2$}
				 $a_0\set a_0+\PushDown_v(v,\lambda)a_1$
			\EndIf
			\State\Return
		\EndIf
		\State\label{N2X:aux}%
			$\ell_1\set\ceil{\ell/2^{\dg_v}}-1$,
			$\ell_2\set\ell-2^{\dg_v}\ell_1$,
			$\ell_2'\set\min(2^{\dg_v},\ell)$
		\State\label{N2X:mu-nu}%
			$\mu\set(\PushDown_v(u,\lambda))_{u\in\leaves_{v_\alpha}}$,
			$\nu\set(\PushDown_v(u,\lambda))_{u\in\leaves_{v_\delta}}$
		\For{$i=0,\dotsc,\ell_1-1$}\label{N2X:rows}
			\State\label{N2X:row}\Call{N2X}{$v_\alpha,\mu,2^{\dg_v},(a_{2^{\dg_v}i},a_{2^{\dg_v}i+1},\dotsc,a_{2^{\dg_v}(i+1)-1})$}
			\State\label{N2X:update-shifts}$\mu\set\mu+(\PushDown_v(u,\sigma_{v,\Gray(i)}))_{u\in\leaves_{v_\alpha}}$
		\EndFor
		\State\label{N2X:last-row}\Call{N2X}{$v_\alpha,\mu,\ell_2,(a_{2^{\dg_v}\ell_1},a_{2^{\dg_v}\ell_1+1},\dotsc,a_{\ell-1})$}
		\For{$j=0,\dotsc,\ell_2-1$}\label{N2X:cols}
			\State\label{N2X:col-i}\Call{N2X}{$v_\delta,\nu,\ell_1+1,(a_j,a_{2^{\dg_v}+j},\dotsc,a_{2^{\dg_v}\ell_1+j})$}
		\EndFor
		\For{$j=\ell_2,\dotsc,\ell_2'-1$}
			\State\label{N2X:col-ii}\Call{N2X}{$v_\delta,\nu,\ell_1,(a_j,a_{2^{\dg_v}+j},\dotsc,a_{2^{\dg_v}(\ell_1-1)+j})$}
		\EndFor\label{N2X:cols-end}
	\end{algorithmic}
\end{algorithm}

\begin{algorithm}[h]
	\caption{$\mathsf{X2N}(v,(\PushDown_v(u,\lambda))_{u\in\leaves_v},\ell,(a_0,\ldots,a_{\ell-1}))$}\label{alg:X2N}
	\begin{algorithmic}[1]
		\Require a vertex $v\in V$, the vector $(\PushDown_v(u,\lambda))_{u\in\leaves_v}\in\F^{n_v}$ for some $\lambda\in\F$, $\ell\in\{1,\dotsc,2^{n_v}\}$, and $a_i=\fout_i\in\F$ for $i\in\{0,\dotsc,\ell-1\}$.
		\Ensure $a_i=f_i\in\F$ for $i\in\{0,\dotsc,\ell-1\}$ such that~\eqref{eqn:NX} holds.
		\If{$v$ is a leaf}\label{X2N:base}
			\If{$\ell=2$}
				$a_0\set a_0+\PushDown_v(v,\lambda)a_1$
			\EndIf
			\State\Return
		\EndIf
		\State
			\label{X2N:aux}%
			$\ell_1\set\ceil{\ell/2^{\dg_v}}-1$, 
			$\ell_2\set\ell-2^{\dg_v}\ell_1$,
			$\ell_2'\set\min(2^{\dg_v},\ell)$
		\State
			\label{X2N:mu-nu}%
			$\mu\set(\PushDown_v(u,\lambda))_{u\in\leaves_{v_\alpha}}$, 
			$\nu\set(\PushDown_v(u,\lambda))_{u\in\leaves_{v_\delta}}$
		\For{$j=0,\dotsc,\ell_2-1$}\label{X2N:cols}
			\State\label{X2N:col-i}\Call{X2N}{$v_\delta,\nu,\ell_1+1,(a_j,a_{2^{\dg_v}+j},\dotsc,a_{2^{\dg_v}\ell_1+j})$}
		\EndFor
		\For{$j=\ell_2,\dotsc,\ell_2'-1$}
			\State\label{X2N:col-ii}\Call{X2N}{$v_\delta,\nu,\ell_1,(a_j,a_{2^{\dg_v}+j},\dotsc,a_{2^{\dg_v}(\ell_1-1)+j})$}
		\EndFor\label{X2N:cols-end}
		\For{$i=0,\dotsc,\ell_1-1$}\label{X2N:rows}
			\State\label{X2N:row}\Call{X2N}{$v_\alpha,\mu,2^{\dg_v},(a_{2^{\dg_v}i},a_{2^{\dg_v}i+1},\dotsc,a_{2^{\dg_v}(i+1)-1})$}
			\State\label{X2N:update-shifts} $\mu\set\mu+(\PushDown_v(u,\sigma_{v,\Gray(i)}))_{u\in\leaves_{v_\alpha}}$
		\EndFor
		\State\label{X2N:last-row}\Call{X2N}{$v_\alpha,\mu,\ell_2,(a_{2^{\dg_v}\ell_1},a_{2^{\dg_v}\ell_1+1},\dotsc,a_{\ell-1})$}
	\end{algorithmic}
\end{algorithm}

\begin{theorem}\label{thm:NX-correctness} Algorithms~\ref{alg:N2X} and~\ref{alg:X2N} are correct.
\end{theorem}
\begin{proof} We only prove correctness for Algorithm~\ref{alg:N2X}, since the proof for Algorithm~\ref{alg:X2N} is almost identical. Suppose that the input vertex $v$ is a leaf. Then $\ell\in\{1,2\}$ since $n_v=1$. Moreover, $\leaves_v=\{v\}$, $\PushDown_v(v,\lambda)=\lambda/\beta_{v,0}$, $X_{\beta_v,0}=N_{\beta_v,0}=1$ and $X_{\beta_v,1}=N_{\beta_v,1}=x/\beta_{v,0}$. It follows that Algorithm~\ref{alg:N2X} produces the correct output whenever the input vertex is a leaf. Therefore, as $(V,E)$ is a full binary tree, it is sufficient to show that for internal $v\in V$, if the algorithm produces the correct output whenever $v_\alpha$ or $v_\delta$ is given as an input, then the algorithm produces the correct output whenever $v$ is given as an input.
	
Let $v\in V$ be an internal vertex and suppose that the algorithm produces the correct output whenever $v_\alpha$ or $v_\delta$ is given as an input. Let $\lambda\in\F$, $\ell\in\{1,\dotsc,2^{n_v}\}$ and $f_0,\dotsc,f_{\ell-1}\in\F$. Suppose that Algorithm~\ref{alg:N2X} is called on $v$, $(\PushDown_v(u,\lambda))_{u\in\leaves_v}$ and~$\ell$, with $a_i=f_i$ for $i\in\{0,\dotsc,\ell-1\}$. Then Lines~\ref{N2X:mu-nu} and~\ref{N2X:update-shifts} of the algorithm and the $\F_2$-linearity of the maps $\PushDown_v(u,{}\cdot{}):\F\rightarrow\F$, for $u\in\leaves_v$, ensure that
\begin{equation*}
	\mu
	=\left(
	\PushDown_v(u,\lambda)
		+\sum^{i-1}_{j=0}
		\PushDown_v\left(u,\sigma_{v,\Gray(j)}\right)
	\right)_{u\in\leaves_{v_\alpha}}
	=\left(
		\PushDown_v
		\left(
			u,
			\lambda
			+\sum^{i-1}_{j=0}
			\sigma_{v,\Gray(j)}
		\right)
	\right)_{u\in\leaves_{v_\alpha}}.
\end{equation*}
each time the recursive call of Line~\ref{N2X:row} is performed. As $\gamma_v=(\beta_{v,\dg_v},\dotsc,\beta_{v,n_v-\dg_v-1})$, we have $\omega_{\gamma_v,0}=0$ and
\begin{equation*}
	\omega_{\gamma_v,i}
	=\omega_{\gamma_v,i-1}
	+
	\sum^{\Gray(i-1)}_{j=0}\beta_{v,\dg_v+j}
	=\omega_{\gamma_v,i-1}
	+
	\sigma_{v,\Gray(i-1)}
	=
	\sum^{i-1}_{j=0}\sigma_{v,\Gray(j)}
\end{equation*}
for $i\in\{1,\dotsc,2^{n_v-\dg_v}-1\}$. Thus,
\begin{equation*}
	\mu
	=\left(
		\PushDown_v
		\left(u,\lambda+\omega_{\gamma_v,i}\right)
	\right)_{u\in\leaves_{v_\alpha}}\\
	=\left(
	\PushDown_{v_\alpha}
		\left(u,\lambda+\omega_{\gamma_v,i}\right)
	\right)_{u\in\leaves_{v_\alpha}}
\end{equation*}
each time the recursive call of Line~\ref{N2X:row} is performed. Similarly,
\begin{equation*}
	\mu
	=
	\left(
		\PushDown_{v_\alpha}
		\left(u,\lambda+\omega_{\gamma_v,\ceil{\ell/2^{\dg_v}}-1}\right)
	\right)_{u\in\leaves_{v_\alpha}}
\end{equation*}
when the recursive call of Line~\ref{N2X:last-row} is performed. Therefore, the assumption that Algorithm~\ref{alg:N2X} produces the correct output whenever $v_\alpha$ is given as an input implies that Lines~\ref{N2X:aux}--\ref{N2X:last-row} set $a_i=\fint_i$ for $i\in\{0,\dotsc,\ell-1\}$, where $\fint_0,\dotsc,\fint_{\ell-1}$ are the unique elements in $\F$ such that \eqref{eqn:NX-rows} holds.

The recursive calls of Lines~\ref{N2X:col-i} and~\ref{N2X:col-ii} are made with
\begin{equation*}
	\nu
	=\left(\PushDown_v\left(u,\lambda\right)\right)_{u\in\leaves_{v_\delta}}
	=\left(\PushDown_{v_\delta}\left(u,\eta\right)\right)_{u\in\leaves_{v_\delta}},
\end{equation*}
where $\eta=(\lambda/\beta_{v,0})^{2^{\dg_v}}-\lambda/\beta_{v,0}$. Therefore, the assumption that Algorithm~\ref{alg:N2X} produces the correct output whenever $v_\delta$ is given as an input implies that Lines~\ref{N2X:cols}--\ref{N2X:cols-end} set $a_i=\fout_i$ for $i\in\{0,\dotsc,\ell-1\}$, where $\fout_0,\dotsc,\fout_{\ell-1}$ are the unique elements in $\F$ such that \eqref{eqn:NX-cols} holds. Hence, Lemma~\ref{lem:NX-reduction} implies that~\eqref{eqn:NX} holds at the end of the algorithm.
\end{proof}

For conversions between the Newton and LCH bases the shift parameter $\lambda$ is equal to zero for the initial calls to Algorithms~\ref{alg:N2X} and~\ref{alg:X2N}. In this case, the vector $(\PushDown_v(u,\lambda))_{u\in\leaves_v}$ contains all zeros. If an application arises where it is necessary for the initial call to be made with $\lambda\neq 0$, then the input vector can be computed with $\bigO(n^2_v)$ field operations. Storing the input vector and the vectors $\mu$ and $\nu$ that appear in the algorithms requires storing $\bigO(n^2_v)$ field elements. Including the computation and storage of the precomputed elements $\PushDown_v(u,\sigma_{v,i})$, it follows that Algorithms~\ref{alg:N2X} and~\ref{alg:X2N} require auxiliary storage for $\bigO(n^2)$ field elements, while all precomputations can be performed with $\bigO(n^3)$ field operations.


The number of multiplications performed by Algorithms~\ref{alg:N2X} and~\ref{alg:X2N} is independent of the choice of reduction tree. This is not true of the number of additions performed by the algorithms due to the updates made to the vector $\mu$ in the recursive case. Line~\ref{N2X:update-shifts} of Algorithm~\ref{alg:N2X} and Line~\ref{X2N:update-shifts} of Algorithm~\ref{alg:X2N} each perform $(\ceil{\ell/2^{\dg_v}}-1)\dg_v$ additions over all iterations of their containing loops. Consequently, we should aim to avoid small values of $\dg_v$ when choosing a reduction tree. Moreover, we expect the number of additions performed by the algorithms to be maximised when the subtree rooted on the initial input vertex satisfies $\image(\dg)\subseteq\{0,1\}$.

\begin{example}\label{ex:NX} If $\beta$ is a Cantor basis, then Proposition~\ref{prop:cantor-trees} implies that $\dg_v\leq 2^{\ceil{\log_2 n_v}-1}$ for all internal $v\in V$, and guarantees the existence of a reduction tree for $\beta$ such that equality always holds. For $n$ up to fifteen, we confirmed experimentally that this choice of reduction tree always minimises the number of additions performed by Algorithms~\ref{alg:N2X} and~\ref{alg:X2N} over all possible reduction trees, while those with $\image(\dg)\subseteq\{0,1\}$ were found to always maximise the number of additions performed by the algorithms. Figure~\ref{fig:cantor_nx} shows the maximum and minimum number of additions performed over all reduction trees for $n=15$, as well as the number of multiplications performed for all trees.
\begin{figure}
	\includegraphics{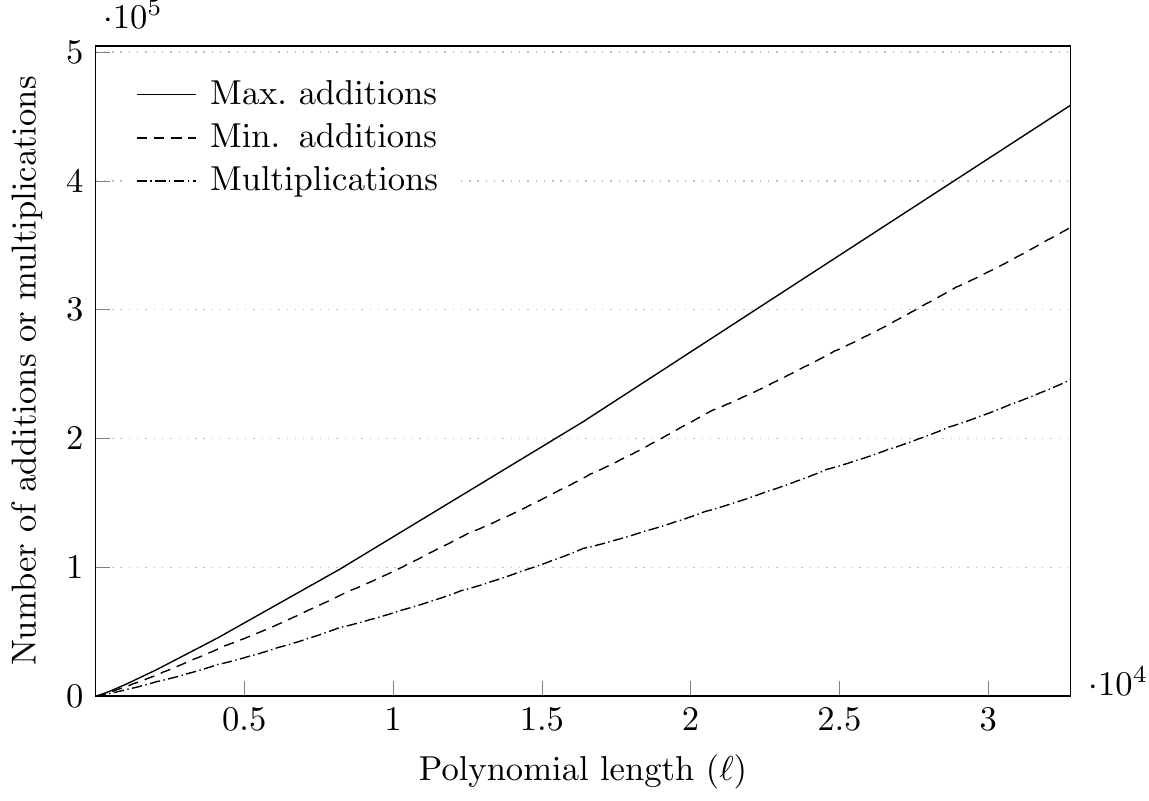}
	\caption{Maximum and minimum number of operations performed by Algorithms~\ref{alg:N2X} and~\ref{alg:X2N} for Cantor bases of dimension $15$.}
	\label{fig:cantor_nx}
\end{figure}
\end{example}

The difference between the maximum and minimum number of additions shown in Figure~\ref{fig:cantor_nx} is reflected in the bounds that we obtain on the complexities of the algorithms.

\begin{theorem}\label{thm:NX-complexity} Algorithms~\ref{alg:N2X} and~\ref{alg:X2N} perform at most $\floor{\ell/2}\ceil{\log_2\ell}$ multiplications and $\ell\left(\ceil{\log_2\ell}-1\right)+1$ additions in~$\F$. If $\dg_v=2^{\ceil{\log_2 n_v}-1}$ for all internal $v\in V$, then the algorithms perform at most $(3\ell-2)\ceil{\log_2\ell}/4$ additions in $\F$.
\end{theorem}

We split the proof of Theorem~\ref{thm:NX-complexity} into three lemmas, one for each of the three bounds. It is clear that Algorithms~\ref{alg:N2X} and~\ref{alg:X2N} perform the same number of multiplications when given identical inputs. Consequently, we only prove the bounds for Algorithm~\ref{alg:N2X}. The three bounds are equal to zero if $\ell=1$, and one if $\ell=2$. Thus, the bounds hold if the input vertex is a leaf. Therefore, for each of the three bounds it is sufficient to show that if $v\in V$ is an internal vertex such that the bound holds whenever the input vertex is $v_\alpha$ or $v_\delta$, then the bound holds whenever $v$ is the input vertex.

\begin{lemma} Algorithms~\ref{alg:N2X} and~\ref{alg:X2N} perform at most $\floor{\ell/2}\ceil{\log_2\ell}$ multiplications~in~$\F$.
\end{lemma}
\begin{proof} Suppose that the input vertex $v\in V$ to Algorithm~\ref{alg:N2X} is an internal vertex such that the algorithm performs at most $\floor{\ell/2}\ceil{\log_2\ell}$ multiplications in $\F$ whenever $v_\alpha$ or $v_\delta$ is given as the input vertex. If $\ell_1=0$, then $\ell_2=\ell_2'=\ell$ and the algorithm performs at most
\begin{equation*}
	\floor{\frac{\ell_2}{2}}
	\ceil{\log_2\ell_2}
	+
	\ell_2
	\floor{\frac{1}{2}}
	\ceil{\log_21}
	=
	\floor{\frac{\ell}{2}}
	\ceil{\log_2\ell}
\end{equation*}
multiplications. Therefore, suppose that $\ell_1\neq 0$. Then, as $\ell_2\leq 2^{\dg_v}$, Lines~\ref{N2X:rows}--\ref{N2X:last-row} of the algorithm perform at most
\begin{equation}\label{eqn:NX-l_1-0-mults}
	2^{\dg_v-1}\ell_1\dg_v
	+
	\floor{\frac{\ell_2}{2}}
	\ceil{\log_2\ell_2}
	\leq
	\left(
		2^{\dg_v-1}\ell_1
		+
		\floor{\frac{\ell_2}{2}}
	\right)
	\dg_v
	=
	\floor{\frac{\ell}{2}}\dg_v
\end{equation}
multiplications. Let $k\in\N$ such that $2^{k-1}<\ell_1+1\leq 2^k$. Then $\ceil{\log_2\ell_1+1}=k$ and $2^{\dg_v+k-1}<2^{\dg_v}(\ell_1+\ell_2/2^{\dg_v})=\ell\leq 2^{\dg_v+k}$, since $k\geq 1$ and $0<\ell_2/2^{\dg_v}\leq 1$. Thus, $\ceil{\log_2\ell_1+1}=\ceil{\log_2\ell}-\dg_v$. It follows that Lines~\ref{N2X:cols}--\ref{N2X:cols-end} of the algorithm perform at most
\begin{multline*}
	\ell_2
	\floor{\frac{\ell_1+1}{2}}
	\ceil{\log_2\ell_1+1}
	+
	\left(2^{\dg_v}-\ell_2\right)
	\floor{\frac{\ell_1}{2}}
	\ceil{\log_2\ell_1}\\
	\leq
	\floor{\frac{\ell_2(\ell_1+1)+\left(2^{\dg_v}-\ell_2\right)\ell_1}{2}}
	\ceil{\log_2\ell_1+1}
	=
	\floor{\frac{\ell}{2}}
	\left(\ceil{\log_2\ell}-\dg_v\right)
\end{multline*}
multiplications. By combining this bound with~\eqref{eqn:NX-l_1-0-mults}, it follows that  Algorithm~\ref{alg:N2X} performs at most $\floor{\ell/2}\ceil{\log_2\ell}$ multiplications if $\ell_1\neq 0$.
\end{proof}

\begin{lemma} Algorithms~\ref{alg:N2X} and~\ref{alg:X2N} perform at most $\ell\left(\ceil{\log_2\ell}-1\right)+1$ additions in~$\F$.
\end{lemma}
\begin{proof} Suppose that the input vertex $v\in V$ to Algorithm~\ref{alg:N2X} is an internal vertex such that the algorithm performs at most $\ell\left(\ceil{\log_2\ell}-1\right)+1$ additions in $\F$ whenever $v_\alpha$ or $v_\delta$ is given as the input vertex. If $\ell_1=0$, then $\ell_2=\ell_2'=\ell$ and the algorithm performs at most
\begin{equation*}
	\ell_2(\ceil{\log_2\ell_2}-1)+1
	+
	\ell_2
	\left(1(\ceil{\log_21}-1)+1\right)
	=
	\ell(\ceil{\log_2\ell}-1)+1
\end{equation*}	
additions. Therefore, suppose that $\ell_1\neq 0$. Then Lines~\ref{N2X:rows}--\ref{N2X:last-row} of the algorithm perform at most
\begin{multline*}
	\ell_1
	\left(
		2^{\dg_v}\left(\dg_v-1\right)
		+1
		+\dg_v
	\right)
	+\ell_2
	\left(
		\ceil{\log_2\ell_2}
		-1
	\right)
	+1\\
	=
	\ell
	\left(\dg_v-1\right)+1
	+
	\ell_1(\dg_v+1)
	+
	\ell_2
	\left(\ceil{\log_2\ell_2}-\dg_v\right)
\end{multline*}
additions, while Lines~\ref{N2X:cols}--\ref{N2X:cols-end} perform at most
\begin{multline*}
	\ell_2
	\left(\ell_1+1\right)
	(\ceil{\log_2\ell_1+1}-1)
	+
	\left(2^{\dg_v}-\ell_2\right)
	\ell_1
	\left(\ceil{\log_2\ell_1}-1\right)
	+2^{\dg_v}\\
	\leq
	\ell\left(\ceil{\log_2\ell_1+1}-1\right)+2^{\dg_v}
	=
	\ell
	\left(
		\ceil{\log_2\ell}
		-\dg_v
	\right)
	-\ell
	+2^{\dg_v}
\end{multline*}
additions. As $1\leq\ell_2\leq2^{\dg_v}$, it follows that Algorithm~\ref{alg:N2X} performs at most
\begin{multline*}
	\ell\left(\ceil{\log_2\ell}-1\right)+1
	+\ell_1\left(\dg_v+1\right)
	+\ell_2\left(\ceil{\log_2\ell_2}-\dg_v\right)
	-\ell+2^{\dg_v}\\
	\begin{aligned}
		&=
		\ell\left(\ceil{\log_2\ell}-1\right)+1
		-\left(2^{\dg_v}-\dg_v-1\right)
		\left(\ell_1-1\right)
		+\dg_v+1
		+\ell_2
		\ceil{\log_2\frac{\ell_2}{2^{\dg_v+1}}}\\
		&\leq
		\ell\left(\ceil{\log_2\ell}-1\right)+1
		+\dg_v+1
		-\frac{\floor{\log_22^{\dg_v+1}/\ell_2}}{2^{\dg_v+1}/\ell_2}
		2^{\dg_v+1}\\
		&\leq
		\ell\left(\ceil{\log_2\ell}-1\right)+1
	\end{aligned}
\end{multline*}
additions if $\ell_1\neq 0$.
\end{proof}

\begin{lemma} Suppose that $\dg_v=2^{\ceil{\log_2 n_v}-1}$ for all internal $v\in V$. Then Algorithms~\ref{alg:N2X} and~\ref{alg:X2N} perform at most $(3\ell-2)\ceil{\log_2\ell}/4$ additions in $\F$.
\end{lemma}
\begin{proof} Suppose that $\dg_v=2^{\ceil{\log_2 n_v}-1}$ for all internal $v\in V$. Furthermore, suppose that the input vertex $v\in V$ to Algorithm~\ref{alg:N2X} is an internal vertex such that the algorithm performs at most $(3\ell-2)\ceil{\log_2\ell}/4$ additions in $\F$ whenever $v_\alpha$ or $v_\delta$ is given as the input vertex. If $\ell_1=0$, then $\ell_2=\ell_2'=\ell$ and the algorithm performs at most
\begin{equation*}
	\frac{3\ell_2-2}{4}\ceil{\log_2\ell_2}
	+\ell_2
	\left(
		\frac{3-2}{4}\ceil{\log_21}
	\right)
	=
	\frac{3\ell-2}{4}\ceil{\log_2\ell}.
\end{equation*}
additions. Therefore, suppose that $\ell_1\neq 0$. Then, as $\ell_2\leq 2^{\dg_v}$, Lines~\ref{N2X:rows}--\ref{N2X:last-row} of the algorithm perform at most
\begin{equation*}
	\ell_1\left(\frac{3\cdot 2^{\dg_v}-2}{4}\dg_v+\dg_v\right)
	+\frac{3\ell_2-2}{4}\ceil{\log_2\ell_2}
	\leq
	\frac{3\ell-2}{4}\dg_v
	+\frac{\ell_1}{2}\dg_v
\end{equation*}
additions, while Lines~\ref{N2X:cols}--\ref{N2X:cols-end} perform at most
\begin{multline*}
	\ell_2
	\frac{3(\ell_1+1)-2}{4}
	\ceil{\log_2\ell_1+1}
	+
	\left(2^{\dg_v}-\ell_2\right)
	\frac{3\ell_1-2}{4}
	\ceil{\log_2\ell_1}\\
	\begin{aligned}
	&\leq
	\frac{3\ell-2^{\dg_v+1}}{4}
	\ceil{\log_2\ell_1+1}\\
	&\leq
	\frac{3\ell-2}{4}
	\left(\ceil{\log_2\ell}-\dg_v\right)
	-
	\frac{2^{\dg_v}-1}{2}
	\log_2\left(\ell_1+1\right)
	\end{aligned}
\end{multline*}
additions. It follows that Algorithm~\ref{alg:N2X} performs at most
\begin{equation*}
	\frac{3\ell-2}{4}
	\ceil{\log_2\ell}
	+
	\frac{\dg_v}{2}
	\log_2\left(\ell_1+1\right)
	\left(
		\frac{\ell_1}{\log_2\left(\ell_1+1\right)}
		-\frac{2^{\dg_v}-1}{\dg_v}
	\right)
	\leq
	\frac{3\ell-2}{4}
	\ceil{\log_2\ell}
\end{equation*}
additions if $\ell_1\neq 0$, since $\ell_1+1\leq 2^{n_v-\dg_v}\leq 2^{\dg_v}$ and the function $x/\log_2(x+1)$ is increasing for $x\geq 1$.
\end{proof}

\subsection{Conversion from the Lagrange basis to the Lin--Chung--Han basis}\label{sec:L2X}

Our algorithms for converting between the Lagrange and LCH bases are direct analogues of Harvey's cache-friendly truncated FFT and inverse truncated FFT algorithms~\cite{harvey2009}. We align our presentation of the algorithms and their proofs of correctness with their counterparts given by Harvey, and direct the reader to Harvey's paper for further motivation behind the algorithms. In this section, we focus on the problem of converting from the Lagrange basis to the LCH basis. The inverse problem of converting from the LCH basis to the Lagrange basis is considered separately in the next section.

We propose Algorithm~\ref{alg:L2X} for converting from the Lagrange basis to the LCH basis. Like the algorithms of Section~\ref{sec:NX}, Algorithm~\ref{alg:L2X} operates on a vector of field elements whose initial entries are overwritten with the output. However, the length of the vector is determined by the input vertex rather than the parameter $\ell$ which bounds the polynomial length. Consequently, the vector may have coordinates that are never used to store input or output values, but which are still used for intermediate computations. If the parameter $c$ is less than $\ell$, then the algorithm differs substantially from the other basis conversion algorithms in this paper by requiring that the vector initially contain a combination of coefficients from  a polynomial's representations on the input and output bases. If the parameter $\flag$ is set equal to one, then the algorithm has the additional unique feature of being required to compute a coefficient from the input basis representation. These two parameters are part of the internal mechanics of the recursive approach used by the algorithm, and in most practical applications, such as the multiplication of binary polynomials, the initial call to the algorithm is made with $c=\ell$ and $\flag=0$. However, one may be required to initially call the algorithm with $c<\ell$ if it used within the Hermite interpolation algorithm of Coxon~\cite{coxon2018}.

The reduction used by Algorithm~\ref{alg:L2X} is provided by the following generalisation of Lemma~\ref{lem:setup-LX-reduction}, which is rephrased to reflect the fact that the algorithm takes in a mixture of coefficients from representations on the Lagrange and LCH bases.

\begin{lemma}\label{lem:LX-reduction} Let $v\in V$ be an internal vertex and $\ell\in\{1,\dotsc,2^{n_v}\}$. Suppose that $f_0,\dotsc,f_{2^{n_v}-1},\lambda,\fout_0,\dotsc,\fout_{\ell-1}\in\F$ satisfy
\begin{equation}\label{eqn:LX}
	\sum^{\ell-1}_{i=0}
	\fout_i
	X_{\beta_v,i}
	=
	\sum^{2^{n_v}-1}_{i=0}
	f_i
	L_{\beta_v,i}\left(x-\lambda\right).
\end{equation}
Then there exist unique elements $\fint_0,\dotsc,\fint_{2^{n_v}-1}\in\F$ such that
\begin{equation}\label{eqn:LX-rows} 
	\sum^{\min\left(2^{\dg_v},\ell\right)-1}_{j=0}
	\fint_{2^{\dg_v}i+j}
	X_{\beta_{v_\alpha},j}
	=
	\sum^{2^{\dg_v}-1}_{j=0}
	f_{2^{\dg_v}i+j}
	L_{\beta_{v_\alpha},j}
	\left(x-\lambda-\omega_{\gamma_v,i}\right)
\end{equation}
for $i\in\{0,\dotsc,2^{n_v-\dg_v}-1\}$, and
\begin{equation}\label{eqn:LX-cols}
	\sum^{\floor{(\ell-1-j)/2^{\dg_v}}}_{i=0}
	\fout_{2^{\dg_v}i+j}
	X_{\beta_{v_\delta},i}
	=
	\sum^{2^{n_v-\dg_v}-1}_{i=0}
	\fint_{2^{\dg_v}i+j}
	L_{\beta_{v_\delta},i}\left(x-\eta\right)
\end{equation}
for $j\in\{0,\dotsc,2^{\dg_v}-1\}$, where $\eta=(\lambda/\beta_{v,0})^{2^{\dg_v}}-\lambda/\beta_{v,0}$. 
\end{lemma}
\begin{proof} Let $v\in V$, $\ell\in\{1,\dotsc,2^{n_v}\}$ and $f_0,\dotsc,f_{2^{n_v}-1},\lambda,\fout_0,\dotsc,\fout_{\ell-1}\in\F$ satisfy the conditions of the lemma. Then there exist unique elements $\fint_0,\dotsc,\fint_{2^{n_v}-1}\in\F$ such that~\eqref{eqn:LX-cols} holds for $j\in\{0,\dotsc,2^{\dg_v}-1\}$. We show that they also satisfy~\eqref{eqn:LX-rows} for $i\in\{0,\dotsc,2^{n_v-\dg_v}-1\}$. If $j\in\{0,\dotsc,2^{\dg_v}-1\}$ satisfies $j\geq\ell$, then \eqref{eqn:LX-cols} implies that $\fint_{2^{\dg_v}i+j}=0$ for $i\in\{0,\dotsc,2^{n_v-\dg_v}-1\}$. Therefore, if we let $\fout_\ell=\dotsb=\fout_{2^{n_v}-1}=0$, then the left-hand sides of \eqref{eqn:LX}, \eqref{eqn:LX-rows} and \eqref{eqn:LX-cols} remain unchanged if we replace $\ell$ by $2^{n_v}$ in the summation bounds. Consequently, \eqref{eqn:LX-rows} must hold for $i\in\{0,\dotsc,2^{n_v-\dg_v}-1\}$, since otherwise Lemma~\ref{lem:setup-LX-reduction} allows us to contradict the uniqueness of the coefficients $f_0,\dotsc,f_{2^{n_v}-1}$ in~\eqref{eqn:LX} by writing the polynomial on the left-hand side of~\eqref{eqn:LX-rows} on the basis
\begin{equation*}
	\left\{
		L_{\beta_{v_\alpha},0}
		\left(x-\lambda-\omega_{\gamma_v,i}\right),
		\dotsc,
		L_{\beta_{v_\alpha},2^{\dg_v}-1}
		\left(x-\lambda-\omega_{\gamma_v,i}\right)
	\right\}
\end{equation*}
for $i\in\{0,\dotsc,2^{n_v-\dg_v}-1\}$. 
\end{proof}

\begin{algorithm}[h]
	\caption{$\mathsf{L2X}(v,(\PushDown_v(u,\lambda))_{u\in\leaves_v},c,\ell,\flag,(a_0,\ldots,a_{2^{n_v}-1}))$}\label{alg:L2X}
	\begin{algorithmic}[1]
		\Require a vertex $v\in V$, the vector $(\PushDown_v(u,\lambda))_{u\in\leaves_v}\in\F^{n_v}$ for some $\lambda\in\F$, $c,\ell\in\N$ such that $c\leq\ell$ and $1\leq \ell\leq 2^{n_v}$, $\flag\in\{0,1\}$ such that $1\leq\flag+c\leq 2^{n_v}$, $a_i=f_i\in\F$ for $i\in\{0,\dotsc,c-1\}$, and $a_i=\fout_i\in\F$ for $i\in\{c,\dotsc,\ell-1\}$.
		\Ensure $a_i=\fout_i\in\F$ for $i\in\{0,\dotsc,c-1\}$ such that~\eqref{eqn:LX} holds for some $f_c,\dotsc,f_{2^{n_v}-1}\in\F$, and $a_c=f_c$ if $\flag=1$.
		\If{$v$ is a leaf}\label{L2X:base}
			\If{$c=2$}
				$a_1\set a_0+a_1$, $a_0\set a_0+\PushDown_v(v,\lambda)a_1$
			\EndIf
 			\If{$c=1$, $\ell=2$ and $\flag=1$}
				$w\set\PushDown_v(v,\lambda)a_1$, $a_1\set a_0+a_1$, $a_0\set a_0+w$
			\EndIf
			\If{$c=1$, $\ell=2$ and $\flag=0$}
				$a_0\set a_0+\PushDown_v(v,\lambda)a_1$
			\EndIf
			\If{$c=0$ and $\ell=2$}
				$a_0\set a_0+\PushDown_v(v,\lambda)a_1$
			\EndIf
			\If{$c=1$, $\ell=1$ and $\flag=1$}
				$a_1\set a_0$
			\EndIf
		\State\Return
		\EndIf\label{L2X:base-end}
		\State
			\label{L2X:aux}%
			$c_1\set\floor{c/2^{\dg_v}}$,
			$c_2\set c-2^{\dg_v}c_1$
		\State
			$\ell_1\set\floor{\ell/2^{\dg_v}}$,
			$\ell_2\set \ell-2^{\dg_v}\ell_1$
		\State
			$\ell_2'\set\min(2^{\dg_v},\ell)$,
			$\flag'\set\min(\flag+c_2,1)$
		\State
			\label{L2X:aux-end}%
			$s\set\min(c_2,\ell_2)$,
			$t\set\max(c_2,\ell_2)$
		\State
			$\mu\set(\PushDown_v(u,\lambda))_{u\in\leaves_{v_\alpha}}$,
			$\nu\set(\PushDown_v(u,\lambda))_{u\in\leaves_{v_\delta}}$
		\For{$i=0,\dotsc,c_1+\flag'-2$}\label{L2X:rows}
			\State
				\label{L2X:row}%
				\Call{L2X}{$v_\alpha,\mu,2^{\dg_v},2^{\dg_v},0,(a_{2^{\dg_v}i},a_{2^{\dg_v}i+1},\dotsc,a_{2^{\dg_v}(i+1)-1})$}
			\State
				\label{L2X:update-shifts}%
				$\mu\set\mu+(\PushDown_v(u,\sigma_{v,\Gray(i)}))_{u\in\leaves_{v_\alpha}}$
		\EndFor
		\If{$\flag'=0$}
			\State\label{L2X:sometimes-last-row}%
			\Call{L2X}{$v_\alpha,\mu,2^{\dg_v},2^{\dg_v},0,(a_{2^{\dg_v}(c_1-1)},a_{2^{\dg_v}(c_1-1)+1},\dotsc,a_{2^{\dg_v}c_1-1})$}
		\EndIf\label{L2X:rows-end}
		\For{$j=c_2,\dotsc,t-1$}\label{L2X:right-cols-1}
			\State
				\label{L2X:right-col-1}%
				\Call{L2X}{$v_\delta,\nu,c_1,\ell_1+1,\flag',(a_j,a_{2^{\dg_v}+j},\dotsc,a_{2^{\dg_v}(2^{n_v-\dg_v}-1)+j})$}
		\EndFor\label{L2X:right-cols-1-end}
		\For{$j=t,\dotsc,\ell_2'-1$}\label{L2X:right-cols-2}
			\State
				\label{L2X:right-col-2}%
				\Call{L2X}{$v_\delta,\nu,c_1,\ell_1,\flag',(a_j,a_{2^{\dg_v}+j},\dotsc,a_{2^{\dg_v}(2^{n_v-\dg_v}-1)+j})$}
		\EndFor\label{L2X:right-cols-2-end}
		\If{$\flag'=1$}\label{L2X:last-row}
			\State
				\label{L2X:bottom-row}%
				\Call{L2X}{$v_\alpha,\mu,c_2,\ell_2',\flag,(a_{2^{\dg_v}c_1},a_{2^{\dg_v}c_1+1},\dotsc,a_{2^{\dg_v}(c_1+1)-1})$}
		\EndIf\label{L2X:last-row-end}
		\For{$j=0,\dotsc,s-1$}\label{L2X:left-cols-1}
			\State
				\label{L2X:left-col-1}%
				\Call{L2X}{$v_\delta,\nu,c_1+1,\ell_1+1,0,(a_j,a_{2^{\dg_v}+j},\dotsc,a_{2^{\dg_v}(2^{n_v-\dg_v}-1)+j})$}
		\EndFor\label{L2X:left-cols-1-end}
		\For{$j=s,\dotsc,c_2-1$}\label{L2X:left-cols-2}
			\State
				\label{L2X:left-col-2}%
				\Call{L2X}{$v_\delta,\nu,c_1+1,\ell_1,0,(a_j,a_{2^{\dg_v}+j},\dotsc,a_{2^{\dg_v}(2^{n_v-\dg_v}-1)+j})$}
		\EndFor\label{L2X:left-cols-2-end}
	\end{algorithmic}
\end{algorithm}

\begin{theorem} Algorithm~\ref{alg:L2X} is correct.
\end{theorem}
\begin{proof} Suppose that $v\in V$ is a leaf. Then Table~\ref{tab:L2X-base-case} displays the input and output requirements of Algorithm~\ref{alg:L2X} on the vector $(a_0,\dotsc,a_{2^{n_v}-1})=(a_0,a_1)$ for each possible input that includes $v$. The table also shows the output of the algorithm as computed by Lines~\ref{L2X:base}--\ref{L2X:base-end}. The elements $f_i$ and $\fout_i$ that appear in a row of the table are the coefficients of~\eqref{eqn:LX} for the specified value of $\ell$. Elements denoted by asterisks are unspecified by the algorithm. As $v$ is a leaf, we have
\begin{equation*}
	X_{\beta_v,0}=1,
	\quad
	X_{\beta_v,1}=\frac{x}{\beta_{v,0}},
	\quad
	L_{\beta_v,0}=\frac{x}{\beta_{v,0}}+1
	\quad\text{and}\quad
	L_{\beta_v,1}=\frac{x}{\beta_{v,0}}.
\end{equation*}
Moreover, $\PushDown_v(v,\lambda)=\lambda/\beta_{v,0}$ for $\lambda\in\F$. Thus, the coefficients of \eqref{eqn:LX} satisfy $\fout_0=f_0+\PushDown_v(v,\lambda)\left(f_0+f_1\right)$ and $\fout_1=f_0+f_1$ if $\ell=2$, and $\fout_0=f_0=f_1$ if $\ell=1$. Using these equation, one can readily verify that the computed output agrees with the required output for all inputs. Consequently, Algorithm~\ref{alg:L2X} produces the correct output whenever the input vertex is a leaf. Therefore, as $(V,E)$ is a full binary tree, it is sufficient to show that for all internal $v\in V$, if the algorithm produces the correct output whenever $v_\alpha$ or $v_\delta$ is given as an input, then it produces the correct output whenever $v$ is given as an input.

\begin{table}[h]
	\setlength{\belowcaptionskip}{0pt}
	\begin{tabular}{ccccc@{}cccccc@{}cc}
		\toprule
		\multicolumn{5}{c}{Input} &  & \multicolumn{4}{c}{Required output}  & & \multicolumn{2}{c}{Computed output} \\
		\cmidrule{1-5} \cmidrule{7-10} \cmidrule{12-13}
		$c$ & $\ell$ & $\flag$ &   $a_0$   &   $a_1$   &  &  & $a_0$ &  & $a_1$  &  &                      $a_0$                       &     $a_1$     \\
		\midrule
		$2$ &  $2$   &   $0$   &   $f_0$   &   $f_1$   &  &  & $h_0$ &  & $h_1$  &  & $f_0+\PushDown_v(v,\lambda)\left(f_0+f_1\right)$ &   $f_0+f_1$   \\
		$1$ &  $2$   &   $1$   &   $f_0$   & $\fout_1$ &  &  & $h_0$ &  & $f_1$  &  &       $f_0+\PushDown_v(v,\lambda)\fout_1$        & $f_0+\fout_1$ \\
		$1$ &  $2$   &   $0$   &   $f_0$   & $\fout_1$ &  &  & $h_0$ &  & $\ast$ &  &       $f_0+\PushDown_v(v,\lambda)\fout_1$        &    $\ast$     \\
		$0$ &  $2$   &   $1$   & $\fout_0$ & $\fout_1$ &  &  & $f_0$ &  & $\ast$ &  &     $\fout_0+\PushDown_v(v,\lambda)\fout_1$      &    $\ast$     \\
		$1$ &  $1$   &   $1$   &   $f_0$   &  $\ast$   &  &  & $h_0$ &  & $f_1$  &  &                      $f_0$                       &     $f_0$     \\
		$1$ &  $1$   &   $0$   &   $f_0$   &  $\ast$   &  &  & $h_0$ &  & $\ast$ &  &                      $f_0$                       &    $\ast$     \\
		$0$ &  $1$   &   $1$   & $\fout_0$ &  $\ast$   &  &  & $f_0$ &  & $\ast$ &  &                    $\fout_0$                     &    $\ast$     \\
		\bottomrule
	\end{tabular}
	\vspace{\abovecaptionskip}
	\caption{Required and computed outputs of Algorithm~\ref{alg:L2X} when $v$ is a leaf.\strut}\label{tab:L2X-base-case}
\end{table}

Let $v\in V$ be an internal vertex and suppose that Algorithm~\ref{alg:L2X} produces the correct output whenever $v_\alpha$ or $v_\delta$ is given as an input. Let $\lambda\in\F$, $c,\ell\in\N$ such that $c\leq\ell$ and $1\leq\ell\leq 2^{n_v}$, and $\flag\in\{0,1\}$ such that $1\leq\flag+c\leq 2^{n_v}$. Suppose that Algorithm~\ref{alg:L2X} is called on $v$, $(\PushDown_v(u,\lambda))_{u\in\leaves_v}$, $c$, $\ell$, $\flag$ and $(a_0,\dotsc,a_{2^{n_v}-1})$, with $a_i=f_i\in\F$ for $i\in\{0,\dotsc,c-1\}$, and $a_i=\fout_i\in\F$ for $i\in\{c,\dotsc,\ell-1\}$. Then there exist unique elements $\fout_0,\dotsc,\fout_{c-1},f_c,\dotsc,f_{2^{n_v}-1}\in\F$ such that~\eqref{eqn:LX} holds. In-turn, Lemma~\ref{lem:LX-reduction} implies that there exist unique elements $\fint_0,\dotsc,\fint_{2^{n_v}-1}\in\F$ such that~\eqref{eqn:LX-rows} and~\eqref{eqn:LX-cols} hold.

Repeating arguments from the proof of Theorem~\ref{thm:NX-correctness} shows that vector $\mu$ is equal to $(\PushDown_{v_\alpha}(u,\lambda+\omega_{\gamma_v,i}))_{u\in\leaves_{v_\alpha}}$ each time the recursive call of Line~\ref{L2X:row} is made, equal to $(\PushDown_{v_\alpha}(u,\lambda+\omega_{\gamma_v,c_1-1}))_{u\in\leaves_{v_\alpha}}$ whenever the recursive call of Line~\ref{L2X:sometimes-last-row} is performed, and equal to $(\PushDown_{v_\alpha}(u,\lambda+\omega_{\gamma_v,c_1}))_{u\in\leaves_{v_\alpha}}$ whenever the recursive call of Line~\ref{L2X:bottom-row} is performed. Similarly, the vector $\nu$ is equal to $(\PushDown_{v_\delta}(u,\eta))_{u\in\leaves_{v_\delta}}$ for the recursive calls of Lines~\ref{L2X:right-col-1}, \ref{L2X:right-col-2}, \ref{L2X:left-col-1} and~\ref{L2X:left-col-2}. It follows that if the recursive call of Line~\ref{L2X:row} is made with $a_{2^{\dg_v} i+j}=f_{2^{\dg_v} i+j}$ for $j\in\{0,\dotsc,2^{\dg_v}-1\}$, then~\eqref{eqn:LX-rows} and the assumption that the algorithm produces the correct output whenever $v_\alpha$ is given as an input imply that $a_{2^{\dg_v}i+j}=\fint_{2^{\dg_v}i+j}$ for $j\in\{0,\dotsc,2^{\dg_v}-1\}$ afterwards. Similarly, if the recursive call of Line~\ref{L2X:right-col-1} is made with $a_{2^{\dg_v}i+j}=\fint_{2^{\dg_v}i+j}$ for $i\in\{0,\dotsc,c_1-1\}$ and $a_{2^{\dg_v}i+j}=\fout_{2^{\dg_v}i+j}$ for $i\in\{c_1,\dotsc,\ell_1\}$, then~\eqref{eqn:LX-cols} and the assumption that the algorithm produces the correct output whenever $v_\delta$ is given as an input imply that $a_{2^{\dg_v}i+j}=\fout_{2^{\dg_v}i+j}$ for $i\in\{0,\dotsc,c_1-1\}$, and $a_{2^{\dg_v}c_1+j}=\fint_{2^{\dg_v}c_1+j}$ if $\flag'=1$, afterwards. Similar statements hold for the remaining recursive calls made by the algorithm.

The remainder of the proof is split into four cases. For each case, we provide an example in either Figure~\ref{fig:L2X-a-and-b} or~\ref{fig:L2X-c-and-d} of how the vector $(a_0,\dotsc,a_{2^{n_v}-1})$ evolves during the algorithm. In the figures, the vector is represented by the $2^{n_v-\dg_v}\times 2^{\dg_v}$ matrix $(a_{2^{\dg_v}i+j})_{0\leq i< 2^{n_v-\dg_v},0\leq j<2^{\dg_v}}$. Under this representation, the subvectors that are subjected to recursive calls by the algorithm correspond to row and column vectors of the matrix. Asterisks in the figures represent unspecified entries, while entries surrounded by parenthesis are computed only if $\flag=1$, and are unspecified otherwise.

\textbf{Case~a:}
Suppose that $\ell_1=0$. Then Lines~\ref{L2X:aux}--\ref{L2X:aux-end} of the algorithm set $c_1=0$, $c_2=s=c$, $\ell_2=\ell_2'=t=\ell$ and $\flag'=1$.
Thus, Lines~\ref{L2X:rows}--\ref{L2X:rows-end} have no effect.
Equation~\eqref{eqn:LX-cols} implies that Lines~\ref{L2X:right-cols-1}--\ref{L2X:right-cols-1-end} set $a_j=\fint_j$ for $j\in\{c,\dotsc,\ell-1\}$.
Lines~\ref{L2X:right-cols-2}--\ref{L2X:right-cols-2-end} have no effect since $\ell_2'=t$.
Equation~\eqref{eqn:LX-rows} implies that Lines~\ref{L2X:last-row}--\ref{L2X:last-row-end} set $a_i=\fint_i$ for $i\in\{0,\dotsc,c-1\}$, and $a_c=f_c$ if $\flag=1$.
Equation~\eqref{eqn:LX-cols} implies that Lines~\ref{L2X:left-cols-1}--\ref{L2X:left-cols-1-end} set $a_i=\fout_i$ for $i\in\{0,\dotsc,c-1\}$.
Finally, Lines~\ref{L2X:left-cols-2}--\ref{L2X:left-cols-2-end} have no effect since $s=c_2$.

\textbf{Case~b:}
Suppose that $\ell_1\neq 0$ and $c_2=0$.
Then Lines~\ref{L2X:aux}--\ref{L2X:aux-end} set $c_1=c/2^{\dg_v}$, $\ell_2'=2^{\dg_v}$, $\flag'=\flag$, $s=0$ and $t=\ell_2$.
Equation~\eqref{eqn:LX-rows} implies that Lines~\ref{L2X:rows}--\ref{L2X:rows-end} set $a_i=\fint_i$ for $i\in\{0,\dotsc,c-1\}$. 
Thus, \eqref{eqn:LX-cols} implies that Lines~\ref{L2X:right-cols-1}--\ref{L2X:right-cols-2-end} set $a_i=\fout_i$ for $i\in\{0,\dotsc,c-1\}$, and $a_{c+j}=\fint_{c+j}$ for $j\in\{0,\dotsc,2^{\dg_v}-1\}$ if $\flag=1$. Equation~\eqref{eqn:LX-rows} implies that Lines~\ref{L2X:last-row}--\ref{L2X:last-row-end} set $a_c=f_c$ if $\flag=1$, and have no effect otherwise.
Lines~\ref{L2X:left-cols-1}--\ref{L2X:left-cols-2-end} have no effect since $s=c_2=0$.

\tikzset{
	highlight/.style={
		densely dotted,
		thick,
		rounded corners
	},
	lxnode/.style={
		scale=0.87,
		inner sep=0,
		outer sep=0,
		anchor=base,
		text depth=5pt,
		text height=2.6ex,
		text width=2.1em,
		align=center,
	},
    lx matrix/.style={
        matrix of math nodes,
		nodes in empty cells,
		left delimiter=(,
		right delimiter=),
		inner sep=0pt,
		column sep=2pt,
		row sep=2pt,
		nodes={lxnode},
    },
    lx subfig/.style={
       	baseline=(m),
		tight background,
		every left delimiter/.style={xshift=.3em},
		every right delimiter/.style={xshift=-.3em},
    }
}

\begin{figure}
	\begin{tabular}{@{}c@{\ \ }c@{}}
		\textbf{Case a} & \textbf{Case b}\\[1ex]
		\begin{tikzpicture}[lx subfig]
			\matrix (m)[lx matrix]{
				f_0  & f_1  & f_2  & \fout_3 & \fout_4 & \fout_5 & \ast & \ast \\
				\ast & \ast & \ast & \ast    & \ast    & \ast    & \ast & \ast \\
				\ast & \ast & \ast & \ast    & \ast    & \ast    & \ast & \ast \\
				\ast & \ast & \ast & \ast    & \ast    & \ast    & \ast & \ast \\
			};
		\end{tikzpicture}
		&
		\begin{tikzpicture}[lx subfig]
			\matrix (m)[lx matrix]{
				f_0        & f_1        & f_2        & f_3        & f_4        & f_5    & f_6    & f_7    \\
				f_8        & f_9        & f_{10}     & f_{11}     & f_{12}     & f_{13} & f_{14} & f_{15} \\
				\fout_{16} & \fout_{17} & \fout_{18} & \fout_{19} & \fout_{20} & \ast   & \ast   & \ast   \\
				\ast       & \ast       & \ast       & \ast       & \ast       & \ast   & \ast   & \ast   \\
			};
		\end{tikzpicture}
		\vspace{1ex}\\
		\multicolumn{2}{c}{(3a)\ Initial contents.}\\[1ex]
		\begin{tikzpicture}[lx subfig]
			\matrix (m)[lx matrix]{
				f_0  & f_1  & f_2  & \fout_3 & \fout_4 & \fout_5 & \ast & \ast \\
				\ast & \ast & \ast & \ast    & \ast    & \ast    & \ast & \ast \\
				\ast & \ast & \ast & \ast    & \ast    & \ast    & \ast & \ast \\
				\ast & \ast & \ast & \ast    & \ast    & \ast    & \ast & \ast \\
			};
		\end{tikzpicture}
		&
		\begin{tikzpicture}[lx subfig]
			\matrix (m)[lx matrix]{
				\fint_0    & \fint_1    & \fint_2    & \fint_3    & \fint_4    & \fint_5    & \fint_6    & \fint_7    \\
				\fint_8    & \fint_9    & \fint_{10} & \fint_{11} & \fint_{12} & \fint_{13} & \fint_{14} & \fint_{15} \\
				\fout_{16} & \fout_{17} & \fout_{18} & \fout_{19} & \fout_{20} & \ast       & \ast       & \ast       \\
				\ast       & \ast       & \ast       & \ast       & \ast       & \ast       & \ast       & \ast       \\
			};
			\begin{scope}[on background layer]
				\draw[highlight] (m-1-1.south west) rectangle (m-1-8.north east);
				\draw[highlight] (m-2-1.south west) rectangle (m-2-8.north east);
			\end{scope}
		\end{tikzpicture}
		\vspace{1ex}\\
		\multicolumn{2}{c}{(3b)\ After Lines~\ref{L2X:rows}--\ref{L2X:rows-end} make recursive calls on the highlighted rows.}\\[1ex]
		\begin{tikzpicture}[lx subfig]
			\matrix (m)[lx matrix]{
				f_0  & f_1  & f_2  & \fint_3 & \fint_4 & \fint_5 & \ast & \ast \\
				\ast & \ast & \ast & \ast    & \ast    & \ast    & \ast & \ast \\
				\ast & \ast & \ast & \ast    & \ast    & \ast    & \ast & \ast \\
				\ast & \ast & \ast & \ast    & \ast    & \ast    & \ast & \ast \\
			};
			\begin{scope}[on background layer]
				\draw[highlight] (m-4-4.south west) rectangle (m-1-4.north east);
				\draw[highlight] (m-4-5.south west) rectangle (m-1-5.north east);
				\draw[highlight] (m-4-6.south west) rectangle (m-1-6.north east);
			\end{scope}
		\end{tikzpicture}
		&
		\begin{tikzpicture}[lx subfig]
			\matrix (m)[lx matrix]{
				\fout_0      & \fout_1      & \fout_2      & \fout_3      & \fout_4      & \fout_5      & \fout_6      & \fout_7      \\
				\fout_8      & \fout_9      & \fout_{10}   & \fout_{11}   & \fout_{12}   & \fout_{13}   & \fout_{14}   & \fout_{15}   \\
				(\fint_{16}) & (\fint_{17}) & (\fint_{18}) & (\fint_{19}) & (\fint_{20}) & (\fint_{21}) & (\fint_{22}) & (\fint_{23}) \\
				\ast         & \ast         & \ast         & \ast         & \ast         & \ast         & \ast         & \ast         \\
			};
			\begin{scope}[on background layer]
				\draw[highlight] (m-4-1.south west) rectangle (m-1-1.north east);
				\draw[highlight] (m-4-2.south west) rectangle (m-1-2.north east);
				\draw[highlight] (m-4-3.south west) rectangle (m-1-3.north east);
				\draw[highlight] (m-4-4.south west) rectangle (m-1-4.north east);
				\draw[highlight] (m-4-5.south west) rectangle (m-1-5.north east);
				\draw[highlight] (m-4-6.south west) rectangle (m-1-6.north east);
				\draw[highlight] (m-4-7.south west) rectangle (m-1-7.north east);
				\draw[highlight] (m-4-8.south west) rectangle (m-1-8.north east);
			\end{scope}
		\end{tikzpicture}
		\vspace{1ex}\\
		\multicolumn{2}{c}{(3c)\ After Lines~\ref{L2X:right-cols-1}--\ref{L2X:right-cols-2-end} make recursive calls on the highlighted columns.}\\[1ex]
		\begin{tikzpicture}[lx subfig]
			\matrix (m)[lx matrix]{
				\fint_0 & \fint_1 & \fint_2 & (f_3) & \ast & \ast & \ast & \ast \\
				\ast    & \ast    & \ast    & \ast  & \ast & \ast & \ast & \ast \\
				\ast    & \ast    & \ast    & \ast  & \ast & \ast & \ast & \ast \\
				\ast    & \ast    & \ast    & \ast  & \ast & \ast & \ast & \ast \\
			};
			\begin{scope}[on background layer]
				\draw[highlight] (m-1-1.south west) rectangle (m-1-8.north east);
			\end{scope}
		\end{tikzpicture}
		&
		\begin{tikzpicture}[lx subfig]
			\matrix (m)[lx matrix]{
				\fout_0  & \fout_1 & \fout_2    & \fout_3    & \fout_4    & \fout_5    & \fout_6    & \fout_7    \\
				\fout_8  & \fout_9 & \fout_{10} & \fout_{11} & \fout_{12} & \fout_{13} & \fout_{14} & \fout_{15} \\
				(f_{16}) & \ast    & \ast       & \ast       & \ast       & \ast       & \ast       & \ast       \\
				\ast     & \ast    & \ast       & \ast       & \ast       & \ast       & \ast       & \ast       \\
			};
			\begin{scope}[on background layer]
				\draw[highlight] (m-3-1.south west) rectangle (m-3-8.north east);
			\end{scope}
		\end{tikzpicture}
		\vspace{1ex}\\
		\multicolumn{2}{c}{(3d)\ After Lines~\ref{L2X:bottom-row} makes a recursive call on the highlighted row.}\\[1ex]
		\begin{tikzpicture}[lx subfig]
			\matrix (m)[lx matrix]{
				\fout_0 & \fout_1 & \fout_2 & (f_3) & \ast & \ast & \ast & \ast \\
				\ast    & \ast    & \ast    & \ast  & \ast & \ast & \ast & \ast \\
				\ast    & \ast    & \ast    & \ast  & \ast & \ast & \ast & \ast \\
				\ast    & \ast    & \ast    & \ast  & \ast & \ast & \ast & \ast \\
			};
			\begin{scope}[on background layer]
				\draw[highlight] (m-4-1.south west) rectangle (m-1-1.north east);
				\draw[highlight] (m-4-2.south west) rectangle (m-1-2.north east);
				\draw[highlight] (m-4-3.south west) rectangle (m-1-3.north east);
			\end{scope}
		\end{tikzpicture}
		&
		\begin{tikzpicture}[lx subfig]
			\matrix (m)[lx matrix]{
				\fout_0  & \fout_1 & \fout_2    & \fout_3    & \fout_4    & \fout_5    & \fout_6    & \fout_7    \\
				\fout_8  & \fout_9 & \fout_{10} & \fout_{11} & \fout_{12} & \fout_{13} & \fout_{14} & \fout_{15} \\
				(f_{16}) & \ast    & \ast       & \ast       & \ast       & \ast       & \ast       & \ast       \\
				\ast     & \ast    & \ast       & \ast       & \ast       & \ast       & \ast       & \ast       \\
			};
		\end{tikzpicture}
		\vspace{1ex}\\
		\multicolumn{2}{c}{(3e)\ After Lines~\ref{L2X:left-cols-1}--\ref{L2X:left-cols-2-end} make recursive calls on the highlighted columns.}
	\end{tabular}
	\caption{Evolution of the vector $(a_0,\dotsc,a_{2^{n_v}-1})$ during Algorithm~\ref{alg:L2X} for $n_v=5$, $\dg_v=3$, $c=3$, $\ell=6$ (Case~a), and $n_v=5$, $\dg_v=3$, $c=16$, $\ell=21$ (Case~b).}\label{fig:L2X-a-and-b}
\end{figure}

\textbf{Case~c:}
Suppose that $\ell_1\neq 0$ and $0<c_2\leq\ell_2$.
Then Lines~\ref{L2X:aux}--\ref{L2X:aux-end} set $\ell_2'=2^{\dg_v}$, $\flag'=1$, $s=c_2$ and $t=\ell_2$.
Equation~\eqref{eqn:LX-rows} implies that Lines~\ref{L2X:rows}--\ref{L2X:rows-end} set $a_i=\fint_i$ for $i\in\{0,\dotsc,2^{\dg_v}c_1-1\}$. 
Thus, \eqref{eqn:LX-cols} implies that Lines~\ref{L2X:right-cols-1}--\ref{L2X:right-cols-2-end} set $a_{2^{\dg_v}i+j}=\fout_{2^{\dg_v}i+j}$ and $a_{2^{\dg_v}c_1+j}=\fint_{2^{\dg_v}c_1+j}$ for $i\in\{0,\dotsc,c_1-1\}$ and $j\in\{c_2,\dotsc,2^{\dg_v}-1\}$.
Equation~\eqref{eqn:LX-rows} implies that Lines~\ref{L2X:last-row}--\ref{L2X:last-row-end} set $a_{2^{\dg_v}c_1+j}=\fint_{2^{\dg_v}c_1+j}$ for $j\in\{0,\dotsc,c_2-1\}$, and $a_c=f_c$ if $\flag=1$.
Equation~\eqref{eqn:LX-cols} implies that Lines~\ref{L2X:left-cols-1}--\ref{L2X:left-cols-1-end} set $a_{2^{\dg_v}i+j}=\fout_{2^{\dg_v}i+j}$ for $i\in\{0,\dotsc,c_1\}$ and $j\in\{0,\dotsc,c_2-1\}$.
Lines~\ref{L2X:left-cols-2}--\ref{L2X:left-cols-2-end} have no effect since $s=c_2$.

\textbf{Case~d:}
Suppose that $\ell_1\neq 0$ and $\ell_2<c_2$.
Then Lines~\ref{L2X:aux}--\ref{L2X:aux-end} set $\ell_2'=2^{\dg_v}$, $\flag'=1$, $s=\ell_2$ and $t=c_2$.
Equation~\eqref{eqn:LX-rows} implies that Lines~\ref{L2X:rows}--\ref{L2X:rows-end} set $a_i=\fint_i$ for $i\in\{0,\dotsc,2^{\dg_v}c_1-1\}$. 
Lines~\ref{L2X:right-cols-1}--\ref{L2X:right-cols-1-end} have no effect since $t=c_2$.
Equation~\eqref{eqn:LX-cols} implies that Lines~\ref{L2X:right-cols-2}--\ref{L2X:right-cols-2-end} set $a_{2^{\dg_v}i+j}=\fout_{2^{\dg_v}i+j}$ and $a_{2^{\dg_v}c_1+j}=\fint_{2^{\dg_v}c_1+j}$ for $i\in\{0,\dotsc,c_1-1\}$ and $j\in\{c_2,\dotsc,2^{\dg_v}-1\}$.
Equation~\eqref{eqn:LX-rows} implies Lines~\ref{L2X:last-row}--\ref{L2X:last-row-end} set $a_{2^{\dg_v}c_1+j}=\fint_{2^{\dg_v}c_1+j}$ for $j\in\{0,\dotsc,c_2-1\}$, and $a_c=f_c$ if $\flag=1$.
Equation~\eqref{eqn:LX-cols} implies that Lines~\ref{L2X:left-cols-1}--\ref{L2X:left-cols-2-end} set $a_{2^{\dg_v}i+j}=\fout_{2^{\dg_v}i+j}$ for $i\in\{0,\dotsc,c_1\}$ and $j\in\{0,\dotsc,c_2-1\}$.

\begin{figure}
	\begin{tabular}{@{}c@{\ \ }c@{}}
		\textbf{Case c} & \textbf{Case d}\\[1ex]
		\begin{tikzpicture}[lx subfig]
			\matrix (m)[lx matrix]{
				f_0        & f_1        & f_2        & f_3        & f_4        & f_5        & f_6        & f_7        \\
				f_8        & f_9        & f_{10}     & f_{11}     & f_{12}     & f_{13}     & f_{14}     & f_{15}     \\
				f_{16}     & f_{17}     & \fout_{18} & \fout_{19} & \fout_{20} & \fout_{21} & \fout_{22} & \fout_{23} \\
				\fout_{24} & \fout_{25} & \fout_{26} & \fout_{27} & \ast       & \ast       & \ast       & \ast       \\
			};
		\end{tikzpicture}
		&
		\begin{tikzpicture}[lx subfig]
			\matrix (m)[lx matrix]{
				f_0        & f_1        & f_2        & f_3        & f_4    & f_5    & f_6        & f_7        \\
				f_8        & f_9        & f_{10}     & f_{11}     & f_{12} & f_{13} & f_{14}     & f_{15}     \\
				f_{16}     & f_{17}     & f_{18}     & f_{19}     & f_{20} & f_{21} & \fout_{22} & \fout_{23} \\
				\fout_{24} & \fout_{25} & \fout_{26} & \fout_{27} & \ast   & \ast   & \ast       & \ast       \\
			};
		\end{tikzpicture}
		\vspace{1ex}\\
		\multicolumn{2}{c}{(4a)\ Initial contents.}\\[1ex]
		\begin{tikzpicture}[lx subfig]
			\matrix (m)[lx matrix]{
				\fint_0    & \fint_1    & \fint_2    & \fint_3    & \fint_4    & \fint_5    & \fint_6    & \fint_7    \\
				\fint_8    & \fint_9    & \fint_{10} & \fint_{11} & \fint_{12} & \fint_{13} & \fint_{14} & \fint_{15} \\
				f_{16}     & f_{17}     & \fout_{18} & \fout_{19} & \fout_{20} & \fout_{21} & \fout_{22} & \fout_{23} \\
				\fout_{24} & \fout_{25} & \fout_{26} & \fout_{27} & \ast       & \ast       & \ast       & \ast       \\
			};
			\begin{scope}[on background layer]
				\draw[highlight] (m-1-1.south west) rectangle (m-1-8.north east);
				\draw[highlight] (m-2-1.south west) rectangle (m-2-8.north east);
			\end{scope}
		\end{tikzpicture}
		&
		\begin{tikzpicture}[lx subfig]
			\matrix (m)[lx matrix]{
				\fint_0    & \fint_1    & \fint_2    & \fint_3    & \fint_4    & \fint_5    & \fint_6    & \fint_7    \\
				\fint_8    & \fint_9    & \fint_{10} & \fint_{11} & \fint_{12} & \fint_{13} & \fint_{14} & \fint_{15} \\
				f_{16}     & f_{17}     & f_{18}     & f_{19}     & f_{20}     & f_{21}     & \fout_{22} & \fout_{23} \\
				\fout_{24} & \fout_{25} & \fout_{26} & \fout_{27} & \ast       & \ast       & \ast       & \ast       \\
			};
			\begin{scope}[on background layer]
				\draw[highlight] (m-1-1.south west) rectangle (m-1-8.north east);
				\draw[highlight] (m-2-1.south west) rectangle (m-2-8.north east);
			\end{scope}
		\end{tikzpicture}
		\vspace{1ex}\\
		\multicolumn{2}{c}{(4b)\ After Lines~\ref{L2X:rows}--\ref{L2X:rows-end} make recursive calls on the highlighted rows.}\\[1ex]
		\begin{tikzpicture}[lx subfig]
			\matrix (m)[lx matrix]{
				\fint_0    & \fint_1    & \fout_2    & \fout_3    & \fout_4    & \fout_5    & \fout_6    & \fout_7    \\
				\fint_8    & \fint_9    & \fout_{10} & \fout_{11} & \fout_{12} & \fout_{13} & \fout_{14} & \fout_{15} \\
				f_{16}     & f_{17}     & \fint_{18} & \fint_{19} & \fint_{20} & \fint_{21} & \fint_{22} & \fint_{23} \\
				\fout_{24} & \fout_{25} & \ast       & \ast       & \ast       & \ast       & \ast       & \ast       \\
			};
			\begin{scope}[on background layer]
				\draw[highlight] (m-4-3.south west) rectangle (m-1-3.north east);
				\draw[highlight] (m-4-4.south west) rectangle (m-1-4.north east);
				\draw[highlight] (m-4-5.south west) rectangle (m-1-5.north east);
				\draw[highlight] (m-4-6.south west) rectangle (m-1-6.north east);
				\draw[highlight] (m-4-7.south west) rectangle (m-1-7.north east);
				\draw[highlight] (m-4-8.south west) rectangle (m-1-8.north east);
			\end{scope}
		\end{tikzpicture}
		&
		\begin{tikzpicture}[lx subfig]
			\matrix (m)[lx matrix]{
				\fint_0    & \fint_1    & \fint_2    & \fint_3    & \fint_4    & \fint_5    & \fout_6    & \fout_7    \\
				\fint_8    & \fint_9    & \fint_{10} & \fint_{11} & \fint_{12} & \fint_{13} & \fout_{14} & \fout_{15} \\
				f_{16}     & f_{17}     & f_{18}     & f_{19}     & f_{20}     & f_{21}     & \fint_{22} & \fint_{23} \\
				\fout_{24} & \fout_{25} & \fout_{26} & \fout_{27} & \ast       & \ast       & \ast       & \ast       \\
			};
			\begin{scope}[on background layer]
				\draw[highlight] (m-4-7.south west) rectangle (m-1-7.north east);
				\draw[highlight] (m-4-8.south west) rectangle (m-1-8.north east);
			\end{scope}
		\end{tikzpicture}
		\vspace{1ex}\\
		\multicolumn{2}{c}{(4c)\ After Lines~\ref{L2X:right-cols-1}--\ref{L2X:right-cols-2-end} make recursive calls on the highlighted columns.}\\[1ex]
		\begin{tikzpicture}[lx subfig]
			\matrix (m)[lx matrix]{
				\fint_0    & \fint_1    & \fout_2    & \fout_3    & \fout_4    & \fout_5    & \fout_6    & \fout_7    \\
				\fint_8    & \fint_9    & \fout_{10} & \fout_{11} & \fout_{12} & \fout_{13} & \fout_{14} & \fout_{15} \\
				\fint_{16} & \fint_{17} & (f_{18})   & \ast       & \ast       & \ast       & \ast       & \ast       \\
				\fout_{24} & \fout_{25} & \ast       & \ast       & \ast       & \ast       & \ast       & \ast       \\
			};
			\begin{scope}[on background layer]
				\draw[highlight] (m-3-1.south west) rectangle (m-3-8.north east);
			\end{scope}
		\end{tikzpicture}
		&
		\begin{tikzpicture}[lx subfig]
			\matrix (m)[lx matrix]{
				\fint_0    & \fint_1    & \fint_2    & \fint_3    & \fint_4    & \fint_5    & \fout_6    & \fout_7    \\
				\fint_8    & \fint_9    & \fint_{10} & \fint_{11} & \fint_{12} & \fint_{13} & \fout_{14} & \fout_{15} \\
				\fint_{16} & \fint_{17} & \fint_{18} & \fint_{19} & \fint_{20} & \fint_{21} & (f_{22})   & \ast       \\
				\fout_{24} & \fout_{25} & \fout_{26} & \fout_{27} & \ast       & \ast       & \ast       & \ast       \\
			};
			\begin{scope}[on background layer]
				\draw[highlight] (m-3-1.south west) rectangle (m-3-8.north east);
			\end{scope}
		\end{tikzpicture}
		\vspace{1ex}\\
		\multicolumn{2}{c}{(4d)\ After Lines~\ref{L2X:bottom-row} makes a recursive call on the highlighted row.}\\[1ex]
		\begin{tikzpicture}[lx subfig]
			\matrix (m)[lx matrix]{
				\fout_0    & \fout_1    & \fout_2    & \fout_3    & \fout_4    & \fout_5    & \fout_6    & \fout_7    \\
				\fout_8    & \fout_9    & \fout_{10} & \fout_{11} & \fout_{12} & \fout_{13} & \fout_{14} & \fout_{15} \\
				\fout_{16} & \fout_{17} & (f_{18})   & \ast       & \ast       & \ast       & \ast       & \ast       \\
				\ast       & \ast       & \ast       & \ast       & \ast       & \ast       & \ast       & \ast       \\
			};
			\begin{scope}[on background layer]
				\draw[highlight] (m-4-1.south west) rectangle (m-1-1.north east);
				\draw[highlight] (m-4-2.south west) rectangle (m-1-2.north east);
			\end{scope}
		\end{tikzpicture}
		&
		\begin{tikzpicture}[lx subfig]
			\matrix (m)[lx matrix]{
				\fout_0    & \fout_1    & \fout_2    & \fout_3    & \fout_4    & \fout_5    & \fout_6    & \fout_7    \\
				\fout_8    & \fout_9    & \fout_{10} & \fout_{11} & \fout_{12} & \fout_{13} & \fout_{14} & \fout_{15} \\
				\fout_{16} & \fout_{17} & \fout_{18} & \fout_{19} & \fout_{20} & \fout_{21} & (f_{22})   & \ast       \\
				\ast       & \ast       & \ast       & \ast       & \ast       & \ast       & \ast       & \ast       \\
			};
			\begin{scope}[on background layer]
				\draw[highlight] (m-4-1.south west) rectangle (m-1-1.north east);
				\draw[highlight] (m-4-2.south west) rectangle (m-1-2.north east);
				\draw[highlight] (m-4-3.south west) rectangle (m-1-3.north east);
				\draw[highlight] (m-4-4.south west) rectangle (m-1-4.north east);
				\draw[highlight] (m-4-5.south west) rectangle (m-1-5.north east);
				\draw[highlight] (m-4-6.south west) rectangle (m-1-6.north east);
			\end{scope}
		\end{tikzpicture}
		\vspace{1ex}\\
		\multicolumn{2}{c}{(4e)\ After Lines~\ref{L2X:left-cols-1}--\ref{L2X:left-cols-2-end} make recursive calls on the highlighted columns.}
	\end{tabular}
	\caption{Evolution of the vector $(a_0,\dotsc,a_{2^{n_v}-1})$ during Algorithm~\ref{alg:L2X} for $n_v=5$, $\dg_v=3$, $c=18$, $\ell=28$ (Case~c), and  $n_v=5$, $\dg_v=3$, $c=22$, $\ell=28$ (Case~d).}\label{fig:L2X-c-and-d}
\end{figure}
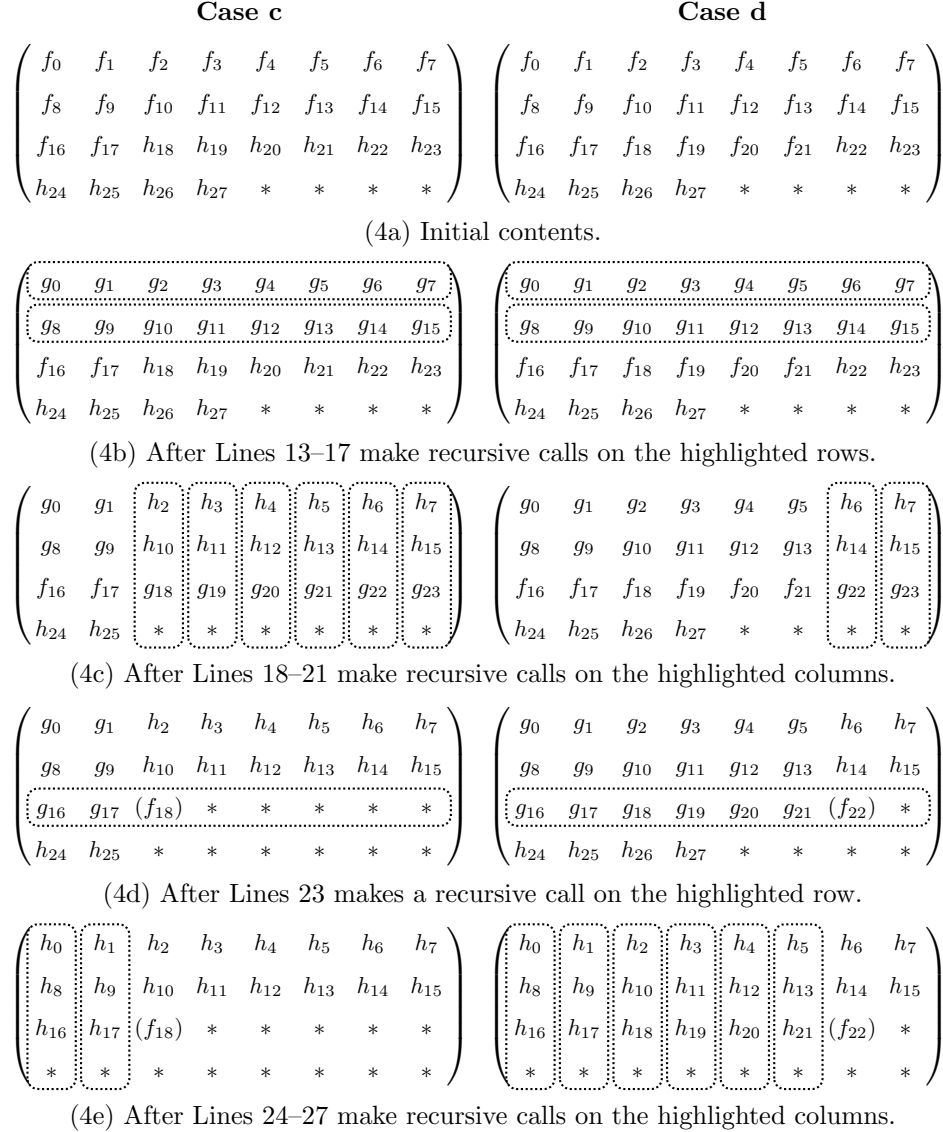

In each of the four cases, Algorithm~\ref{alg:L2X} terminates with $a_i=\fout_i$ for $i\in\{0,\dotsc,c-1\}$, and $a_c=f_c$ if $\flag=1$, as required. Hence, for internal $v\in V$, if the algorithm produces the correct output whenever $v_\alpha$ or $v_\delta$ is given as an input, then it produces the correct output whenever $v$ is given as an input.
\end{proof}

Algorithm~\ref{alg:L2X} requires the same precomputations as the algorithms of Section~\ref{sec:NX}. However, as the length of the vector on which the algorithm operates is no longer tied to the polynomial length $\ell$, the auxiliary space requirements grow to $2^{n_v}-\ell+\bigO(n^2)$ field elements, where the last term accounts for the storage of the vectors $\mu$ and $\nu$, and the precomputed elements $\PushDown_{v}(u,\sigma_{v,i})$. The update to the vector $\mu$ in Line~\ref{L2X:update-shifts} of the algorithm performs $(c_1+\flag'-1)\dg_v=(\floor{c/2^{\dg_v}}+\flag'-1)\dg_v$ additions over all iterations of its containing loop. Therefore, as for the algorithms of Section~\ref{sec:NX}, small values of $\dg_v$ should be avoided when choosing a reduction tree.

\begin{example}\label{ex:L2X} For $\beta$ equal to a Cantor basis of dimension $15$, and inputs $\ell\in\{1,\dotsc,2^{15}\}$, $c=\ell$ and $\flag=0$, Figure~\ref{fig:cantor_lx} shows the maximum and minimum number of additions performed by Algorithm~\ref{alg:L2X} over all possible reduction trees for the basis. The number of multiplications performed by the algorithm for each value of $\ell$, which is independent of the choice of reduction tree, is also shown in the figure. As for Example~\ref{ex:NX}, the number of additions performed for each value of $\ell$ is maximised by the tree with $\image(\dg)\subseteq\{0,1\}$, and minimised by the tree such that $\dg_v=2^{\ceil{\log_2 n_v}-1}$ for all internal $v\in V$.
\begin{figure}
	\includegraphics{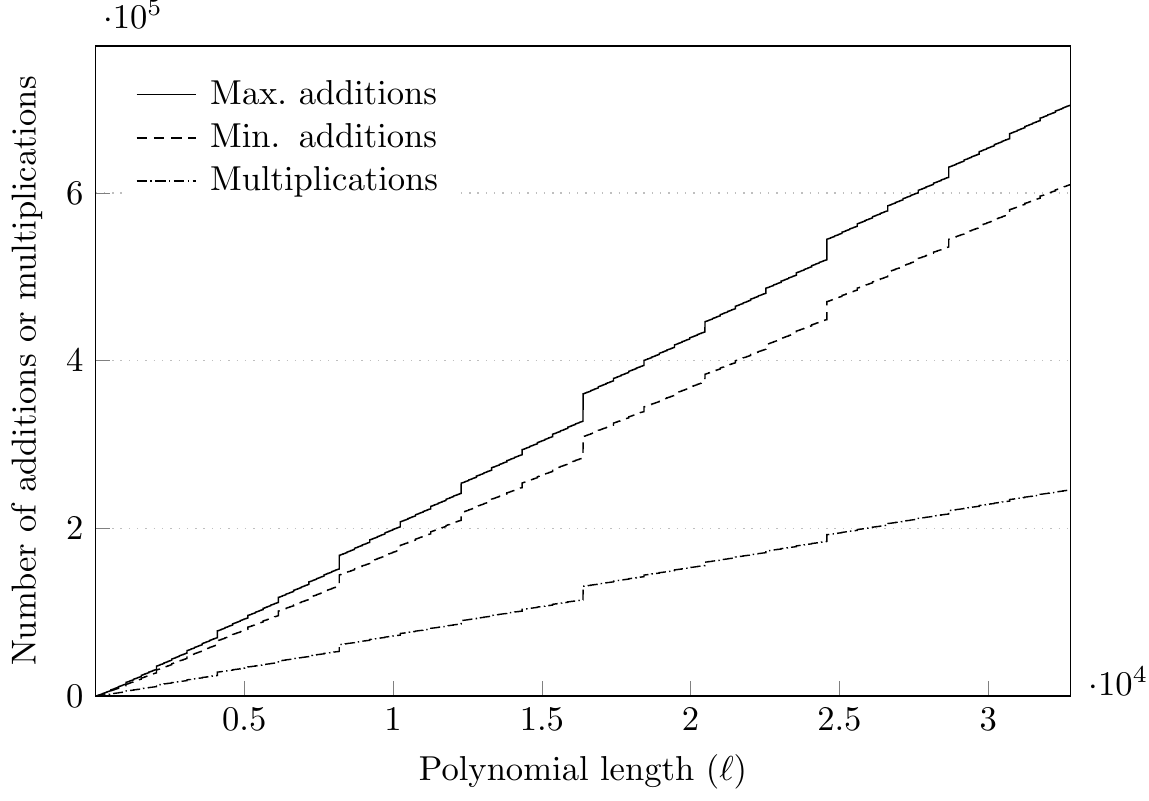}
	\caption{Maximum and minimum number of operations performed by Algorithm~\ref{alg:L2X} (Algorithm~\ref{alg:X2L}) for Cantor bases of dimension $15$, with parameters $c=\ell$ and $\flag=0$ (respectively, $c=\ell$).}
	\label{fig:cantor_lx}
\end{figure}
\end{example}

\begin{theorem}\label{thm:L2X-complexity} Algorithm~\ref{alg:L2X} performs at most
\begin{equation*}
	\min
	\left(
		\frac{c+\flag-1}{2}
		\left(\ceil{\log_2c+\flag}-1\right)
		+\ell-1,
		2^{n_v-1}n_v
	\right)
\end{equation*}
multiplications in $\F$, and at most
\begin{equation*}
	\min
	\left(
		\frac{c+\flag-1}{2}
		\left(3\ceil{\log_2c+\flag}-1\right)
		+\ell-1,
		2^{n_v-1}\left(3n_v-2\right)+1
	\right)
\end{equation*}
additions in $\F$.
\end{theorem}

We split the proof of Theorem~\ref{thm:L2X-complexity} into four lemmas, one for each bound. It is readily verified that the bounds hold if the input vertex is a leaf. Therefore, as $(V,E)$ is a full binary tree, it is sufficient to show for each bound that if $v\in V$ is an internal vertex such that the bound holds whenever the input vertex is $v_\alpha$ or $v_\delta$, then the bound holds whenever $v$ is the input vertex.

\begin{lemma}\label{lem:L2X-multiplications} Algorithm~\ref{alg:L2X} performs at most $2^{n_v-1}n_v$ multiplications in $\F$.
\end{lemma}
\begin{proof} Suppose that Algorithm~\ref{alg:L2X} is called on an internal vertex $v\in V$ such that the bound of the lemma holds whenever the input vertex is $v_\alpha$ or $v_\delta$. Then Lines~\ref{L2X:rows}--\ref{L2X:rows-end} of the algorithm perform at most $c_12^{\dg_v-1}\dg_v\leq(2^{n_v-\dg_v}-\flag')2^{\dg_v-1}\dg_v$ multiplications, since $c_1\leq 2^{n_v-\dg_v}$ with equality implying $c_2=\flag=\flag'=0$. If $\ell\geq 2^{\dg_v}$, then $\ell_2'=2^{\dg_v}\geq t\geq c_2$. If $\ell<2^{\dg_v}$, then $\ell=\ell_2=\ell_2'\geq c=c_2$. Thus, the inequalities $c_2\leq t\leq\ell_2'\leq 2^{\dg_v}$ hold in either case. It follows that Lines~\ref{L2X:right-cols-1}--\ref{L2X:right-cols-2-end} perform at most $(2^{\dg_v}-c_2)2^{n_v-\dg_v-1}(n_v-\dg_v)$ multiplications. Lines~\ref{L2X:last-row}--\ref{L2X:last-row-end} performs at most $\flag'2^{\dg_v-1}\dg_v$ multiplications, while Lines~\ref{L2X:left-cols-1}--\ref{L2X:left-cols-2-end} perform at most $c_22^{n_v-\dg_v-1}(n_v-\dg_v)$. Summing these contributions, it follows that Algorithm~\ref{alg:L2X} performs at most $2^{n_v-1}n_v$ multiplications. 
\end{proof}

\begin{lemma} Algorithm~\ref{alg:L2X} performs at most
\begin{equation}\label{eqn:L2X-mults-small-bnd}
	\frac{c+\flag-1}{2}
	\left(\ceil{\log_2c+\flag}-1\right)
	+\ell-1
\end{equation}
multiplications in $\F$.
\end{lemma}
\begin{proof} Suppose that Algorithm~\ref{alg:L2X} is called on an internal vertex $v\in V$ such that the bound of the lemma holds whenever the input vertex is $v_\alpha$ or $v_\delta$. Then Lemma~\ref{lem:L2X-multiplications} implies that Lines~\ref{L2X:rows}--\ref{L2X:rows-end} perform at most
\begin{equation}\label{eqn:L2X-rows-mults}
	c_12^{\dg_v-1}\dg_v=\frac{c-c_2}{2}\dg_v
\end{equation}
multiplications. As $c_2\leq t\leq\ell_2'\leq 2^{\dg_v}$, Lines~\ref{L2X:right-cols-1}--\ref{L2X:right-cols-2-end} perform at most
\begin{equation}\label{eqn:L2X-right-cols-mults}
	\left(2^{\dg_v}-c_2\right)
	\frac{c_1+\flag'-1}{2}
	\left(\ceil{\log_2c_1+\flag'}-1\right)
	+\left(\ell_2'-c_2\right)
	\left(\ell_1-1\right)
	+t-c_2
\end{equation}
multiplications. Lines~\ref{L2X:last-row}--\ref{L2X:last-row-end} perform at most
\begin{equation}\label{eqn:L2X-last-row-mults}
	\flag'
	\left(
		\frac{c_2+\flag-1}{2}
		\left(\ceil{\log_2\max(c_2+\flag,1)}-1\right)
		+\ell_2'-1
	\right)
\end{equation}
multiplications, while Lines~\ref{L2X:left-cols-1}--\ref{L2X:left-cols-2-end} perform at most
\begin{equation}\label{eqn:L2X-left-cols-mults}
	c_2
	\left(
		\frac{c_1}{2}
		\left(\ceil{\log_2c_1+1}-1\right)
		+\ell_1-1
	\right)
	+s
\end{equation}
multiplications. We show that the sum of these contributions is bounded by~\eqref{eqn:L2X-mults-small-bnd}.

Suppose that $\flag'=1$. Then $1\leq c_2+\flag\leq c+\flag$. Thus, the sum of \eqref{eqn:L2X-rows-mults}--\eqref{eqn:L2X-left-cols-mults} in this case is at most
\begin{equation}\label{eqn:L2X-mults-flag-1}
	\frac{c-c_2}{2}
	\left(\ceil{\log_2c_1+1}+\dg_v-1\right)
	+\frac{c_2+\flag-1}{2}
	\left(\ceil{\log_2c+\flag}-1\right)
	+\ell-1,
\end{equation}
since
\begin{equation*}
	\ell_2'\left(\ell_1-1\right)
	+s+t-c_2
	=\ell_2'\left(\ell_1-1\right)+\ell_2
	\leq 2^{\dg_v}\ell_1+\ell_2-\ell_2'
	=\ell-\ell_2'.
\end{equation*}
If $c_1=0$, then $c_2=c$. If $c_1\neq 0$ and $c_2\neq 0$, then
\begin{equation*}
	\ceil{\log_2c_1+1}
	=\ceil{\log_2\ceil{c/2^{\dg_v}}}
	=\ceil{\log_2 c}-\dg_v
	\leq\ceil{\log_2 c+\flag}-\dg_v.
\end{equation*}
If $c_1\neq 0$ and $c_2=0$, then $\flag=1$ since $c_2+\flag\geq 1$, and
\begin{equation*}
	\ceil{\log_2c_1+1}
	=\ceil{\log_2 c/2^{\dg_v}+1}
	=\ceil{\log_2 c+1}-\dg_v
	=\ceil{\log_2 c+b}-\dg_v.
\end{equation*}
In all three of these cases, substituting into~\eqref{eqn:L2X-mults-flag-1} yields~\eqref{eqn:L2X-mults-small-bnd}.

Suppose now that $\flag'=0$. Then $c_2=0$ and $\flag=0$. As $c+\flag\geq 1$, it follows that $c_1=c/2^{\dg_v}\neq 0$. Thus, $\ell_2'=2^{\dg_v}$, since $\ell\geq c\geq 2^{\dg_v}$. Therefore, the sum of \eqref{eqn:L2X-rows-mults}--\eqref{eqn:L2X-left-cols-mults} in this case is at most
\begin{multline*}
	\frac{c}{2}\dg_v
	+\frac{c-2^{\dg_v}}{2}\left(\ceil{\log_2c_1}-1\right)
	+\ell-2^{\dg_v}\\
	\begin{aligned}
		&=
		\frac{c}{2}\dg_v
		+\frac{c-1}{2}\left(\ceil{\log_2c}-\dg_v-1\right)
		-\frac{2^{\dg_v}-1}{2}\left(\ceil{\log_2c_1}-1\right)
		+\ell-2^{\dg_v}\\
		&\leq
		\frac{c}{2}\dg_v
		+\frac{c-1}{2}\left(\ceil{\log_2c}-\dg_v-1\right)
		+\frac{2^{\dg_v}-1}{2}
		+\ell-2^{\dg_v}\\
		&=
		\frac{c-1}{2}
		\left(\ceil{\log_2c}-1\right)
		+\ell-1
		-\frac{1}{2}\left(2^{\dg_v}-\dg_v-1\right)\\
		&\leq
		\frac{c+\flag-1}{2}
		\left(\ceil{\log_2c+\flag}-1\right)
		+\ell-1,
	\end{aligned}
\end{multline*}
as required.
\end{proof}

\begin{lemma}\label{lem:L2X-additions} Algorithm~\ref{alg:L2X} performs at most $2^{n_v-1}\left(3n_v-2\right)+1$ additions in $\F$.
\end{lemma}
\begin{proof} Suppose that Algorithm~\ref{alg:L2X} is called on an internal vertex $v\in V$ such that the bound of the lemma holds whenever the input vertex is $v_\alpha$ or $v_\delta$. Then Lines~\ref{L2X:rows}--\ref{L2X:rows-end} of the algorithm perform at most
\begin{multline*}
	c_1
	\left(
		2^{\dg_v-1}
		\left(3\dg_v-2\right)
		+1
		+\dg_v
	\right)
	-\left(1-\flag'\right)\dg_v\\
	\leq
	\left(2^{n_v-\dg_v}-\flag'\right)
	\left(
		2^{\dg_v-1}
		\left(3\dg_v-2\right)
		+1
	\right)
	+\left(2^{n_v-\dg_v}-1\right)\dg_v
\end{multline*}
additions, since $c_1\leq 2^{n_v-\dg_v}$ with equality implying that $c_2=\flag=\flag'=0$. As $c_2\leq t\leq\ell_2'\leq 2^{\dg_v}$, Lines~\ref{L2X:right-cols-1}--\ref{L2X:right-cols-2-end} perform at most $(2^{\dg_v}-c_2)(2^{n_v-\dg_v-1}(3(n_v-\dg_v)-2)+1)$ additions. Lines~\ref{L2X:last-row}--\ref{L2X:last-row-end} perform at most $\flag'(2^{\dg_v-1}(3\dg_v-2)+1)$ additions, while Lines~\ref{L2X:left-cols-1}--\ref{L2X:left-cols-2-end} perform at most $c_2(2^{n_v-\dg_v-1}(3(n_v-\dg_v)-2)+1)$ additions. Summing these contributions, it follows that Algorithm~\ref{alg:L2X} performs at most
\begin{equation*}
	2^{n_v-1}(3n_v-2)+1
	-\left(2^{\dg_v}-\dg_v-1\right)\left(2^{n_v-\dg_v}-1\right)
	\leq
	2^{n_v-1}(3n_v-2)+1
\end{equation*}
additions.
\end{proof}

\begin{lemma} Algorithm~\ref{alg:L2X} performs at most
\begin{equation}\label{eqn:L2X-adds-small-bnd}
	\frac{c+\flag-1}{2}
	\left(3\ceil{\log_2c+\flag}-1\right)
	+\ell-1
\end{equation}
additions in $\F$.
\end{lemma}
\begin{proof} Suppose that Algorithm~\ref{alg:L2X} is called on an internal vertex $v\in V$ such that the bound of the lemma holds whenever the input vertex is $v_\alpha$ or $v_\delta$. Then Lemma~\ref{lem:L2X-additions} implies that Lines~\ref{L2X:rows}--\ref{L2X:rows-end} perform at most
\begin{equation}\label{eqn:L2X-rows-adds}
	c_1
	\left(
		2^{\dg_v-1}
		\left(3\dg_v-2\right)
		+1
		+\dg_v
	\right)
	-\left(1-\flag'\right)\dg_v
	\leq
	\frac{c-c_2}{2}
	3\dg_v
	-\left(1-\flag'\right)\dg_v
\end{equation}
additions. As $c_2\leq t\leq\ell_2'\leq 2^{\dg_v}$, Lines~\ref{L2X:right-cols-1}--\ref{L2X:right-cols-2-end} perform at most
\begin{equation}\label{eqn:L2X-right-cols-adds}
	\left(2^{\dg_v}-c_2\right)
	\frac{c_1+\flag'-1}{2}
	\left(3\ceil{\log_2c_1+\flag'}-1\right)
	+\left(\ell_2'-c_2\right)
	\left(\ell_1-1\right)
	+t-c_2
\end{equation}
additions. Lines~\ref{L2X:last-row}--\ref{L2X:last-row-end} perform at most
\begin{equation}\label{eqn:L2X-last-row-adds}
	\flag'
	\left(
		\frac{c_2+\flag-1}{2}
		\left(3\ceil{\log_2\max(c_2+\flag,1)}-1\right)
		+\ell_2'-1
	\right)
\end{equation}
additions, while Lines~\ref{L2X:left-cols-1}--\ref{L2X:left-cols-2-end} perform at most
\begin{equation}\label{eqn:L2X-left-cols-adds}
	c_2
	\left(
		\frac{c_1}{2}
		\left(3\ceil{\log_2c_1+1}-1\right)
		+\ell_1-1
	\right)
	+s
\end{equation}
additions. We show that the sum of these contributions is bounded by~\eqref{eqn:L2X-adds-small-bnd}.

Suppose that $\flag'=1$. Then the sum of \eqref{eqn:L2X-rows-adds}--\eqref{eqn:L2X-left-cols-adds} is at most
\begin{equation}\label{eqn:L2X-adds-flag-1}
	\frac{c-c_2}{2}
	\left(3\ceil{\log_2c_1+1}+3\dg_v-1\right)
	+\frac{c_2+\flag-1}{2}
	\left(3\ceil{\log_2c+\flag}-1\right)
	+\ell-1.
\end{equation}
If $c_1=0$, then $c_2=c$. If not, then $\ceil{\log_2c_1+1}\leq\ceil{\log_2 c+\flag}-\dg_v$. In either case, substituting into~\eqref{eqn:L2X-adds-flag-1} yields~\eqref{eqn:L2X-adds-small-bnd}.

Suppose now that $\flag'=0$. Then $c_2=\flag=0$, $c_1\neq 0$ and $\ell_2'=2^{\dg_v}$. It follows that the sum of \eqref{eqn:L2X-rows-adds}--\eqref{eqn:L2X-left-cols-adds} in this case is at most
\begin{multline*}
	\frac{c}{2}3\dg_v-\dg_v
	+\frac{c-2^{\dg_v}}{2}
	\left(3\ceil{\log_2c_1}-1\right)
	+\ell-2^{\dg_v}\\
	\begin{aligned}
		 &\leq
		 \frac{c-1}{2}3\dg_v+\frac{1}{2}\dg_v
		 +\frac{c-1}{2}\left(3\ceil{\log_2c}-3\dg_v-1\right)
		 +\frac{2^{\dg_v}-1}{2}
		 +\ell-2^{\dg_v}\\
		 &=
		 \frac{c-1}{2}\left(3\ceil{\log_2c}-1\right)+\ell-1
		 -\frac{1}{2}\left(2^{\dg_v}-\dg_v-1\right)\\
		 &\leq
		 \frac{c+\flag-1}{2}\left(3\ceil{\log_2c+\flag}-1\right)
		 +\ell-1,
	\end{aligned}
\end{multline*}
as required.
\end{proof}

\subsection{Conversion from the Lin--Chung--Han basis to the Lagrange basis}\label{sec:X2L}

We propose Algorithm~\ref{alg:X2L} for converting from the LCH basis to the Lagrange basis. The parameter $c$ plays a different role than in Algorithm~\ref{alg:L2X}, with its function being to specify the number Lagrange basis coefficients returned by the algorithm, rather than to specify a mixture of coefficients. For conversion from the LCH basis to the Lagrange basis, Algorithm~\ref{alg:X2L} is initially called with $c=2^{n_v}$. Smaller initial values of $c$ are relevant, for example, when using the algorithm within the Hermite evaluation algorithm of Coxon~\cite{coxon2018}.

\begin{algorithm}[h!]
	\caption{$\mathsf{X2L}(v,(\PushDown_v(u,\lambda))_{u\in\leaves_v},c,\ell,(a_0,\ldots,a_{2^{n_v}-1}))$}\label{alg:X2L}
	\begin{algorithmic}[1]
		\Require a vertex $v\in V$, the vector $(\PushDown_v(u,\lambda))_{u\in\leaves_v}\in\F^{n_v}$ for some $\lambda\in\F$, $c,\ell\in\{1,2,\dotsc,2^{n_v}\}$, and $a_i=\fout_i\in\F$ for $i\in\{0,\dotsc,\ell-1\}$.
		\Ensure $a_i=f_i\in\F$ for $i\in\{0,\dotsc,c-1\}$ such that~\eqref{eqn:LX} holds for some $f_c,\dotsc,f_{2^{n_v}-1}\in\F$.
		\If{$v$ is a leaf}\label{X2L:base}
			\If{$c=2$ and $\ell=2$}
				$a_0\set a_0+\PushDown_v(v,\lambda)a_1$, $a_1\set a_0+a_1$
			\EndIf
 			\If{$c=1$ and $\ell=2$}
	 			$a_0\set a_0+\PushDown_v(v,\lambda)a_1$
			\EndIf
			\If{$c=2$ and $\ell=1$}
				$a_1\set a_0$
			\EndIf
		\State\Return
		\EndIf\label{X2L:base-end}
		\State
			\label{X2L:aux}%
			$c_1\set\ceil{c/2^{\dg_v}}-1$,
			$c_2\set c-2^{\dg_v}c_1$
		\State
			\label{X2L:aux-end}%
			$\ell_1\set\floor{\ell/2^{\dg_v}}$,
			$\ell_2\set\ell-2^{\dg_v}\ell_1$
			$\ell_2'\set\min(2^{\dg_v},\ell)$
		\State
			$\mu\set(\PushDown_v(u,\lambda))_{u\in\leaves_{v_\alpha}}$,
			$\nu\set(\PushDown_v(u,\lambda))_{u\in\leaves_{v_\delta}}$
		\For{$j=0,\dotsc,\ell_2-1$}\label{X2L:left-cols}
			\State
				\label{X2L:left-col}%
				\Call{X2L}{$v_\delta,\nu,c_1+1,\ell_1+1,(a_j,a_{2^{\dg_v}+j},\dotsc,a_{2^{\dg_v}(2^{n_v-\dg_v}-1)+j})$}
		\EndFor\label{X2L:left-cols-end}
		\For{$j=\ell_2,\dotsc,\ell_2'-1$}\label{X2L:right-cols}
			\State
				\label{X2L:right-col}%
				\Call{X2L}{$v_\delta,\nu,c_1+1,\ell_1,(a_j,a_{2^{\dg_v}+j},\dotsc,a_{2^{\dg_v}(2^{n_v-\dg_v}-1)+j})$}
		\EndFor\label{X2L:right-cols-end}
		\For{$i=0,\dotsc,c_1-1$}\label{X2L:rows}
			\State
				\label{X2L:row}%
				\Call{X2L}{$v_\alpha,\mu,2^{\dg_v},\ell_2',(a_{2^{\dg_v}i},a_{2^{\dg_v}i+1},\dotsc,a_{2^{\dg_v}(i+1)-1})$}
			\State
				\label{X2L:update-shifts}%
				$\mu\set\mu+(\PushDown_v(u,\sigma_{v,\Gray(i)}))_{u\in\leaves_{v_\alpha}}$
		\EndFor\label{X2L:rows-end}
		\State\label{X2L:last-row}\Call{X2L}{$v_\alpha,\mu,c_2,\ell_2',(a_{2^{\dg_v}c_1},a_{2^{\dg_v}c_1+1},\dotsc,a_{2^{\dg_v}(c_1+1)-1})$}
	\end{algorithmic}
\end{algorithm}

\begin{theorem} Algorithm~\ref{alg:X2L} is correct.
\end{theorem}
\begin{proof} Table~\ref{tab:X2L-base-case} displays the input and output requirements of Algorithm~\ref{alg:X2L} when the input vertex $v$ is a leaf, as well as showing the output of the algorithm as computed by Lines~\ref{X2L:base}--\ref{X2L:base-end}. The elements $f_i$ and $\fout_i$ that appear in a row of the table are the coefficients of~\eqref{eqn:LX} for the specified value of $\ell$. Elements denoted by asterisks are unspecified by the algorithm. As $v$ is a leaf, the coefficients of \eqref{eqn:LX} satisfy $f_0=\fout_0+\PushDown_v(v,\lambda)\fout_1$ and $f_1=\fout_1+(\fout_0+\PushDown_v(v,\lambda)\fout_1)$ if $\ell=2$, and $f_0=f_1=\fout_0$ if $\ell=1$. Using these equation, one can readily verify that the computed output agrees with the required output for all inputs. Consequently, Algorithm~\ref{alg:X2L} produces the correct output whenever the input vertex is a leaf. Therefore, as $(V,E)$ is a full binary tree, it is sufficient to show that for all internal $v\in V$, if the algorithm produces the correct output whenever $v_\alpha$ or $v_\delta$ is given as an input, then it produces the correct output whenever $v$ is given as an input.

\begin{table}[h]
	\setlength{\belowcaptionskip}{0pt}
	\begin{tabular}{cccc@{}cccccc@{}cc}
		\toprule
		\multicolumn{4}{c}{Input} &  & \multicolumn{4}{c}{Required output}  & & \multicolumn{2}{c}{Computed output} \\
		\cmidrule{1-4} \cmidrule{6-9} \cmidrule{11-12}
		$c$ & $\ell$ &   $a_0$   &   $a_1$   &&& $a_0$ && $a_1$  &&                  $a_0$                  & $a_1$     \\
		\midrule
		$2$ &  $2$   & $\fout_0$ & $\fout_1$ &&& $f_0$ && $f_1$  && $\fout_0+\PushDown_v(v,\lambda)\fout_1$ & $\fout_1+(\fout_0+\PushDown_v(v,\lambda)\fout_1)$ \\
		$1$ &  $2$   & $\fout_0$ & $\fout_1$ &&& $f_0$ && $\ast$ && $\fout_0+\PushDown_v(v,\lambda)\fout_1$ & $\ast$    \\
		$2$ &  $1$   & $\fout_0$ &  $\ast$   &&& $f_0$ && $f_1$  &&                $\fout_0$                & $\fout_0$ \\
		$1$ &  $1$   & $\fout_0$ &  $\ast$   &&& $f_0$ && $\ast$ &&                $\fout_0$                & $\ast$    \\
		\bottomrule
	\end{tabular}
	\vspace{\abovecaptionskip}
	\caption{Required and computed outputs of Algorithm~\ref{alg:X2L} when $v$ is a leaf.}\label{tab:X2L-base-case}
\end{table}

Let $v\in V$ be an internal vertex and suppose that Algorithm~\ref{alg:X2L} produces the correct output whenever $v_\alpha$ or $v_\delta$ is given as an input. Suppose that the algorithm is called on $v$, the vector $(\PushDown_v(u,\lambda))_{u\in\leaves_v}$ for some $\lambda\in\F$, integers $c,\ell\in\{1,2,\dotsc,2^{n_v}\}$ and $(a_0,\dotsc,a_{2^{n_v}-1})$, with $a_i=\fout_i\in\F$ for $i\in\{0,\dotsc,\ell-1\}$. Then there exist unique elements $f_0,\dotsc,f_{2^{n_v}-1}\in\F$ such that~\eqref{eqn:LX} holds. In-turn, Lemma~\ref{lem:LX-reduction} implies that there exist unique elements $\fint_0,\dotsc,\fint_{2^{n_v}-1}\in\F$ such that~\eqref{eqn:LX-rows} and~\eqref{eqn:LX-cols} hold. 


Once again repeating arguments from the proof of Theorem~\ref{thm:NX-correctness} shows that $\nu=(\PushDown_{v_\delta}(u,\eta))_{u\in\leaves_{v_\delta}}$ for the recursive calls of Lines~\ref{X2L:left-cols}--\ref{X2L:right-cols-end}, $\mu=(\PushDown_{v_\alpha}(u,\lambda+\omega_{\gamma_v,i}))_{u\in\leaves_{v_\alpha}}$ each time the recursive call of Line~\ref{X2L:row} is performed, and finally $\mu=(\PushDown_{v_\alpha}(u,\lambda+\omega_{\gamma_v,c_1}))_{u\in\leaves_{v_\alpha}}$ for the recursive call of Line~\ref{X2L:last-row}. Thus, \eqref{eqn:LX-cols} and the assumption that the algorithm produces the correct output whenever $v_\delta$ is given as an input imply that  Lines~\ref{X2L:left-cols}--\ref{X2L:right-cols-end} set $a_{2^{\dg_v}i+j}=\fint_{2^{\dg_v}i+j}$ for $i\in\{0,\dotsc,c_1\}$ and $j\in\{0,\dotsc,\min(2^{\dg_v},\ell)-1\}$. Consequently, \eqref{eqn:LX-rows} and the assumption that the algorithm produces the correct output whenever $v_\alpha$ is given as an input imply that Lines~\ref{X2L:rows}--\ref{X2L:rows-end} set $a_i=f_i$ for $i\in\{0,\dotsc,2^{\dg_v}c_1-1\}$, and that Line~\ref{X2L:last-row} sets $a_i=f_i$ for $i\in\{2^{\dg_v}c_1,\dotsc,c-1\}$. Therefore, the algorithm terminates with $a_i=f_i$ for $i\in\{0,\dotsc,c-1\}$, as required. Hence, for internal $v\in V$, if the algorithm produces the correct output whenever $v_\alpha$ or $v_\delta$ is given as an input, then it produces the correct output whenever $v$ is given as an input.
\end{proof}

Algorithm~\ref{alg:X2L} requires the same precomputations as the algorithms of Sections~\ref{sec:NX} and~\ref{sec:L2X}, while the algorithm requires auxiliary space for $2^{n_v}-\max(c,\ell)+\bigO(n^2)$ field elements. Small values of $\dg_v$ should once again be avoided when choosing a reduction tree for the algorithm, in order to help reduce the number of additions performed by the updates to the vector $\mu$ in Line~\ref{X2L:update-shifts}.

\begin{example}\label{ex:X2L} For $\beta$ equal to a Cantor basis of dimension $15$, and inputs $\ell\in\{1,\dotsc,2^{15}\}$ and $c=\ell$,  Figure~\ref{fig:cantor_lx} also shows the maximum and minimum number of additions performed by Algorithm~\ref{alg:X2L} over all possible reduction trees for the basis, as well as the number of multiplications performed by the algorithm for all such trees. Thus, Algorithm~\ref{alg:L2X} with $c=\ell$ and $\flag=0$ performs the same number of operations for both extremes (see Example~\ref{ex:L2X}). As for Examples~\ref{ex:NX} and~\ref{ex:L2X}, the maximum and minimum number of additions performed for each $\ell$ are given respectively by the trees with $\dg_v=1$ and $\dg_v=2^{\ceil{\log_2 n_v}-1}$ for all internal $v\in V$.
\end{example}

\begin{theorem}\label{thm:X2L-complexity} Algorithm~\ref{alg:X2L} performs at most
\begin{equation*}
	\min
	\left(
		\frac{c-1}{2}
		\left(\ceil{\log_2c}-1\right)
		+\ell-1,
		2^{n_v-1}n_v
	\right)
\end{equation*}
multiplications in $\F$, and at most
\begin{equation*}
	\min
	\left(
		\frac{c-1}{2}
		\left(3\ceil{\log_2c}-1\right)
		+\ell-1,
		2^{n_v-1}\left(3n_v-2\right)+1
	\right)
\end{equation*}
additions in $\F$.
\end{theorem}

We split the proof of Theorem~\ref{thm:X2L-complexity} into four lemmas, one for each bound. It is readily verified that the bounds hold if the input vertex is a leaf. Therefore, as $(V,E)$ is a full binary tree, it is sufficient to show for each bound that if $v\in V$ is an internal vertex such that the bound holds whenever the input vertex is $v_\alpha$ or $v_\delta$, then the bound holds whenever $v$ is the input vertex.

\begin{lemma}\label{lem:X2L-multiplications} Algorithm~\ref{alg:X2L} performs at most $2^{n_v-1}n_v$ multiplications in $\F$.
\end{lemma}
\begin{proof} Suppose that Algorithm~\ref{alg:X2L} is called on an internal vertex $v\in V$ such that the bound of the lemma holds whenever the input vertex is $v_\alpha$ or $v_\delta$. Then Lines~\ref{X2L:left-cols}--\ref{X2L:right-cols-end} of the algorithm perform at most $\ell_2'2^{n_v-\dg_v-1}(n_v-\dg_v)\leq 2^{n_v-1}(n_v-\dg_v)$ multiplications, while Lines~\ref{X2L:rows}--\ref{X2L:last-row} perform at most $(c_1+1)2^{\dg_v-1}\dg_v\leq 2^{n_v-1}\dg_v$ multiplications. Summing theses contributions, it follows that Algorithm~\ref{alg:X2L} performs at most $2^{n_v-1}n_v$ multiplications.
\end{proof}

\begin{lemma} Algorithm~\ref{alg:X2L} performs at most $(c-1)(\ceil{\log_2 c}-1)/2+\ell-1$ multiplications in $\F$.
\end{lemma}
\begin{proof} Suppose that Algorithm~\ref{alg:X2L} is called on an internal vertex $v\in V$ such that the bound of the lemma holds whenever the input vertex is $v_\alpha$ or $v_\delta$. Then Lines~\ref{X2L:left-cols}--\ref{X2L:right-cols-end} perform at most
\begin{equation*}
	2^{\dg_v}
	\frac{c_1}{2}
	\left(\ceil{\log_2 c_1+1}-1\right)
	+\ell_2'\left(\ell_1-1\right)
	+\ell_2
	\leq
	\frac{c-c_2}{2}
	\left(\ceil{\log_2 c_1+1}-1\right)
	+\ell-\ell_2'
\end{equation*}
multiplications. Lemma~\ref{lem:X2L-multiplications} implies that Lines~\ref{X2L:rows}--\ref{X2L:last-row} perform at most
\begin{equation*}
	c_12^{\dg_v-1}\dg_v
	+\frac{c_2-1}{2}
	\left(\ceil{\log_2c_2}-1\right)
	+\ell_2'-1
	=\frac{c-c_2}{2}\dg_v
	+\frac{c_2-1}{2}
	\left(\ceil{\log_2c_2}-1\right)
	+\ell_2'-1
\end{equation*}
multiplications. As $c_2\leq c$, it follows that Algorithm~\ref{alg:X2L} performs at most
\begin{equation*}
	\frac{c-c_2}{2}
	\left(\ceil{\log_2 c_1+1}+\dg_v-1\right)
	+\frac{c_2-1}{2}
	\left(\ceil{\log_2c}-1\right)
	+\ell-1
\end{equation*}
multiplications. Hence, Algorithm~\ref{alg:X2L} performs at most $(c-1)(\ceil{\log_2 c}-1)/2+\ell-1$ multiplications, since $c_2=c$ if $c_1=0$, and $\ceil{\log_2c_1+1}=\ceil{\log_2 c}-\dg_v$ otherwise.
\end{proof}

\begin{lemma}\label{lem:X2L-additions} Algorithm~\ref{alg:X2L} performs at most $2^{n_v-1}(3n_v-2)+1$ additions in $\F$.
\end{lemma}
\begin{proof} Suppose that Algorithm~\ref{alg:X2L} is called on an internal vertex $v\in V$ such that the bound of the lemma holds whenever the input vertex is $v_\alpha$ or $v_\delta$. Then Lines~\ref{X2L:left-cols}--\ref{X2L:right-cols-end} of the algorithm perform at most
\begin{equation*}
	\ell_2'\left(2^{n_v-\dg_v-1}\left(3\left(n_v-\dg_v\right)-2\right)+1\right)
	\leq
	2^{n_v-1}\left(3n_v-3\dg_v\right)
	-2^{\dg_v}\left(2^{n_v-\dg_v}-1\right)
\end{equation*}
additions. Lines~\ref{X2L:rows}--\ref{X2L:last-row} perform at most
\begin{equation*}
	(c_1+1)\left(2^{\dg_v-1}(3\dg_v-2)+\dg_v+1\right)-\dg_v
	\leq 2^{n_v-1}\left(3\dg_v-2\right)+1
	+\left(\dg_v+1\right)
	\left(2^{n_v-\dg_v}-1\right)
\end{equation*}
additions. It follows that Algorithm~\ref{alg:X2L} performs at most
\begin{equation*}
	2^{n_v-1}\left(3n_v-2\right)+1
	-\left(2^{\dg_v}-\dg_v-1\right)
	\left(2^{n_v-\dg_v}-1\right)
	\leq
	2^{n_v-1}\left(3n_v-2\right)+1
\end{equation*}
additions.
\end{proof}

\begin{lemma} Algorithm~\ref{alg:X2L} performs at most $(c-1)(3\ceil{\log_2c}-1)/2+\ell-1$ additions in $\F$.
\end{lemma}
\begin{proof} Suppose that Algorithm~\ref{alg:X2L} is called on an internal vertex $v\in V$ such that the bound of the lemma holds whenever the input vertex is $v_\alpha$ or $v_\delta$. Then Lines~\ref{X2L:left-cols}--\ref{X2L:right-cols-end} perform at most
\begin{equation*}
	2^{\dg_v}\frac{c_1}{2}
	\left(3\ceil{\log_2 c_1+1}-1\right)
	+\ell_2'\left(\ell_1-1\right)
	+\ell_2
	\leq
	\frac{c-c_2}{2}
	\left(3\ceil{\log_2 c_1+1}-1\right)
	+\ell-\ell_2'
\end{equation*}
additions. Lemma~\ref{lem:X2L-additions} implies that Lines~\ref{X2L:rows}--\ref{X2L:rows-end} perform at most
\begin{equation*}
	c_1\left(2^{\dg_v-1}\left(3\dg_v-2\right)+1+\dg_v\right)
	=\frac{c-c_2}{2}3\dg_v
	-c_1\left(2^{\dg_v}-\dg_v-1\right)
	\leq
	\frac{c-c_2}{2}3\dg_v
\end{equation*}
additions. Line~\ref{X2L:last-row} performs at most
\begin{equation*}
	\frac{c_2-1}{2}
	\left(3\ceil{\log_2c_2}-1\right)
	+\ell_2'-1
\end{equation*}
additions. As $c_2\leq c$, it follows that Algorithm~\ref{alg:X2L} performs at most
\begin{equation*}
	\frac{c-c_2}{2}
	\left(3\ceil{\log_2 c_1+1}+3\dg_v-1\right)
	+\frac{c_2-1}{2}
	\left(3\ceil{\log_2c}-1\right)
	+\ell-1
\end{equation*}
additions. Hence, Algorithm~\ref{alg:X2L} performs at most $(c-1)(3\ceil{\log_2 c}-1)/2+\ell-1$ additions, since $c_2=c$ if $c_1=0$, and $\ceil{\log_2c_1+1}=\ceil{\log_2 c}-\dg_v$ otherwise.
\end{proof}

\subsection{Interlude: generalised Taylor expansion}\label{sec:taylor}

The generalised Taylor expansion of a polynomial $F\in\F[x]$ at a degree $t\geq 1$ polynomial $T\in\F[x]$, also called its $T$-adic expansion, is the series expansion
\begin{equation*}
	F=F_0+F_1T+F_2T^2+\dotsb
\end{equation*}
such that $F_i\in\F[x]_t$ for $i\in\N$. Gao and Mateer~\cite[Section~II]{gao2010} provide a fast algorithm for computing the coefficients of the Taylor expansion when $T=x^t-x$ with $t\geq 2$. The algorithm is then utilised as part of their additive FFT algorithms. Our algorithm for converting from the monomial basis to the LCH basis similarly relies on their generalised Taylor expansion algorithm. Consequently, we make a brief aside to recall their algorithm.

The algorithm of Gao and Mateer can be viewed as a specialisation of the recursive algorithm of von zur Gathen~\cite{gathen1990} that takes advantage of easy division by $(x^t-x)^{2^k}=x^{2^kt}-x^{2^k}$ in characteristic two. We present a nonrecursive version of their algorithm modelled on the basis conversion algorithms of van der Hoeven and Schost~\cite[Section~2.2]{hoeven2013}. We also present the inverse algorithm, which recovers a polynomial from the coefficients of its Taylor expansion at $x^t-x$, as it is required by our algorithm for converting from the LCH basis to monomial basis. Finally, we derive a bound on the complexity of both algorithms that is tighter than the one provided by Gao and Mateer.

Let $F\in\F[x]_\ell$ and $t\geq 2$ be an integer. For $k\in\N$, define $F_{k,0},F_{k,1},\dotsc\in\F[x]_{2^k t}$ by the equation
\begin{equation}\label{eqn:taylor}
	F
	=\sum_{i\in\N}
	F_{k,i}
	\left(x^t-x\right)^{2^ki}.
\end{equation}
Then $F_{0,0},F_{0,1},\dotsc$ are the coefficients of the Taylor expansion at $x^t-x$, while $F_{k,0}=F$ for $k\geq\ceil{\log_2\ceil{\ell/t}}$. By grouping terms of indices $2i$ and $2i+1$ in~\eqref{eqn:taylor}, it follows that
\begin{equation*}
	F
	=\sum_{i\in\N}
	\left(
		F_{k,2i}
		+F_{k,2i+1}
		\left(x^{2^kt}-x^{2^k}\right)
	\right)
	\left(x^t-x\right)^{2^{k+1}i}
	\quad\text{for $k\in\N$}.
\end{equation*}
Thus, we obtain the recursive formula
\begin{equation*}
	F_{k+1,i}
	=F_{k,2i}
	+x^{2^k}F_{k,2i+1}
	+x^{2^k t}F_{k,2i+1}
	\quad\text{for $k,i\in\N$}.
\end{equation*}
Given $F_{k,2i}$ and $F_{k,2i+1}$ on the monomial basis, the recursive formula allows $F_{k+1,i}$ to be readily computed on the monomial basis. The formula also allows this computation to be easily inverted. Therefore, given the Taylor coefficients $F_{0,0},F_{0,1},\dotsc$  on the monomial basis, we can efficiently compute $F=F_{\ceil{\log_2\ceil{\ell/t}},0}$ on the monomial basis by means of the recursive formula, and vice versa. Using this observation, we obtain Algorithms~\ref{alg:M2T} and~\ref{alg:T2M}.

\begin{algorithm}[h]
	\caption{$\mathsf{TaylorExpansion}(t,\ell,(a_0,\dotsc,a_{\ell-1}))$}\label{alg:M2T}
	\begin{algorithmic}[1]
		\Require Integers $t\geq 2$ and $\ell\geq 1$, and $a_i=f_i\in\F$ for $i\in\{0,\dotsc,\ell-1\}$.
		\Ensure $a_i=\ftay_i$ for $i\in\{0,\dotsc,\ell-1\}$ such that
		\begin{equation}\label{eqn:TM}
			\sum^{\ceil{\ell/t}-1}_{i=0}
			\left(
				\sum^{\min(\ell-ti,t)-1}_{j=0}\ftay_{ti+j}x^j
			\right)
			\left(x^t-x\right)^i
			=
			\sum^{\ell-1}_{i=0}f_ix^i.
		\end{equation}
		\For{$k=\ceil{\log_2\ceil{\ell/t}}-1,\dotsc,0$}
			\State%
				$\ell_1\set\floor{\ell/(2^{k+1}t)}$,
				$\ell_2\set\ell-2^{k+1}t\ell_1$
			\For{$i=0,\dotsc,\ell_1-1$}
				\For{$j=2^kt-1,\dotsc,0$}
					\State $a_{2^kt(2i)+2^k+j}\set a_{2^kt(2i)+2^k+j}+a_{2^kt(2i+1)+j}$
				\EndFor
			\EndFor
			\For{$j=\ell_2-2^kt-1,\dotsc,0$}
				\State $a_{2^kt(2\ell_1)+2^k+j}\set a_{2^kt(2\ell_1)+2^k+j}+a_{2^kt(2\ell_1+1)+j}$
			\EndFor
		\EndFor
	\end{algorithmic}
\end{algorithm}

\begin{algorithm}[h]
	\caption{$\mathsf{InverseTaylorExpansion}(t,\ell,(a_0,\dotsc,a_{\ell-1}))$}\label{alg:T2M}
	\begin{algorithmic}[1]
		\Require Integers $t\geq 2$ and $\ell\geq 1$, and $a_i=\ftay_i\in\F$ for $i\in\{0,\dotsc,\ell-1\}$.
		\Ensure $a_i=f_i$ $i\in\{0,\dotsc,\ell-1\}$ such that~\eqref{eqn:TM} holds.	
		\For{$k=0,\dotsc,\ceil{\log_2\ceil{\ell/t}}-1$}
			\State
				\label{T2M:main-loop-start}%
				$\ell_1\set\floor{\ell/(2^{k+1}t)}$,
				$\ell_2\set\ell-2^{k+1}t\ell_1$
			\For{$i=0,\dotsc,\ell_1-1$}
				\For{$j=0,\dotsc,2^kt-1$}
					\State $a_{2^kt(2i)+2^k+j}\set a_{2^k t(2i)+2^k+j}+a_{2^k t(2i+1)+j}$
				\EndFor
			\EndFor
			\For{$j=0,\dotsc,\ell_2-2^kt-1$}
				\State $a_{2^k t(2\ell_1)+2^k+j}\set a_{2^kt(2\ell_1)+2^k+j}+a_{2^k t(2\ell_1+1)+j}$
			\EndFor\label{T2M:main-loop-end}
		\EndFor
	\end{algorithmic}
\end{algorithm}

\begin{lemma}\label{lem:taylor} Algorithms~\ref{alg:M2T} and~\ref{alg:T2M} perform at most $\floor{\ell/2}\ceil{\log_2\ceil{\ell/t}}$ additions in~$\F$.
\end{lemma}
\begin{proof} For each $k\in\{0,\dotsc,\ceil{\log_2\ceil{\ell/t}}-1\}$, Lines~\ref{T2M:main-loop-start}--\ref{T2M:main-loop-end} of either algorithm perform at most
\begin{equation*}
	2^kt\ell_1+\max(\ell_2-2^kt,0)
	\leq 2^kt\ell_1+\max(\ell_2-\ceil{\ell_2/2},\floor{\ell_2/2})
	= 2^kt\ell_1+\floor{\ell_2/2}
	=\floor{\ell/2}
\end{equation*}
additions in $\F$.
\end{proof}

\subsection{Conversion between the Lin--Chung--Han and monomial bases}\label{sec:MX}

We use Lemma~\ref{lem:reduction} to provide algorithms for converting between the monomial basis and the ``twisted'' LCH basis $\{X_{\beta_v,0}(\beta_{v,0}x),\dotsc,X_{\beta_v,\ell-1}(\beta_{v,0}x)\}$ of~$\F[x]_\ell$, for $v\in V$ and $\ell\in\{1,\dotsc,2^{n_v}\}$. Conversions between the LCH and monomial bases then require at most an additional $\max(2\ell-3,0)$ multiplications for performing the substitution $x\mapsto x/\beta_{v,0}$ or $x\mapsto\beta_{v,0}x$. In particular, no additional multiplications are required if $\beta$ is a Cantor basis, since $\beta_{v,0}=1$ for all $v\in V$ (see Remark~\ref{rmk:cantor-precomputations}). We base the conversion algorithms on the following analogue of Lemma~\ref{lem:setup-LX-reduction}.

\begin{lemma}\label{lem:XM-reduction} Let $v\in V$ be an internal vertex and $\ell\in\{1,\dotsc,2^{n_v}\}$. Suppose that $\fout_0,\dotsc,\fout_{\ell-1},\lambda,\fint_0,\dotsc,\fint_{\ell-1},\ftay_0,\dotsc,\ftay_{\ell-1}\in\F$ satisfy
\begin{equation}\label{eqn:XM-rows}
	\sum^{\min\left(\ell-2^{\dg_v}i,2^{\dg_v}\right)-1}_{j=0}
	\fint_{2^{\dg_v}i+j}
	x^j
	=
	\sum^{\min\left(\ell-2^{\dg_v}i,2^{\dg_v}\right)-1}_{j=0}
	\fout_{2^{\dg_v}i+j}
	X_{\beta_{v_\alpha},j}\left(\beta_{v_\alpha,0}x\right)
\end{equation}
for $i\in\{0,\dotsc,\ceil{\ell/2^{\dg_v}}-1\}$, and
\begin{equation}\label{eqn:XM-cols}
	\sum^{\ceil{(\ell-j)/2^{\dg_v}}-1}_{i=0}
	\ftay_{2^{\dg_v}i+j}
	x^i
	=\sum^{\ceil{(\ell-j)/2^{\dg_v}}-1}_{i=0}
	\fint_{2^{\dg_v}i+j}
	X_{\beta_{v_\delta},i}\left(\beta_{v_\delta,0} x\right)
\end{equation}
for $j\in\{0,\dotsc,\min(2^{\dg_v},\ell)-1\}$. Then
\begin{equation}\label{eqn:XT}
	\sum^{\ceil{\ell/2^{\dg_v}}-1}_{i=0}
	\left(
		\sum^{\min\left(\ell-2^{\dg_v}i,2^{\dg_v}\right)-1}_{j=0}
		\ftay_{2^{\dg_v}i+j}
		x^j
	\right)
	\left(
		\frac{x^{2^{\dg_v}}-x}
		{\beta_{v_\delta,0}}
	\right)^i
	=\sum^{\ell-1}_{i=0}
	\fout_i
	X_{\beta_v,i}\left(\beta_{v,0}x\right).
\end{equation}
\end{lemma}
\begin{proof} Let $v\in V$ be an internal vertex and $\ell\in\{1,\dotsc,2^{n_v}\}$. Suppose that $\fout_0,\dotsc,\fout_{\ell-1},\lambda,\fint_0,\dotsc,\fint_{\ell-1},\ftay_0,\dotsc,\ftay_{\ell-1}\in\F$ satisfy equations~\eqref{eqn:XM-rows} and~\eqref{eqn:XM-cols}. Then $\beta_{v,i}/\beta_{v,0}\in\F_{2^{\dg_v}}$ for $i\in\{0,\dotsc,\dg_v-1\}$, since $(V,E)$ is a reduction tree for $\beta$. Thus, Lemma~\ref{lem:reduction} implies that
\begin{multline*}
	\sum^{\ell-1}_{i=0}
	\fout_i
	X_{\beta_v,i}\left(\beta_{v,0}x\right)
	=\sum^{\ceil{\ell/2^{\dg_v}}-1}_{i=0}
	\left(
		\sum^{\min\left(\ell-2^{\dg_v}i,2^{\dg_v}\right)-1}_{j=0}
		\fout_{2^{\dg_v}i+j}
		X_{\beta_{v_\alpha},j}\left(\beta_{v,0}x\right)
	\right)\\
	\times 
	X_{\beta_{v_\delta},i}\left(x^{2^{\dg_v}}-x\right).
\end{multline*}
Substituting in $\beta_{v,0}=\beta_{v_\alpha,0}$, \eqref{eqn:XM-rows} and~\eqref{eqn:XM-cols}, it follows that
\begin{align*}
	\sum^{\ell-1}_{i=0}
	\fout_i
	X_{\beta_v,i}\left(\beta_{v,0}x\right)
	&=
	\sum^{\ceil{\ell/2^{\dg_v}}-1}_{i=0}
	\left(
		\sum^{\min\left(\ell-2^{\dg_v}i,2^{\dg_v}\right)-1}_{j=0}
		\fint_{2^{\dg_v}i+j}
		x^j
	\right)
	X_{\beta_{v_\delta},i}\left(x^{2^{\dg_v}}-x\right)\\
	&=
	\sum^{\min\left(2^{\dg_v},\ell\right)-1}_{j=0}
	\left(
		\sum^{\ceil{(\ell-j)/2^{\dg_v}}-1}_{i=0}
		\fint_{2^{\dg_v}i+j}
		X_{\beta_{v_\delta},i}
		\left(
			\beta_{v_\delta,0}
			\frac{x^{2^{\dg_v}}-x}{\beta_{v_\delta,0}}
		\right)
	\right)
	x^j\\
	&=
	\sum^{\min\left(2^{\dg_v},\ell\right)-1}_{j=0}
	\left(
		\sum^{\ceil{(\ell-j)/2^{\dg_v}}-1}_{i=0}
		\ftay_{2^{\dg_v}i+j}
		\left(
			\frac{x^{2^{\dg_v}}-x}{\beta_{v_\delta,0}}
		\right)^i
	\right)
	x^j.
\end{align*}
Hence, \eqref{eqn:XT} holds.
\end{proof}

Using Lemma~\ref{lem:XM-reduction}, we obtain Algorithms~\ref{alg:X2M} and~\ref{alg:M2X} for converting between the monomial basis and the twisted basis $\{X_{\beta_v,0}(\beta_{v,0}x),\dotsc,X_{\beta_v,\ell-1}(\beta_{v,0}x)\}$ of $\F[x]_\ell$. Each algorithm makes what is now a familiar pattern of recursive calls, but with the addition of the computation of either a generalised Taylor expansion or the inverse transformation, for which the algorithms of Section~\ref{sec:taylor} are used.

\begin{algorithm}[h]
	\caption{$\mathsf{X2M}(v,\ell,(a_0,\ldots,a_{\ell-1}))$}\label{alg:X2M}
	\begin{algorithmic}[1]
		\Require a vertex $v\in V$, $\ell\in\{1,\dotsc,2^{n_v}\}$, and $a_i=\fout_i\in\F$ for $i\in\{0,\dotsc,\ell-1\}$.
		\Ensure $a_i=f_i\in\F$ for $i\in\{0,\dotsc,\ell-1\}$ such that
		\begin{equation}\label{eqn:XM}
			\sum^{\ell-1}_{i=0}
			f_ix^i
			=
			\sum^{\ell-1}_{i=0}
			\fout_i
			X_{\beta_v,i}\left(\beta_{v,0}x\right).
		\end{equation}
		\If{$\ell\leq 2$}
			\Return
		\EndIf
		\State\label{X2M:aux}%
			$\ell_1\set\ceil{\ell/2^{\dg_v}}-1$,
			$\ell_2\set \ell-2^{\dg_v}\ell_1$,
			$\ell_2'\set\min(2^{\dg_v},\ell)$
		\For{$i=0,\dotsc,\ell_1-1$}\label{X2M:rows}
			\State\Call{X2M}{$v_\alpha,2^{\dg_v},(a_{2^{\dg_v}i},a_{2^{\dg_v}i+1},\dotsc,a_{2^{\dg_v}(i+1)-1})$}
		\EndFor
		\State\label{X2M:last-row}\Call{X2M}{$v_\alpha,\ell_2,(a_{2^{\dg_v}\ell_1},a_{2^{\dg_v}\ell_1+1},\dotsc,a_{\ell-1})$}
		\For{$j=0,\dotsc,\ell_2-1$}\label{X2M:cols}
			\State\Call{X2M}{$v_\delta,\ell_1+1,(a_j,a_{2^{\dg_v}+j},\dotsc,a_{2^{\dg_v}\ell_1+j})$}
		\EndFor\label{X2M:left-cols-end}
		\For{$j=\ell_2,\dotsc,\ell_2'-1$}\label{X2M:right-cols}
			\State\Call{X2M}{$v_\delta,\ell_1,(a_j,a_{2^{\dg_v}+j},\dotsc,a_{2^{\dg_v}(\ell_1-1)+j})$}
		\EndFor\label{X2M:cols-end}
		\If{$\ell_1\neq 0$ and $1/\beta_{v_\delta,0}\neq 1$}\label{X2M:scale}
			\State $w\set 1/\beta_{v_\delta,0}$
			\For{$i=1,\dotsc,\ell_1-1$}
				\For{$j=0,\dotsc,2^{\dg_v}-1$}
					\State $a_{2^{\dg_v}i+j}\set wa_{2^{\dg_v}i+j}$
				\EndFor
				\State $w\set w/\beta_{v_\delta,0}$
			\EndFor
			\For{$j=0,\dotsc,\ell_2-1$}
				\State $a_{2^{\dg_v}\ell_1+j}\set wa_{2^{\dg_v}\ell_1+j}$
			\EndFor
		\EndIf\label{X2M:scale-end}
		\State\label{X2M:taylor}\Call{InverseTaylorExpansion}{$2^{\dg_v},\ell,(a_0,a_1,\dotsc,a_{\ell-1})$}
	\end{algorithmic}
\end{algorithm}

\begin{algorithm}[h]
	\caption{$\mathsf{M2X}(v,\ell,(a_0,\ldots,a_{\ell-1}))$}\label{alg:M2X}
	\begin{algorithmic}[1]
		\Require a vertex $v\in V$, $\ell\in\{1,\dotsc,2^{n_v}\}$, and $a_i=f_i\in\F$ for $i\in\{0,\dotsc,\ell-1\}$.
		\Ensure $a_i=\fout_i\in\F$ for $i\in\{0,\dotsc,\ell-1\}$ such that \eqref{eqn:XM} holds.
		\If{$\ell\leq 2$}
			\Return
		\EndIf
		\State\label{M2X:aux}%
			$\ell_1\set\ceil{\ell/2^{\dg_v}}-1$,
			$\ell_2\set \ell-2^{\dg_v}\ell_1$,
			$\ell_2'\set\min(2^{\dg_v},\ell)$
		\State\label{M2X:taylor}\Call{TaylorExpansion}{$2^{\dg_v},\ell,(a_0,a_1,\dotsc,a_{\ell-1})$}
		\If{$\ell_1\neq 0$ and $\beta_{v_\delta,0}\neq 1$}\label{M2X:scale}
			\State $w\set\beta_{v_\delta,0}$
			\For{$i=1,\dotsc,\ell_1-1$}
				\For{$j=0,\dotsc,2^{\dg_v}-1$}
					\State $a_{2^{\dg_v}i+j}\set wa_{2^{\dg_v}i+j}$
				\EndFor
				\State $w\set\beta_{v_\delta,0}w$
			\EndFor
			\For{$j=0,\dotsc,\ell_2-1$}
				\State $a_{2^{\dg_v}\ell_1+j}\set wa_{2^{\dg_v}\ell_1+j}$
			\EndFor
		\EndIf\label{M2X:scale-end}
		\For{$j=0,\dotsc,\ell_2-1$}\label{M2X:cols}
			\State\Call{M2X}{$v_\delta,\ell_1+1,(a_j,a_{2^{\dg_v}+j},\dotsc,a_{2^{\dg_v}\ell_1+j})$}
		\EndFor
		\For{$j=\ell_2,\dotsc,\ell_2'-1$}
			\State\Call{M2X}{$v_\delta,\ell_1,(a_j,a_{2^{\dg_v}+j},\dotsc,a_{2^{\dg_v}(\ell_1-1)+j})$}
		\EndFor\label{M2X:cols-end}
		\For{$i=0,\dotsc,\ell_1-1$}\label{M2X:rows}
			\State\Call{M2X}{$v_\alpha,2^{\dg_v},(a_{2^{\dg_v}i},a_{2^{\dg_v}i+1},\dotsc,a_{2^{\dg_v}(i+1)-1})$}
		\EndFor
		\State\label{M2X:last-row}\Call{M2X}{$v_\alpha,\ell_2,(a_{2^{\dg_v}\ell_1},a_{2^{\dg_v}\ell_1+1},\dotsc,a_{\ell-1})$}
	\end{algorithmic}
\end{algorithm}

\begin{theorem} Algorithms~\ref{alg:X2M} and~\ref{alg:M2X} are correct.
\end{theorem}
\begin{proof} We prove correctness for Algorithm~\ref{alg:X2M} by induction on $\ell$. The proof of correctness for Algorithm~\ref{alg:M2X} is omitted since it is almost identical. For $v\in V$, we have $X_{\beta_v,0}(\beta_{v,0}x)=1$ and $X_{\beta_v,1}(\beta_{v,0}x)=x$. Thus, Algorithm~\ref{alg:X2M} produces the correct output for all inputs with $\ell\leq 2$. In particular, it follows that the algorithm produces the correct output whenever the input vertex is a leaf. Therefore, it is sufficient to show that for internal $v\in V$, if the algorithm produces the correct output whenever $v_\alpha$ or $v_\delta$ is given as an input, then it produces the correct output whenever $v$ and $\ell\in\{3,\dotsc,2^{n_v}\}$ are given as inputs.

Let $v\in V$ be an internal vertex and suppose that Algorithm~\ref{alg:X2M} produces the correct output whenever $v_\alpha$ or $v_\delta$ is given as an input. Suppose that the algorithm is called on $v$ and $\ell\in\{3,\dotsc,2^{n_v}\}$, with $a_i=\fout_i\in\F$ for $i\in\{0,\dotsc,\ell-1\}$. Then the assumption that Algorithm~\ref{alg:X2M} produces the correct output whenever $v_\alpha$ is given as an input implies that Lines~\ref{X2M:aux}--\ref{X2M:last-row} of the algorithm set $a_i=\fint_i$ for $i\in\{0,\dotsc,\ell-1\}$, where $\fint_0,\dotsc,\fint_{\ell-1}$ are the unique elements in $\F$ such that \eqref{eqn:XM-rows} holds. Similarly, the assumption implies that Lines~\ref{X2M:cols}--\ref{X2M:cols-end} then set $a_i=\ftay_i$ for $i\in\{0,\dotsc,\ell-1\}$, where $\ftay_0,\dotsc,\ftay_{\ell-1}$ are the unique elements in $\F$ such that \eqref{eqn:XM-cols} holds. As $v$ is an internal vertex, Lemma~\ref{lem:XM-reduction} implies that $\ftay_0,\dotsc,\ftay_{\ell-1}$ also satisfy~\eqref{eqn:XT}.

Let $f_0,\dotsc,f_{\ell-1}$ be the unique elements in $\F$ such that \eqref{eqn:XM} holds. If ${\ell\leq 2^{\dg_v}}$, then \eqref{eqn:XT} and \eqref{eqn:XM} imply that $f_i=\ftay_i$ for $i\in\{0,\dotsc,\ell-1\}$. Moreover, Lines~\ref{X2M:scale}--\ref{X2M:taylor} have no effect in this case. Therefore, the algorithm produces the correct output if $\ell\leq 2^{\dg_v}$. If $\ell>2^{\dg_v}$, then Lines~\ref{X2M:scale}--\ref{X2M:scale-end} set
$a_{2^{\dg_v}i+j}=\ftay_{2^{\dg_v}i+j}/\beta^i_{v_\delta,0}$ for $i\in\{1,\dotsc,\ceil{\ell/2^{\dg_v}}-1\}$ and $j\in\{0,\dotsc,\min(\ell-2^{\dg_v}i,2^{\dg_v})-1\}$. Substituting into~\eqref{eqn:XT}, it follows that
\begin{equation*}
	\sum^{\ceil{\ell/2^{\dg_v}}-1}_{i=0}
	\left(
	\sum^{\min(\ell-2^{\dg_v}i,2^{\dg_v})-1}_{j=0}
	a_{2^{\dg_v}i+j}
	x^j
	\right)
	\left(x^{2^{\dg_v}}-x\right)^i
	=
	\sum^{\ell-1}_{i=0}
	\fout_i
	X_{\beta_v,i}\left(\beta_{v,0}x\right)
\end{equation*}
when \textsf{InverseTaylorExpansion} is called in Line~\ref{X2M:taylor}. Thus, the algorithm produces the correct output if $\ell>2^{\dg_v}$. Hence, for internal $v\in V$, if the algorithm produces the correct output whenever $v_\alpha$ or $v_\delta$ is given as an input, then it produces the correct output whenever $v$ and $\ell\in\{3,\dotsc,2^{n_v}\}$ are given as inputs.
\end{proof}

Algorithm~\ref{alg:M2X} requires the precomputation and storage of the elements $\beta_{v_\delta,0}$, while their inverses are required for Algorithm~\ref{alg:X2M}. Consequently, the algorithms require auxiliary storage for $\bigO(n)$ field elements, while all precomputations can be performed with $\bigO(n^2)$ field operations. If $\ell_1=\ceil{\ell/2^{\dg_v}}-1$ is nonzero, then Lines~\ref{X2M:scale}--\ref{X2M:scale-end} of Algorithm~\ref{alg:X2M} and Lines~\ref{M2X:scale}--\ref{M2X:scale-end} of Algorithm~\ref{alg:M2X} perform
\begin{equation*}
	\left(\ell_1-1\right)\left(2^{\dg_v}+1\right)+\ell_2
	=\ell+\ceil{\ell/2^{\dg_v}}-2^{\dg_v}-2
\end{equation*}
multiplications, while the calls made by the algorithms to either \textsf{TaylorExpansion} or \textsf{InverseTaylorExpansion} perform at most $\floor{\ell/2}\ceil{\log_2\ceil{\ell/2^{\dg_v}}}$ additions. It follows that we should once again aim to avoid small values of $\dg_v$ when choosing a reduction tree for the algorithms. However, compared to the algorithms for conversion between the LCH and the Newton and Lagrange bases, a much greater cost in terms of multiplications and additions is incurred if one fails to do so. If $\beta$ is a Cantor basis, then $\beta_{v_\delta,0}=1$ for all internal $v\in V$, regardless of the choice of reduction tree (see Remark~\ref{rmk:cantor-precomputations}). Thus, Algorithms~\ref{alg:X2M} and~\ref{alg:M2X} perform no multiplications in this case, and require no precomputations.

Lin et al.~\cite{lin2016a} provide two algorithms for converting from the monomial basis to the LCH basis when $\ell$ is a power of two, one for arbitrary bases and one for Cantor bases. The reduction strategy they apply for the arbitrary bases corresponds to reduction trees with $\image(\dg)\subseteq\{0,1\}$. For such reduction trees, Algorithm~\ref{alg:M2X} performs the same number of additions as their algorithm, but fewer multiplications in the recursive case (after equalising precomputations). For Cantor bases, Algorithm~\ref{alg:M2X} reduces to their algorithm by choosing the reduction tree so that $\dg_v=2^{\ceil{\log_2n_v}-1}$ for all internal $v\in V$.

\begin{theorem}\label{thm:XM-complexity} Algorithms~\ref{alg:X2M} and~\ref{alg:M2X} perform at most $\floor{\ell/2}\left(3\ceil{\log_2\ell}-4\right)+1$ multiplications and $\floor{\ell/2}\binom{\ceil{\log_2\ell}}{2}$ additions in $\F$. If $\dg_v=2^{\ceil{\log_2 n_v}-1}$ for all internal $v\in V$, then the algorithms perform at most $\floor{\ell/2}\ceil{\log_2 \ell}\ceil{\log_2\log_2\max(\ell,2)}$ additions in~$\F$. If $\beta$ is a Cantor basis, then the algorithms perform no multiplications.
\end{theorem}

We have already shown that Algorithms~\ref{alg:M2T} and~\ref{alg:T2M} perform no multiplications when $\beta$ is a Cantor basis. We split the remainder of the proof of Theorem~\ref{thm:XM-complexity} into three lemmas, one for each of three remaining bounds. It is clear that Algorithms~\ref{alg:X2M} and~\ref{alg:M2X} perform the same number of multiplications when given identical inputs. Consequently, we only prove the bounds for Algorithm~\ref{alg:X2M}. All three bounds are equal to zero or one for $\ell\leq 2$, while Algorithm~\ref{alg:X2M} performs no additions or multiplications for such input values of $\ell$. In particular, it follows that all three bound holds if the input vertex is a leaf. Consequently, for each of the three bounds it is sufficient to show that if $v\in V$ is an internal vertex such that the bound holds whenever the input vertex is $v_\alpha$ or $v_\delta$, then the bound holds whenever the input vertex is $v$ and $\ell\in\{3,\dotsc,2^{n_v}\}$.

\begin{lemma} Algorithms~\ref{alg:X2M} and~\ref{alg:M2X} perform at most $\floor{\ell/2}\left(3\ceil{\log_2\ell}-4\right)+1$ multiplications in $\F$.
\end{lemma}
\begin{proof} Suppose that for some internal vertex $v\in V$, Algorithm~\ref{alg:X2M} performs at most $\floor{\ell/2}(3\ceil{\log_2\ell}-4)+1$ multiplications in $\F$ whenever $v_\alpha$ or $v_\delta$ is given as the input vertex. Furthermore, suppose that $v$ and $\ell\in\{3,\dotsc,2^{n_v}\}$ are given as inputs to the algorithm. If $\ell_1=0$, then $\ell_2=\ell_2'=\ell$ and the algorithm performs at most
\begin{equation*}
	\floor{\frac{\ell_2}{2}}
	\left(3\ceil{\log_2\ell_2}-4\right)
	+1
	+\ell_2\times 0
	=
	\floor{\frac{\ell}{2}}
	\left(3\ceil{\log_2\ell}-4\right)
	+1
\end{equation*}	
multiplications. Therefore, suppose that $\ell_1\neq 0$. Then, as $\ell_2\leq 2^{\dg_v}$, Lines~\ref{X2M:rows}--\ref{X2M:last-row} of the algorithm perform at most
\begin{equation*}
	\ell_1
	\left(
		2^{\dg_v-1}
		\left(3\dg_v-4\right)
		+1
	\right)
	+
	\floor{\frac{\ell_2}{2}}
	\left(3\ceil{\log_2 \ell_2}-4\right)
	+1
	\leq
	\floor{\frac{\ell}{2}}
	\left(3\dg_v-4\right)
	+\ell_1+1
\end{equation*}
multiplications. Lines~\ref{X2M:cols}--\ref{X2M:cols-end} perform at most
\begin{multline*}
	\ell_2
	\floor{\frac{\ell+1}{2}}
	\left(3\ceil{\log_2\ell_1+1}-4\right)
	+\left(2^{\dg_v}-\ell_2\right)
	\floor{\frac{\ell_1}{2}}
	\left(3\ceil{\log_2\ell_1}-4\right)
	+2^{\dg_v}\\
	\begin{aligned}
		&\leq
		\floor{\frac{\ell_2(\ell+1)+(2^{\dg_v}-\ell_2)\ell_1}{2}}
		\left(3\ceil{\log_2\ell_1+1}-4\right)
		+2^{\dg_v}\\
		&=\floor{\frac{\ell}{2}}
		\left(3\ceil{\log_2\ell}-3\dg_v-4\right)
		+2^{\dg_v}
	\end{aligned}
\end{multline*}
multiplications, Lines~\ref{X2M:scale}--\ref{X2M:scale-end} perform $\left(\ell_1-1\right)\left(2^{\dg_v}+1\right)+\ell_2$ multiplications, and Line~\ref{X2M:taylor} performs no multiplications. Summing these bounds, it follows that Algorithm~\ref{alg:X2M} performs at most
\begin{multline*}
	\floor{\frac{\ell}{2}}
	\left(3\ceil{\log_2\ell}-4\right)
	+1
	-4\floor{\frac{\ell}{2}}
	+\left(2^{\dg_v}+2\right)\ell_1
	+\ell_2
	-1\\
	\begin{aligned}
		&\leq
		\floor{\frac{\ell}{2}}
		\left(3\ceil{\log_2\ell}-4\right)
		+1
		-2\ell
		+\left(2^{\dg_v}+2\right)\ell_1
		+\ell_2
		+1\\
		&=
		\floor{\frac{\ell}{2}}
		\left(3\ceil{\log_2\ell}-4\right)
		+1
		-\left(2^{\dg_v}-2\right)\ell_1
		-(\ell_2-1)\\
		&\leq
		\floor{\frac{\ell}{2}}
		\left(3\ceil{\log_2\ell}-4\right)
		+1
	\end{aligned}
\end{multline*}
multiplications.
\end{proof}

\begin{lemma} Algorithms~\ref{alg:X2M} and~\ref{alg:M2X} perform at most $\floor{\ell/2}\binom{\ceil{\log_2\ell}}{2}$ additions in $\F$.
\end{lemma}
\begin{proof} Suppose that for some internal vertex $v\in V$, Algorithm~\ref{alg:X2M} performs at most $\floor{\ell/2}\binom{\ceil{\log_2\ell}}{2}$ additions in $\F$ whenever $v_\alpha$ or $v_\delta$ is given as the input vertex. Furthermore, suppose that $v$ and $\ell\in\{3,\dotsc,2^{n_v}\}$ are given as inputs to the algorithm. If $\ell_1=0$, then $\ell_2=\ell_2'=\ell$ and the algorithm performs at most
\begin{equation*}
	\floor{\frac{\ell_2}{2}}
	\binom{\ceil{\log_2 \ell_2}}{2}
	+\ell_2\times 0
	+\floor{\frac{\ell}{2}}
	\ceil{\log_21}
	=\floor{\frac{\ell}{2}}
	\binom{\ceil{\log_2 \ell}}{2}
\end{equation*}	
additions. Therefore, suppose that $\ell_1>0$. Then, as $\ell_2\leq 2^{\dg_v}$, Lines~\ref{X2M:rows}--\ref{X2M:last-row} of the algorithm perform at most
\begin{equation*}
	\ell_12^{\dg_v-1}
	\binom{\dg_v}{2}
	+\floor{\frac{\ell_2}{2}}
	\binom{\ceil{\log_2\ell_2}}{2}
	\leq
	\ell_12^{\dg_v-1}
	\binom{\dg_v}{2}
	+\floor{\frac{\ell_2}{2}}
	\binom{\dg_v}{2}
	=
	\floor{\frac{\ell}{2}}
	\binom{\dg_v}{2}
\end{equation*}
additions. Lines~\ref{X2M:cols}--\ref{X2M:cols-end} of the algorithm  perform at most
\begin{equation*}
	\ell_2
	\floor{\frac{\ell+1}{2}}
	\binom{\ceil{\log_2\ell_1+1}}{2}
	+\left(2^{\dg_v}-\ell_2\right)
	\floor{\frac{\ell_1}{2}}
	\binom{\ceil{\log_2\ell_1}}{2}
	\leq
	\floor{\frac{\ell}{2}}
	\binom{\ceil{\log_2\ell_1+1}}{2}
\end{equation*}
additions, since $\ell_2(\ell+1)+(2^{\dg_v}-\ell_2)\ell_1=\ell$. Lines~\ref{X2M:scale}--\ref{X2M:scale-end} perform no additions, while Lemma~\ref{lem:taylor} implies that Line~\ref{X2M:taylor} performs at most $\floor{\ell/2}\ceil{\log_2\ceil{\ell/2^{\dg_v}}}=\floor{\ell/2}\ceil{\log_2\ell_1+1}$ additions. As $\ceil{\log_2\ell_1+1}=\ceil{\log_2\ell}-\dg_v$, it follows by summing these bounds that Algorithm~\ref{alg:X2M} performs at most
\begin{equation*}
	\floor{\frac{\ell}{2}}
	\left(
		\binom{\ceil{\log_2\ell}}{2}
		-\ceil{\log_2\ell_1+1}
		\left(\dg_v-1\right)
	\right)\\
	\leq
	\floor{\frac{\ell}{2}}
	\binom{\ceil{\log_2\ell}}{2}
\end{equation*}
additions.
\end{proof}

\begin{lemma} Suppose that $\dg_v=2^{\ceil{\log_2 n_v}-1}$ for all internal $v\in V$. Then Algorithms~\ref{alg:X2M} and~\ref{alg:M2X} perform at most $\floor{\ell/2}\ceil{\log_2 \ell}\ceil{\log_2\log_2\max(\ell,2)}$ additions in~$\F$.
\end{lemma}
\begin{proof} Suppose that $\dg_v=2^{\ceil{\log_2 n_v}-1}$ for all internal vertices $v\in V$. Furthermore, suppose that for some internal vertex $v\in V$, Algorithm~\ref{alg:X2M} performs at most $\floor{\ell/2}\ceil{\log_2 \ell}\ceil{\log_2\log_2\max(\ell,2)}$ additions in $\F$ whenever $v_\alpha$ or $v_\delta$ is given as the input vertex. Finally, suppose that $v$ and $\ell\in\{3,\dotsc,2^{n_v}\}$ are given as inputs to the algorithm. If $\ell_1=0$, then $\ell_2=\ell_2'=\ell$ and the algorithm performs at most
\begin{equation*}
	\floor{\frac{\ell_2}{2}}
	\ceil{\log_2 \ell_2}
	\ceil{\log_2\log_2\ell_2}
	+\ell_2\times 0
	+\floor{\frac{\ell}{2}}
	\ceil{\log_21}
	=\floor{\frac{\ell}{2}}
	\ceil{\log_2 \ell}
	\ceil{\log_2\log_2\ell}
\end{equation*}	
additions. Therefore, suppose that $\ell_1>0$. Then, as $\ell_2\leq 2^{\dg_v}<\ell$, Lines~\ref{X2M:rows}--\ref{X2M:last-row} of the algorithm perform at most
\begin{equation}\label{eqn:XM-cantor-rows-adds}
	\ell_12^{\dg_v-1}
	\dg_v\ceil{\log_2\log_2\ell}
	+
	\floor{\frac{\ell_2}{2}}
	\dg_v\ceil{\log_2\log_2\ell}
	=
	\floor{\frac{\ell}{2}}
	\dg_v\ceil{\log_2\log_2\ell}
\end{equation}
additions. Lines~\ref{X2M:cols}--\ref{X2M:left-cols-end} of the algorithm perform at most
\begin{equation*}
	\ell_2
	\floor{\frac{\ell_1+1}{2}}
	\ceil{\log_2 \ell_1+1}
	\ceil{\log_2\log_2\ell_1+1}
\end{equation*}
additions, while Lines~\ref{X2M:right-cols}--\ref{X2M:cols-end} perform at most
\begin{equation*}
	\left(2^{\dg_v}-\ell_2\right)
	\floor{\frac{\ell_1}{2}}
	\ceil{\log_2 \ell_1}
	\ceil{\log_2\log_2\max\left(\ell_1,2\right)}
\end{equation*}
additions. As $\ell_2(\ell_1+1)+(2^{\dg_v}-\ell_2)\ell_1=\ell$, it follows that Lines~\ref{X2M:cols}--\ref{X2M:cols-end} of the algorithm perform at most $\floor{\ell/2}\ceil{\log_2 \ell_1+1}\ceil{\log_2\log_2\ell_1+1}$ additions. If $\ell_1\geq 2$, then there exists an integer $k\geq 1$ such that $2^{2^{k-1}}<\ell_1+1\leq 2^{2^k}$. Then $\ceil{\log_2\log_2\ell_1+1}=k$, $2^{2^{k-1}}<2^{n_v-\dg_v}\leq 2^{\dg_v}$ and $\ell=2^{\dg_v}(\ell_1+\ell_2/2^{\dg_v})>2^{\dg_v+2^{k-1}}>2^{2^k}$. Thus, $\ceil{\log_2\log_2\ell_1+1}\leq\ceil{\log_2\log_2\ell}-1$ if $\ell_1\geq 2$. As $\ell\geq 3$, the inequality also holds if $\ell_1=1$. Therefore, Lines~\ref{X2M:cols}--\ref{X2M:cols-end} of the algorithm perform at most
\begin{equation}\label{eqn:XM-cantor-cols-adds}
	\floor{\frac{\ell}{2}}
	\left(\ceil{\log_2 \ell}-\dg_v\right)
	\ceil{\log_2\log_2\ell}
	-\floor{\frac{\ell}{2}}
	\ceil{\log_2 \ell_1+1}
\end{equation}
additions. Lines~\ref{X2M:scale}--\ref{X2M:scale-end} of the algorithm perform no additions, while Lemma~\ref{lem:taylor} implies that Line~\ref{X2M:taylor} performs at most $\floor{\ell/2}\ceil{\log_2\ceil{\ell/2^{\dg_v}}}=\floor{\ell/2}\ceil{\log_2\ell_1+1}$ additions. By combining this last bound with the bounds~\eqref{eqn:XM-cantor-rows-adds} and~\eqref{eqn:XM-cantor-cols-adds} on the number of additions performed by Lines~\ref{X2M:rows}--\ref{X2M:last-row} and Lines~\ref{X2M:cols}--\ref{X2M:cols-end}, it follows that Algorithm~\ref{alg:X2M} performs at most $\floor{\ell/2}\ceil{\log_2 \ell}\ceil{\log_2\log_2\ell}$ additions, which is the required bound since $\ell\geq 3$.
\end{proof}

\section{Constructing a basis and reduction tree}\label{sec:construction}

Let $\beta=(\beta_0,\dotsc,\beta_{n-1})\in\F^n$ have entries that are linearly independent over~$\F_2$. If $n=1$, then there exists a unique reduction tree for $\beta$, the tree consisting of a single vertex. If $n>1$, then a full binary tree is a reduction tree for $\beta$ if it has $n$ leaves, the subtrees rooted on the children $r_\alpha$ and $r_\delta$ of the tree's root vertex $r$ are themselves reduction trees for $\alpha(\beta,\dg(r))$ and $\delta(\beta,\dg(r))$, respectively, and the quotients $\beta_0/\beta_0,\dotsc,\beta_{\dg(r)-1}/\beta_0$ belong to $\F_{2^{\dg(r)}}$. The requirement on the quotients is trivially satisfied if $\dg(r)=1$. Consequently, the full binary tree with $n$ leaves and $\image(\dg)\subseteq\{0,1\}$ is a reduction tree for $\beta$ (Proposition~\ref{prop:trivial-reduction-tree}). We view this tree as the trivial choice of reduction tree for the basis, and as capturing the approach used by existing algorithms. We expect such trees to yield the worst algebraic complexity for the algorithms of Section~\ref{sec:algorithms}. Accordingly, when we have freedom to choose the basis vector, our choice should enable us to avoid their use. Cantor bases provide such a choice, and we witnessed the benefits they provide in Section~\ref{sec:algorithms}. However, Cantor bases are restricted to extensions with degree divisible by sufficiently large powers of two. In this section, we propose new basis constructions that allow us to benefit similarly in other extensions.

Regardless of the chosen basis vector, the choice of reduction trees is limited by the subfield structure of $\F$.

\begin{proposition}\label{prop:subfields} If $(V,E)$ is a reduction tree for some vector in $\F^n$, then
\begin{equation*}
	 \dg(v_\alpha)<\dg(v)<n\leq [\F:\F_2]
	 \quad\text{and}\quad
	 \max\left(\dg(v_\alpha),1\right)\mid\dg(v)\mid[\F:\F_2]
\end{equation*}
for all internal $v\in V$.
\end{proposition}
\begin{proof} Suppose that $(V,E)$ is a reduction tree for some vector $\beta\in\F^n$. Then $n\leq[\F:\F_2]$ since Definition~\ref{def:reduction-tree} requires $\beta$ to have linearly independent entries over $\F_2$. The definition also requires the tree to have $n$ leaves. Thus,
\begin{equation*}
	 \dg(v_\alpha)
	 <\left|\leaves_{v_\alpha}\right|
	 =\dg(v)
	 <\left|\leaves_v\right|
	 \leq n
	 \leq[\F:\F_2]
\end{equation*}
for all internal $v\in V$.

Let $v\in V$ be an internal vertex. Then it follows from Definition~\ref{def:reduction-tree} that the subtree rooted on $v$ is a reduction tree for some vector $\beta_v=(\beta_{v,0},\dotsc,\beta_{v,\left|\leaves_v\right|-1})\in\F^{\left|\leaves_v\right|}$ that has linearly independent entries over $\F_2$. Moreover, as $\left|\leaves_v\right|>1$, the definition implies that the quotients $\beta_{v,0}/\beta_{v,0},\dotsc,\beta_{v,\dg(v)-1}/\beta_{v,0}$ belong to $\F_{2^{\dg(v)}}$. These quotients inherit linear independence over $\F_2$. Thus, they form a basis of the extension $\F_{2^{\dg(v)}}/\F_2$. As the quotients also belong to $\F$, it follows that $\F_{2^{\dg(v)}}$ is a subfield of $\F$. Therefore, $\dg(v)$ divides $[\F:\F_2]$. Similarly, if the vertex is $v_\alpha$ is an internal vertex, then $\F_{2^{\dg(v_\alpha)}}$ is a subfield of $\F_{2^{\dg(v)}}$, since the subtree rooted on $v_\alpha$ is a reduction tree for $\alpha(\beta_v,\dg(v))=(\beta_{v,0},\dotsc,\beta_{v,\dg(v)-1})$. As $v_\alpha$ is an internal vertex if and only if $\dg(v_\alpha)\geq 1$, it follows that $\max\left(\dg(v_\alpha),1\right)$ divides $\dg(v)$.
\end{proof}

\begin{corollary}\label{cor:subfields} Suppose that the entries of $\beta\in\F^n$ are linearly independent over~$\F_2$, and $[\F:\F_2]$ has no proper factor less than $n$. Then a full binary tree is a reduction tree for $\beta$ if and only if it has $n$ leaves and $\image(\dg)\subseteq\{0,1\}$.
\end{corollary}
\begin{proof} Proposition~\ref{prop:trivial-reduction-tree} implies that it is sufficient to have $n$ leaves and $\image(\dg)\subseteq\{0,1\}$, while Proposition~\ref{prop:subfields} implies that it is also necessary.
\end{proof}

\subsection{A construction for arbitrary fields}\label{sec:tower}

It follows from Proposition~\ref{prop:subfields} that a path in a reduction tree that consists of two or more edges of the form $\{v,v_\alpha\}$ admits a nontrivial tower of subfields of $\F$. However, the existence of a basis vector of a prescribed dimension that has a nontrivial reduction tree is not guaranteed by the existence of nontrivial tower of subfields. Indeed, Corollary~\ref{cor:subfields} shows that it is necessary for the tower to contain a subfield other than $\F_2$ of degree bounded by the dimension. In this section, we show that this requirement is also sufficient.

\begin{theorem}\label{thm:sufficient} Suppose there exists a tower of subfields
\begin{equation}\label{eqn:tower}
	\F_2=\F_{2^{n_0}}\subset\F_{2^{n_1}}\subset\dotsb\subset\F_{2^{n_m}}=\F.
\end{equation}
Let $\{\vartheta_{k,0},\dotsc,\vartheta_{k,n_{k+1}/n_k-1}\}$ be a basis of $\F_{2^{n_{k+1}}}/\F_{2^{n_k}}$ for $k\in\{0,\dotsc,m-1\}$, and
\begin{equation*}
	\beta_i=\prod^{m-1}_{k=0}\vartheta_{k,i_k}
	\quad\text{such that}\quad
	\sum^{m-1}_{k=0}i_kn_k=i
\end{equation*}
for $i\in\{0,\dotsc,n_m-1\}$. Then $\beta_0,\dotsc,\beta_{n_m-1}\in\F$ are linearly independent over~$\F_2$. Moreover, a full binary tree $(V,E)$ with $n\leq n_m$ leaves is a reduction tree for $(\beta_0,\dotsc,\beta_{n-1})$ if $\image(\dg)\subseteq\{0,n_0,\dotsc,n_{m-1}\}$ and $\dg(v_\delta)\leq\dg(v)$ for all internal $v\in V$.
\end{theorem}

The requirements of Theorem~\ref{thm:sufficient} are satisfied by the full binary tree with $n$ leaves and $\image(\dg)\subseteq\{0,1\}$. Consequently, Proposition~\ref{prop:trivial-reduction-tree} follows from the case $m=1$. We delay the proof of the theorem until the end of the section. Instead, we now show that the basis vectors given by the construction of the theorem allow a nontrivial and, more importantly, beneficial choice of reduction trees.

\begin{proposition}\label{prop:max-tree-sufficient} If $I\subseteq\N$ and $(V,E)$ is a full binary tree such that
\begin{equation*}
	\dg(v)=\max\left\{i\in I\mid i<\left|\leaves_v\right|\right\}
\end{equation*}
for all internal $v\in V$, then $\dg(v_\delta)\leq\dg(v)$ for all internal $v\in V$.
\end{proposition}
\begin{proof} Suppose that $I\subseteq\N$ and a full binary tree $(V,E)$ satisfy the conditions of the proposition. Then $\dg(v_\delta)<\left|\leaves_{v_\delta}\right|<\left|\leaves_v\right|$ and $\dg(v_\delta)\in I\cup\{0\}$ for all internal $v\in V$. Hence, $\dg(v_\delta)\leq\max\{i\in I\mid i<\left|\leaves_v\right|\}=\dg(v)$ for all internal $v\in V$.
\end{proof}

\newcommand{\maxtree}[2]{T^{{\scriptscriptstyle(}#1{\scriptscriptstyle)}}_{#2}}

For tuples of positive integers $(n_0,\dotsc,n_m)$ such that~\eqref{eqn:tower} holds, let $\maxtree{n_0,\dotsc,n_m}{n}$ denote the full binary tree with $n$ leaves and $\dg(v)=\max\{n_k\mid n_k<\left|\leaves_v\right|\}$ for all internal vertices $v$. If $n_1<n\leq n_m$, then the root vertex~$r$ of $\maxtree{n_0,\dotsc,n_m}{n}$ satisfies $\dg(r)\geq n_1>1$, establishing the existence of a tree with $\image(\dg)\nsubseteq\{0,1\}$ that satisfies the conditions of Theorem~\ref{thm:sufficient}. Moreover, we expect this tree to approximately minimise the algebraic complexity of the conversion algorithms of Section~\ref{sec:algorithms} over all trees that satisfy the conditions of the theorem.

\begin{example}\label{ex:tower} Suppose that $\F=\F_{2^{12}}$. Then there are eight tuples of positive integers $(n_0,\dotsc,n_m)$ such that~\eqref{eqn:tower} holds. For each such tuple, Figure~\ref{fig:relative-adds} displays the relative number of additions performed by the basis conversion algorithms of Section~\ref{sec:algorithms} for $\beta=(\beta_0,\dotsc,\beta_{11})$ given by the constructed of Theorem~\ref{thm:sufficient} (the choice of the bases for the extensions $\F_{2^{n_{k+1}}}/\F_{2^{n_k}}$ doesn't matter here), the reduction tree $\maxtree{n_0,\dotsc,n_m}{12}$, and polynomial length ($\ell$) ranging over $\{1,\dotsc,2^{12}\}$. The number of additions performed in each case is given as a fraction of the number performed for the tuple $(1,12)$, which is taken to be one if zero additions are performed for both tuples. The reduction tree that corresponds to the tuple $(1,12)$ has $\image(\dg)=\{0,1\}$. Thus, it represents the complexity obtained with the reduction strategy of existing algorithms. The additional parameters $c=\ell$ and $\flag=0$ are used for Algorithm~\ref{alg:L2X}, and $c=\ell$ is used for Algorithm~\ref{alg:X2L}. Figure~\ref{fig:relative-mults} similarly displays the relative number of multiplications performed by Algorithms~\ref{alg:X2M} and~\ref{alg:M2X}, under the assumption that $\beta_{v_\delta,0}$ is never equal to one in Line~\ref{X2M:scale} of Algorithm~\ref{alg:X2M} and Line~\ref{M2X:scale} of Algorithm~\ref{alg:M2X}. The daggered tuples that appear in the figure are discussed in the next section. 
\begin{figure}
	\includegraphics{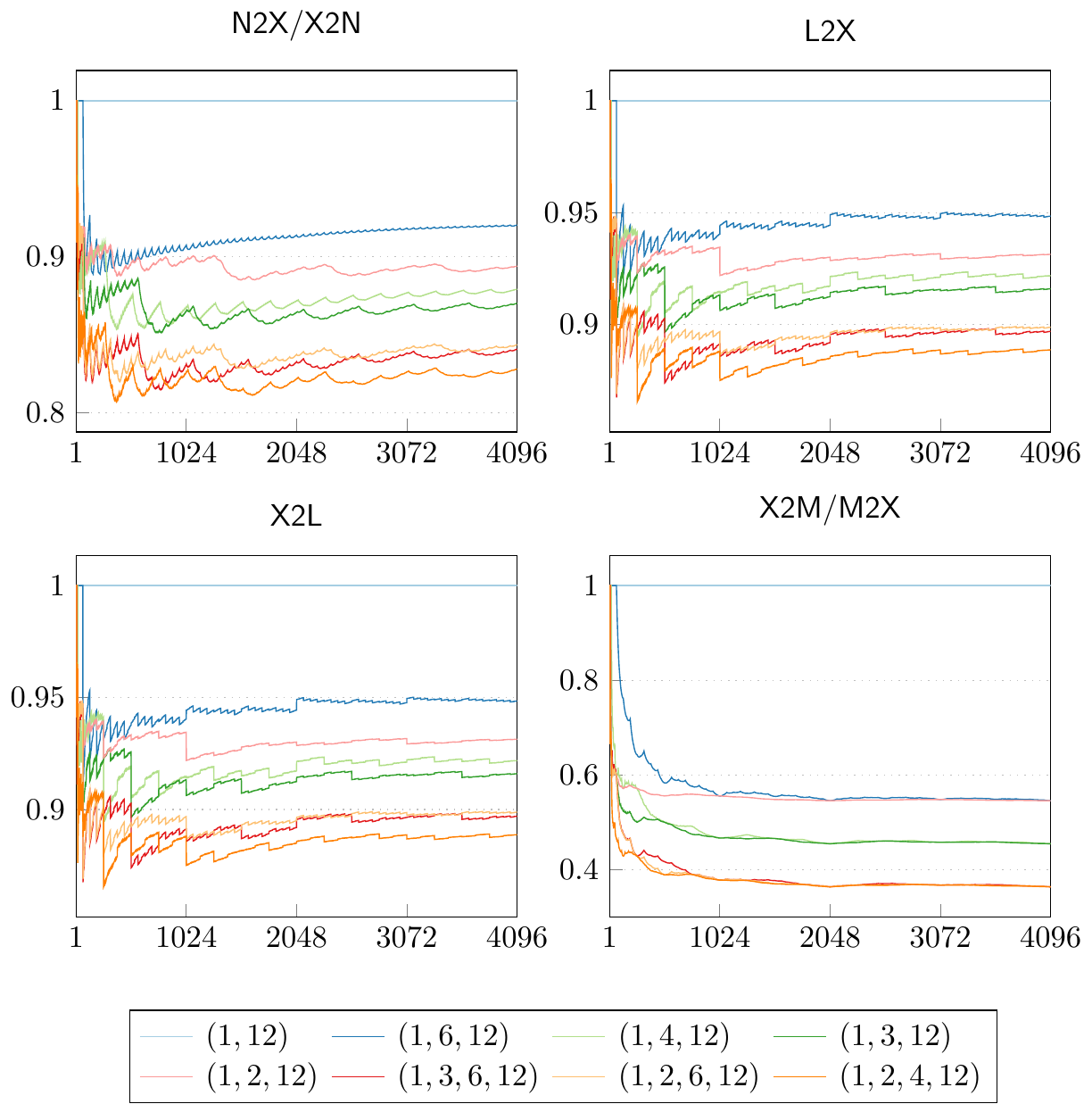}
	\caption{Relative number of additions performed by the algorithms of Section~\ref{sec:algorithms} (see Example~\ref{ex:tower}).}
	\label{fig:relative-adds}
\end{figure}
\begin{figure}
	\includegraphics{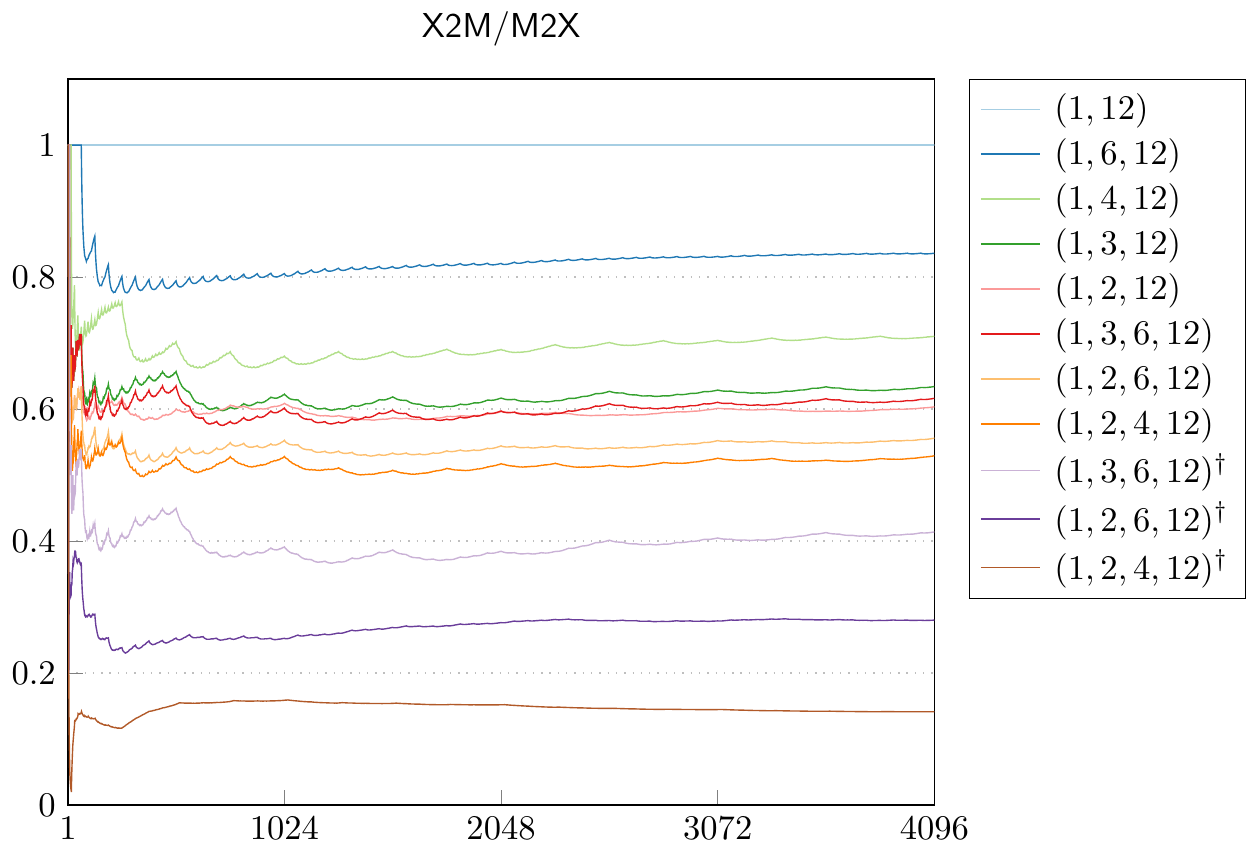}
	\caption{Relative number of multiplications performed by Algorithms~\ref{alg:X2M} and~\ref{alg:M2X} (see Example~\ref{ex:tower}).}
	\label{fig:relative-mults}
\end{figure}
\end{example}

Example~\ref{ex:tower} demonstrates that the construction of Theorem~\ref{thm:sufficient} and the choice of reduction trees it provides allows us to achieve a lower algebraic complexity if $[\F:\F_2]$ has even a single sufficiently small factor. The potential benefits are greater still when the degree of the field contains many small prime factors, echoing the benefits obtained by using roots of unity with smooth order in multiplicative FFTs. While a reduction in algebraic complexity is certainly desirable, it is not the only consideration in practice. For example, Harvey's ``cache-friendly'' variant~\cite{harvey2009} of the radix-2 truncated Fourier transform~\cite{hoeven2004} obtains better practical performance by optimising cache effects. This variant employs a reduction strategy that more rapidly reduces to problems of size that fit into cache, helping to reduce data exchanges with RAM. An analogous approach for the algorithms of Section~\ref{sec:algorithms} is to use reduction trees that balance the size of $\left|\leaves_{v_\alpha}\right|$ and $\left|\leaves_{v_\delta}\right|$ for internal vertices. If $[\F:\F_2]$ is smooth and the construction of Theorem~\ref{thm:sufficient} is applied with the number of subfields in the tower taken as large as possible, then the following proposition shows that it is possible to construct such trees while meeting requirements of the theorem by choosing $\dg(v)\in\{n_0,\dotsc,n_m\}$ to minimise $\max(n_k,\left|\leaves_v\right|-n_k)$ for all internal vertices. While there are trade-offs between optimising cache effects and reducing the operation count to be considered in practice, it appears that Theorem~\ref{thm:sufficient} offers us some freedom, especially when the degree of the field is smooth, to tune the algorithms of Section~\ref{sec:algorithms}. 

\begin{proposition} If $I\subseteq\N$ and $(V,E)$ is a full binary tree such that
\begin{equation*}
	\dg(v)\in\argmin_{i\in I}\max\left(i,\left|\leaves_v\right|-i\right)
\end{equation*}
for all internal $v\in V$, then $\dg(v_\delta)\leq\dg(v)$ for all internal $v\in V$.
\end{proposition}
\begin{proof} We prove the proposition by contradiction. Suppose that $I\subseteq\N$ and a full binary tree $(V,E)$ satisfy the conditions of the proposition. Furthermore, suppose there exists an internal vertex $v\in V$ such that $\dg(v_\delta)>\dg(v)$. Then $v_\delta$ is an internal vertex since $\dg(v_\delta)>1$. Thus, $\dg(v_\delta)\in I$ and
\begin{align*}
	\max\left(\dg(v_\delta),\left|\leaves_v\right|-\dg(v_\delta)\right)
	&<\max\left(\left|\leaves_{v_\delta}\right|,\left|\leaves_v\right|-\dg(v)\right)\\
	&=\max\left(\left|\leaves_v\right|-\left|\leaves_{v_\alpha}\right|,\left|\leaves_v\right|-\dg(v)\right)\\
	&=\max\left(\left|\leaves_v\right|-\dg(v),\left|\leaves_v\right|-\dg(v)\right)\\
	&\leq\max\left(\dg(v),\left|\leaves_v\right|-\dg(v)\right),
\end{align*}
which contradicts the minimality of $\max\left(\dg(v),\left|\leaves_v\right|-\dg(v)\right)$.	
\end{proof}

We now turn our attention to the proof of Theorem~\ref{thm:sufficient}, which we obtain as a consequence of the following more general result.

\begin{lemma}\label{lem:sufficient} Suppose there exists a tower of subfields $\F_2=\F_{2^{n_0}}\subset\dotsb\subset\F_{2^{n_m}}=\F$. Let $\beta=(\beta_0,\dotsc,\beta_{n-1})\in\F^n$ have entries that are linearly independent over $\F_2$, and~$(V,E)$ be a full binary tree with $n$ leaves and root vertex $r\in V$. Then $(V,E)$ is a reduction tree for $\beta$ if the following conditions are satisfied:
\begin{enumerate}
	\item\label{sufficient-image} $\image(\dg)\subseteq\{0,n_0,\dotsc,n_{m-1}\}$,
	\item\label{sufficient-delta} $\dg(v_\delta)\leq\dg(v)$ for all internal $v\in V$, and
	\item\label{sufficient-quotients} $\beta_i/\beta_{n_k\floor{i/n_k}}\in\F_{2^{n_k}}$ for $i\in\{0,\dotsc,n-1\}$ and $k\in\{0,\dotsc,m-1\}$ such that $n_k\leq\dg(r)$.
\end{enumerate}
\end{lemma}
\begin{proof} We prove the lemma by induction on $n$. The lemma holds trivially if $n=1$, since it is sufficient for $(V,E)$ to have $n$ leaves in this case. Therefore, let $n\geq 2$ and suppose that the lemma is true for all smaller values of $n$. Suppose that $\beta=(\beta_0,\dotsc,\beta_{n-1})\in\F^n$ and a full binary tree $(V,E)$ satisfy the conditions of the lemma. Then the root vertex $r\in V$ of the tree is not a leaf since $\left|\leaves_r\right|=n\geq 2$. Thus, \eqref{sufficient-image} implies that $\dg(r)=n_\ell$ for some $\ell\in\{0,\dotsc,m-1\}$ such that $n_\ell<n$.

Let $\alpha(\beta,\dg(r))=(\alpha_0,\dotsc,\alpha_{n_\ell-1})$ and $\delta(\beta,\dg(r))=(\delta_0,\dotsc,\delta_{n-n_\ell-1})$, which have linearly independent entries over $\F_2$ by Lemma~\ref{lem:reduction}. Then, as $\dg(r_\alpha)<\dg(r)$, \eqref{sufficient-quotients} implies that $\alpha_i/\alpha_{n_k\floor{i/n_k}}=\beta_i/\beta_{n_k\floor{i/n_k}}\in\F_{2^{n_k}}$ for $i\in\{0,\dotsc,n_\ell-1\}$ and $k\in\{0,\dotsc,m-1\}$ such that $n_k\leq\dg(r_\alpha)$. Moreover, as $n_0,\dotsc,n_\ell$ divide $n_\ell$, we have
\begin{equation*}
	\frac{\beta_{n_\ell+i}}{\beta_0}
	=\frac{\beta_{n_\ell+i}}{\beta_{n_\ell+n_k\floor{i/n_k}}}
	\frac{\beta_{n_\ell+n_k\floor{i/n_k}}}{\beta_0}
	=\frac{\beta_{n_\ell+i}}{\beta_{n_k\floor{(n_\ell+i)/n_k}}}
	\frac{\beta_{n_\ell+n_k\floor{i/n_k}}}{\beta_0},
\end{equation*}
where $\beta_{n_\ell+i}/\beta_{n_k\floor{(n_\ell+i)/n_k}}\in\F_{2^{n_k}}\subseteq\F_{2^{n_\ell}}$, for $i\in\{0,\dotsc,n-n_\ell-1\}$ and $k\in\{0,\dotsc,\ell\}$. It follows that
\begin{align*}
	\delta_i
	&=\frac{\beta_{n_\ell+i}}{\beta_{n_k\floor{(n_\ell+i)/n_k}}}
	\left(\left(\frac{\beta_{n_\ell+n_k\floor{i/n_k}}}{\beta_0}\right)^{2^{n_\ell}}-\frac{\beta_{n_\ell+n_k\floor{i/n_k}}}{\beta_0}\right)\\
	&=\frac{\beta_{n_\ell+i}}{\beta_{n_k\floor{(n_\ell+i)/n_k}}}
	\delta_{n_k\floor{i/n_k}}
\end{align*}
for $i\in\{0,\dotsc,n-n_\ell-1\}$ and $k\in\{0,\dotsc,\ell\}$. As $\dg(r_\delta)\leq\dg(r)=n_\ell$ by \eqref{sufficient-delta}, it follows that $\delta_i/\delta_{n_k\floor{i/n_k}}\in\F_{2^{n_k}}$ for $i\in\{0,\dotsc,n-n_\ell-1\}$ and $k\in\{0,\dotsc,m-1\}$ such that $n_k\leq\dg(r_\delta)$. Conditions~\eqref{sufficient-image} and~\eqref{sufficient-delta} are satisfied by the subtrees of $(V,E)$ rooted on $r_\alpha$ and $r_\delta$ through inheritance. The subtree rooted on $r_\alpha$ has $\left|L_{r_\alpha}\right|=\dg(r)=n_\ell<n$ leaves, while the subtree rooted on $r_\delta$ has $\left|L_{r_\delta}\right|=\left|\leaves_r\right|-\left|L_{r_\alpha}\right|=n-\dg(r)=n-n_\ell<n$ leaves. Therefore, the induction hypothesis implies that the subtree rooted on $r_\alpha$ is a reduction tree for $\alpha(\beta,\dg(r))$, and the subtree rooted on $r_\delta$ is a reduction tree for $\delta(\beta,\dg(r))$. Finally, \eqref{sufficient-quotients} implies that
$\beta_i/\beta_0=\beta_i/\beta_{n_\ell\floor{i/n_\ell}}\in\F_{2^{\dg(r)}}$ for $i\in\{0,\dotsc,\dg(r)-1\}$. Hence, $(V,E)$ is a reduction tree for $\beta$.
\end{proof}

If $\beta=(\beta_0,\dotsc,\beta_{n-1})\in\F^n$ is a Cantor basis with $n\geq 4$, then properties~\eqref{cantor-li} and~\eqref{cantor-subfields} of Lemma~\ref{lem:cantor} imply that $\beta_3\in\F_{16}\setminus\F_4$. Thus, $\beta_3/\beta_2$ does not belong to~$\F_4$, since $\beta_3=\beta_2/\beta_3+1$. Consequently, Proposition~\ref{prop:cantor-trees} implies that the converse of Lemma~\ref{lem:sufficient} does not hold. We now complete the proof of Theorem~\ref{thm:sufficient} by establishing linear independence and showing that the vectors $(\beta_0,\dotsc,\beta_{n-1})$ for $n\in\{1,\dotsc,n_m\}$ always satisfy the third condition of Lemma~\ref{lem:sufficient}.

\begin{lemma} Suppose that $\beta_0,\dotsc,\beta_{n_m-1}$ are given by the construction of Theorem~\ref{thm:sufficient}. Then $\beta_0,\dotsc,\beta_{n_m-1}$ are linearly independent over $\F_2$, and $\beta_i/\beta_{n_k\floor{i/n_k}}\in\F_{2^{n_k}}$ for $i\in\{0,\dotsc,n_m-1\}$ and $k\in\{0,\dotsc,m-1\}$.
\end{lemma}
\begin{proof} Suppose that $\beta_0,\dotsc,\beta_{n_m-1}$ are given by the construction of Theorem~\ref{thm:sufficient}. Let $i\in\{0,\dotsc,n_m-1\}$ and $i=\sum^{m-1}_{k=0}i_kn_k$ with $i_k\in\{0,\dotsc,n_{k+1}/n_k-1\}$ for $k\in\{0,\dotsc,m-1\}$. Then
\begin{equation*}
	\frac{\beta_i}{\beta_{n_k\floor{i/n_k}}}
	=\frac{
		\vartheta_{0,i_0}\dotsm\vartheta_{{k-1},i_{k-1}}
		\vartheta_{k,i_k}\dotsm\vartheta_{m-1,i_{m-1}}
	}{
		\vartheta_{0,0}\dotsm\vartheta_{k-1,0}
		\vartheta_{k,i_k}\dotsm\vartheta_{m-1,i_{m-1}}
	}
	=\frac{
		\vartheta_{0,i_0}\dotsm\vartheta_{k-1,i_{k-1}}
	}{
		\vartheta_{0,0}\dotsm\vartheta_{k-1,0}
	}
	\in\F_{2^{n_k}}
\end{equation*}
for $k\in\{0,\dotsc,m\}$, as required. Similarly,
\begin{equation}\label{eqn:tower-construction-quotients}
	\frac{\beta_{in_k+j}}{\beta_0}
	=
	\frac{\vartheta_{k,i}}{\vartheta_{k,0}}
	\frac{\beta_j}{\beta_0}
	\quad\text{and}\quad
	\frac{\beta_j}{\beta_0}
	=\frac{\beta_j}{\beta_{n_k\floor{j/n_k}}}
	\in\F_{2^{n_k}}
\end{equation}
for $i\in\{0,\dotsc,n_{k+1}/n_k-1\}$, $j\in\{0,\dotsc,n_k-1\}$ and $k\in\{0,\dotsc,m-1\}$. Now $\{\vartheta_{k,0}/\vartheta_{k,0},\dotsc,\vartheta_{k,n_{k+1}/n_k-1}/\vartheta_{k,0}\}$ is a basis of $\F_{2^{n_{k+1}}}/\F_{2^{n_k}}$ for $k\in\{0,\dotsc,m-1\}$. Therefore, if $\beta_0/\beta_0,\dotsc,\beta_{n_k-1}/\beta_0$ are linearly independent over $\F_2$ for some $k\in\{0,\dotsc,m-1\}$, then~\eqref{eqn:tower-construction-quotients} implies that $\beta_0/\beta_0,\dotsc,\beta_{n_{k+1}-1}/\beta_0$ are linearly independent over $\F_2$. As $n_0=1$, it follows that $\beta_0,\dotsc,\beta_{n_m-1}$ are linearly independent over $\F_2$.
\end{proof}

\subsection{Fewer multiplications for quadratic extensions}\label{sec:quadratic}

Theorem~\ref{thm:sufficient} does not place any restrictions on the choice of bases for the extensions $\F_{2^{n_{k+1}}}/\F_{2^{n_k}}$. In this section, we show that if some of these extension are quadratic, then it possible to choose their bases so that Algorithms~\ref{alg:X2M} and~\ref{alg:M2X} perform fewer multiplications.

\begin{theorem}\label{thm:ugly} Let $\beta_0,\dotsc,\beta_{n_m-1}\in\F$ be constructed as per Theorem~\ref{thm:sufficient}, $n\in\{1,\dotsc,n_m\}$,
and $V$ be the vertex set of the tree $\maxtree{n_0,\dotsc,n_m}{n}$. Recursively define vectors $\beta_v=(\beta_{v,0},\dotsc,\beta_{v,\left|\leaves_v\right|-1})$ for $v\in V$ as follows: if $v$ is the root of the tree, then $\beta_v=(\beta_0,\dotsc,\beta_{n-1})$; and if $v$ is an internal vertex, then $\beta_{v_\alpha}=\alpha(\beta_v,\dg(v))$ and $\beta_{v_\delta}=\delta(\beta_v,\dg(v))$. Suppose there exists $t\in\{0,\dotsc,m-1\}$ such that
\begin{equation}\label{eqn:ugly}
	\frac{n_{t+1}}{n_t}=2
	\quad\text{and}\quad
	\trace_{\F_{2^{n_{t+1}}}/\F_{2^{n_{t\vphantom{+1}}}}}
	\left(\frac{\vartheta_{t,1}}{\vartheta_{t,0}}\right)
	=1.
\end{equation}
Then $\beta_{v_\delta,0}=1$ for all $v\in V$ such that $\dg(v)=n_t$.
\end{theorem}

The definition of the vectors $\beta_v$ in Theorem~\ref{thm:ugly} matches that of Section~\ref{sec:algorithms}. Consequently, for $\beta=(\beta_0,\dotsc,\beta_{n-1})$ given by the construction of Theorem~\ref{thm:sufficient}, and the reduction tree $\maxtree{n_0,\dotsc,n_m}{n}$, Lines~\ref{X2M:scale}--\ref{X2M:scale-end} of Algorithm~\ref{alg:X2M} (similarly, Lines~\ref{M2X:scale}--\ref{M2X:scale-end} of Algorithm~\ref{alg:M2X}) perform no multiplications if the input vertex $v$ satisfies $\dg(v)=n_t$ for some $t$ such that~\eqref{eqn:ugly} holds. If $n_{t+1}/n_t=2$, then we may ensure that the condition on the trace in~\eqref{eqn:ugly} is satisfied by taking $\vartheta_{t,0}=1$ and $\vartheta_{t,1}\in\F_{2^{n_{t+1}}}$ with $\trace_{\F_{2^{n_{t+1}}}/\F_{2^{n_{t\vphantom{+1}}}}}(\vartheta_{t,1})=1$, which yields a basis since the trace of one is equal to zero. Therefore, it is possible to achieve a significant reduction in the number of multiplications performed by Algorithms~\ref{alg:X2M} and~\ref{alg:M2X} when the tower used in the construction contains several quadratic extensions. Returning to Example~\ref{ex:tower}, we see such an improvement for $\F=\F_{2^{12}}$ by looking at the relative number of multiplications performed for the daggered tuples in Figure~\ref{fig:relative-mults}. For these tuples, it is assumed that $\beta_{v_\delta,0}=1$ if and only if $\dg(v)=n_t$ for some $t$ such that $n_{t+1}/n_t=2$.


We now turn our attention to the proof of Theorem~\ref{thm:ugly}.

\begin{lemma}\label{lem:beta_v_delta} Assume the hypothesis and notation of Theorem~\ref{thm:ugly}. If $v\in V$ is an internal vertex and $\dg(v)=n_\ell$, then there exist elements $\vartheta_{\ell,0}',\vartheta_{\ell,1}',\dotsc\in\F_{2^{n_{\ell+1}}}$ that are linearly independent over $\F_{2^{n_\ell}}$, for which
\begin{equation}\label{eqn:beta_v_delta_i}
	\beta_{v_\delta,i}
	=
	\frac{\vartheta_{0,i_0}}
	{\vartheta_{0,0}}
	\dotsm
	\frac{\vartheta_{\ell-1,i_{\ell-1}}}
	{\vartheta_{\ell-1,0}}
	\vartheta_{\ell,i_\ell}'
	\quad\text{such that}\quad
	\sum^\ell_{k=0}i_kn_k=i,
\end{equation}
for $i\in\{0,\dotsc,\left|\leaves_{v_\delta}\right|-1\}$. Furthermore,
\begin{equation}\label{eqn:vartheta_ell_0}
	\vartheta_{\ell,0}'
	=
	\left(
		\frac{\vartheta_{\ell,1}}
		{\vartheta_{\ell,0}}
	\right)^{2^{n_\ell}}
	-
	\frac{\vartheta_{\ell,1}}
	{\vartheta_{\ell,0}}
\end{equation}
if there is no vertex $u\in V$ such that $u_\delta=v$ and $\dg(u)=n_\ell$.
\end{lemma}
\begin{proof} Throughout the proof, if $i\in\{0,\dotsc,n_m-1\}$, then $i_0,\dotsc,i_{m-1}$ denote the coefficients of the expansion $i=\sum^{m-1}_{k=0}i_kn_k$ such that $i_k\in\{0,\dotsc,n_{k+1}/n_k-1\}$ for $k\in\{0,\dotsc,m-1\}$. We note in particular that $\{0,\dotsc,\left|\leaves_v\right|-1\}\subseteq\{0,\dotsc,n_m-1\}$ for $v\in V$, since $\left|\leaves_v\right|\leq n\leq n_m$. We begin by showing that the assertions of the lemma hold for a special subset of the internal vertices in the tree, before completing the proof by induction.

Let $v\in V$ be an internal vertex such that the (possibly trivial) path $r=v_0,\dotsc,v_h=v$ that connects $v$ to the root vertex $r\in V$ satisfies $v_i=(v_{i-1})_\alpha$ for $i\in\{1,\dotsc,h\}$. Then
\begin{equation*}
	\beta_{v,i}=\beta_{v_{h-1},i}=\dotsb=\beta_{v_0,i}=\beta_{r,i}=\prod^{m-1}_{k=0}\vartheta_{k,i_k}
	\quad
	\text{for $i\in\{0,\dotsc,\left|\leaves_v\right|-1\}$}.
\end{equation*}
As $v$ is an internal vertex, $\dg(v)=\max\{n_k\mid n_k<\left|\leaves_v\right|\}$ by the definition of the tree $\maxtree{n_0,\dotsc,n_m}{n}$. Thus, $\dg(v)=n_\ell$ for some $\ell\in\{0,\dotsc,m-1\}$. Moreover, $\left|\leaves_{v_\delta}\right|=\left|\leaves_v\right|-n_\ell<n_{\ell+1}$, since otherwise the maximality of $n_\ell$ is contradicted. Consequently, 
\begin{equation}\label{eqn:quotients}
	\frac{\beta_{v,n_\ell+i}}{\beta_{v,0}}
	=
	\frac{\vartheta_{0,i_0}}
	{\vartheta_{0,0}}
	\dotsm
	\frac{\vartheta_{\ell-1,i_{\ell-1}}}
	{\vartheta_{\ell-1,0}}
	\frac{\vartheta_{\ell,1+i_\ell}}{\vartheta_{\ell,0}}
	\quad
	\text{for $i\in\{0,\dotsc,\left|\leaves_{v_\delta}\right|-1\}$}.
\end{equation}
As the quotients $\vartheta_{k,i}/\vartheta_{k,0}\in\F_{2^{n_{k+1}}}\subseteq\F_{2^{n_\ell}}$ for $k\in\{0,\dotsc,\ell-1\}$, it follows that~\eqref{eqn:beta_v_delta_i} holds with
\begin{equation*}
	\vartheta_{\ell,i}'
	=
	\left(\frac{\vartheta_{\ell,i+1}}{\vartheta_{\ell,0}}\right)^{2^{n_\ell}}
	-\frac{\vartheta_{\ell,i+1}}{\vartheta_{\ell,0}}
	\quad\text{for $i\in\{0,\dotsc,n_{\ell+1}/n_\ell-2\}$}.
\end{equation*}
Consequently, \eqref{eqn:vartheta_ell_0} also holds. Finally, $\vartheta_{\ell,0}',\dotsc,\vartheta_{\ell,n_{\ell+1}/n_\ell-2}'$ inherit linear independence over $\F_{2^{n_\ell}}$ from $\vartheta_{\ell,0},\dotsc,\vartheta_{\ell,n_{\ell+1}/n_\ell-1}$, since
\begin{equation*}
	\sum^{n_{\ell+1}/n_\ell-2}_{i=0}
	\lambda_{i+1}
	\vartheta_{\ell,i}'=0
	\quad\text{if and only if}\quad
	\sum^{n_{\ell+1}/n_\ell-1}_{i=1}
	\lambda_i
	\frac{\vartheta_{\ell,i}}
	{\vartheta_{\ell,0}}
	\in\F_{2^{n_\ell}}
\end{equation*}
for $\lambda_1,\dotsc,\lambda_{n_{\ell+1}/n_\ell-1}\in\F_{2^{n_\ell}}$.

We now proceed by induction on the depth of $v$, i.e., on the length $h$ of the path $r=v_0,\dotsc,v_h=v$ that connects $v$ to the root vertex of the tree. If $v\in V$ is an internal vertex of depth zero, then $v$ is the root of the tree which is covered by the already proved case. Therefore, let $h$ be a positive integer, and suppose that the assertions of the lemma holds for all internal vertices with depth less than $h$. Let $v\in V$ be an internal vertex of depth $h$ (if no such vertex exists, then we are done) and $r=v_0,\dotsc,v_h=v$ be the path that connects $v$ to the root vertex $r\in V$. We may assume that $v_i=(v_{i-1})_\delta$ for some $i\in\{1,\dotsc,h\}$. Let $j$ the maximum of all such indices. Then $v_{j-1}$ is an internal vertex. Thus, $\dg(v_{j-1})=n_k$ for some $k\in\{0,\dotsc,m-1\}$. Consequently, the induction hypothesis and the choice of $j$ imply that there exists a basis $\{\vartheta_{k,0}'',\dotsc,\vartheta_{k,n_{k+1}/n_k-1}''\}$ of $\F_{2^{n_{k+1}}}/\F_{2^{n_k}}$ such that
\begin{equation}\label{eqn:beta_v_i}
	\beta_{v,i}
	=\beta_{v_{h-1},i}=\dotsb=\beta_{v_j,i}
	=
	\frac{\vartheta_{0,i_0}}
	{\vartheta_{0,0}}
	\dotsm
	\frac{\vartheta_{k-1,i_{k-1}}}
	{\vartheta_{k-1,0}}
	\vartheta_{k,i_k}''
	\quad\text{for $i\in\{0,\dotsc,\left|\leaves_v\right|-1\}$}.
\end{equation}

As $v$ is an internal vertex, $\dg(v)=n_\ell$ for some $\ell\in\{0,\dotsc,m-1\}$. Moreover, the definition of the tree $\maxtree{n_0,\dotsc,n_m}{n}$ implies that
\begin{equation*}
 n_\ell
 =\max\left\{n_t\mid n_t<\left|\leaves_v\right|\right\}
 \leq\max\left\{n_t\mid n_t<\left|\leaves_{v_{j-1}}\right|\right\}
 =n_k,
\end{equation*}
since $v$ is descended from $v_{j-1}$. If $\ell=k$, then~\eqref{eqn:beta_v_i} implies that~\eqref{eqn:beta_v_delta_i} holds with
\begin{equation*}
	\vartheta_{\ell,i}'
	=
	\left(
		\frac{\vartheta_{\ell,i+1}''}
		{\vartheta_{\ell,0}''}
	\right)^{2^{n_\ell}}
	-\frac{\vartheta_{\ell,i+1}''}
	{\vartheta_{\ell,0}''}
	\quad\text{for $i\in\{0,\dotsc,n_{\ell+1}/n_\ell-2\}$},
\end{equation*}
which inherit linear independence over $\F_{2^{n_\ell}}$ from $\vartheta_{k,0}'',\dotsc,\vartheta_{k,n_{k+1}/n_k-1}''$. Moreover, we must have $j=h$, since otherwise the maximality of $j$ implies that $v$ is descended from $(v_j)_\alpha$, which in-turn implies that $n_k=n_\ell<\left|\leaves_v\right|\leq\dg(v_j)\leq n_k$. It follows that $u=v_{h-1}$ satisfies $u_\delta=v$ and $\dg(u)=n_\ell$. Consequently, we are not required to show that~\eqref{eqn:vartheta_ell_0} holds in this case.

If $\ell<k$, then $\left|\leaves_{v_\delta}\right|=\left|\leaves_v\right|-n_\ell<n_{\ell+1}\leq n_k$, since $n_\ell=\max\{n_t\mid n_t<\left|\leaves_v\right|\}$. Thus, \eqref{eqn:beta_v_i} implies that \eqref{eqn:quotients} holds, which we have already shown to be sufficient for the two assertions of the lemma to hold. Hence, the lemma follows by induction.
\end{proof}

\begin{remark} Lemma~\ref{lem:beta_v_delta} implies that $\beta_{v_\delta,0}=\vartheta_{\ell,0}'$ lies in the subfield $\F_{2^{n_{\ell+1}}}$ regardless of the choice of bases used in the construction of Theorem~\ref{thm:sufficient}. Consequently, if $\beta_{v_\delta,0}$ cannot be forced to equal one, then it may still be possible to reduce the cost of the multiplications performed in Lines~\ref{X2M:scale}--\ref{X2M:scale-end} of Algorithm~\ref{alg:X2M} (similarly, Lines~\ref{M2X:scale}--\ref{M2X:scale-end} of Algorithm~\ref{alg:M2X}) by choosing the representation of the elements in $\F$ so that the cost of multiplication is reduced whenever one of the multiplicands belongs to $\F_{2^{n_{\ell+1}}}$. Such optimisations have previously been shown to be beneficial in practice, particularly for multiplications by elements of small subfields, by Bernstein and Chou~\cite{bernstein2014} and Chen et al.~\cite{chen2017}. These considerations also extend to the algorithms for conversion between the LCH and the Newton or Lagrange bases of Sections~\ref{sec:NX}--\ref{sec:X2L}. For these algorithms, if the tower used in the construction of the basis contains small subfields, then Lemma~\ref{lem:beta_v_delta} can be used to show that some of the precomputed elements $\PushDown_v(u,\sigma_{v,i})$ belong to small subfields. Consequently, if the initial shift parameter $\lambda$ also lies in a small subfield, as is the case when it is zero, then so too do some of the multiplicands in the base cases of the algorithms.
\end{remark}

\begin{proof}[Proof of Theorem~\ref{thm:ugly}] Suppose there exists $t\in\{0,\dotsc,m-1\}$ such that~\eqref{eqn:ugly} holds, and there exists a vertex $v\in V$ such that $\dg(v)=n_t$. If $v=u_\delta$ for some $u\in V$, then $\dg(u)\neq n_t$, since otherwise $n_{t+1}=2n_t<n_t+\left|\leaves_v\right|=n_t+(\left|\leaves_u\right|-n_t)=\left|\leaves_u\right|$, contradicting the maximality of $n_t$. Therefore, Lemma~\ref{lem:beta_v_delta} implies that
\begin{equation*}
	\beta_{v_\delta,0}
	=
	\frac{\vartheta_{0,0}}
	{\vartheta_{0,0}}
	\dotsm
	\frac{\vartheta_{t-1,0}}
	{\vartheta_{t-1,0}}
	\left(
		\left(
			\frac{\vartheta_{t,1}}
			{\vartheta_{t,0}}
		\right)^{2^{n_t}}
		-
		\frac{\vartheta_{t,1}}
		{\vartheta_{t,0}}
	\right)
	=
	\trace_{\F_{2^{n_{t+1}}}/\F_{2^{n_{t\vphantom{+1}}}}}
	\left(
		\frac{\vartheta_{t,1}}
		{\vartheta_{t,0}}
	\right)
	=1,
\end{equation*}
as required.
\end{proof}

\subsection{Generalised Cantor basis}\label{sec:cantor}

Gao and Mateer propose a generic method of constructing Cantor bases in the appendix of their paper~\cite{gao2010}. We generalise their construction to one that extends an arbitrarily chosen basis of $\F_{2^\basedeg}/\F_2$ to a basis of $\F_{2^{2^m\basedeg}}/\F_2$ that enjoys properties similar to those offered by Cantor bases. The original construction of Gao and Mateer then corresponds to the case $\basedeg=1$. By generalising their construction, we are able to take advantage of quadratic extensions in a different manner to the previous section in order to provide a greater selection of reduction trees.

Hereafter, we assume that $\F_{q^{2^m}}\subseteq\F$ with $m\in\N$ positive and $q=2^\basedeg$ for some positive $\basedeg\in\N$. We also fix a basis $\beta_0,\dotsc,\beta_{2^m\basedeg-1}$ of $\F_{q^{2^m}}/\F_2$ which is given by the following generalisation of the construction of Gao and Mateer: choose a basis $\{\vartheta_0,\dotsc,\vartheta_{\basedeg-1}\}$ of $\F_q/\F_2$, choose $\beta_{(2^m-1)\basedeg},\dotsc,\beta_{2^m\basedeg-1}\in\F_{q^{2^m}}$ such that 
\begin{equation*}
	\trace_{\F_{q^{2^m}}/\F_q}
	\left(\beta_{(2^m-1)\basedeg+i}\right)
	=\vartheta_i
	\quad\text{for $i\in\{0,\dotsc,\basedeg-1\}$},
\end{equation*}
and recursively define
\begin{equation*}
	\beta_i=\beta^q_{i+\basedeg}-\beta_{i+\basedeg}
	\quad\text{for $i\in\{0,\dotsc,(2^m-1)\basedeg-1\}$}.
\end{equation*}
For this construction, we provide generalisations of the properties of Cantor bases given in Lemma~\ref{lem:cantor}. The properties are then used to prove the following theorem, which provides a method of constructing reduction trees for the vectors $(\beta_0,\dotsc,\beta_{n-1})$ for $n\in\{1,\dotsc,2^m\basedeg\}$.

\begin{theorem}\label{thm:cantor-tree} Let $n\in\{1,\dotsc,2^m\basedeg\}$, $T_0=(V_0,E_0)$ be a full binary tree with~$\ceil{n/\basedeg}$ leaves such that $\image(\dg)\subseteq\{0,2^0,2^1,\dotsc,2^{\ceil{\log_2\ceil{n/\basedeg}}-1}\}$, and $T_i=(V_i,E_i)$ be a reduction tree for $(\vartheta_0,\dotsc,\vartheta_{\min(n-(i-1)\basedeg,\basedeg)-1})$, for $i\in\{1,\dotsc,\ceil{n/\basedeg}\}$. Let $u_1,\dotsc,u_{\ceil{n/\basedeg}}\in V_0$ be the leaves of $T_0$, ordered such that for all $i,j\in\{1,\dotsc,\ceil{n/\basedeg}\}$ with $i<j$, there exists an internal vertex $v\in V_0$ with $u_i\in\leaves_{v_\alpha}$ and~$u_j\in\leaves_{v_\delta}$. Construct a new tree $T=(V,E)$ by identifying the root vertex of $T_i$ with $u_i$ for $i\in\{1,\dotsc,\ceil{n/\basedeg}\}$, as shown in Figure~\ref{fig:tree-construction}. Then $T$ is a reduction tree for~$(\beta_0,\dotsc,\beta_{n-1})$.
\begin{figure}
	\begin{tikzpicture}
		\tikzset{vertex/.style={circle,fill=black,inner sep=0pt,minimum size=3pt,label=above:{\small#1}}}
		\def\x{3.5}
		\def\y{1}
		
		\path[draw, dashed, use as bounding box]
		($(-1,0.3)+(-\x,\y)$) 
		-- ($(1,0.3)+(\x,\y)$) node [pos=.96, label=below:{$T$}] {} 
		-- ($(1,-1.75)+(\x,-\y)$) 
		-- ($(-1,-1.75)+(-\x,-\y)$) 
		-- cycle;
		
		\draw (0,0) ellipse ({\x} and {\y});
		
		\def\ph{\vphantom{\ceil{n/\basedeg}}}
		\node [vertex = $u_{1\ph}$]                (u1) at ($(0,0)+(-140:{\x} and {\y})$) {};
		\node [vertex = $u_{2\ph}$]                (u2) at ($(0,0)+(-110:{\x} and {\y})$) {};
		\node [vertex = $u_{\ceil{n/\basedeg}-1}$] (u3) at ($(0,0)+( -70:{\x} and {\y})$) {};
		\node [vertex = $u_{\ceil{n/\basedeg}}$]   (u4) at ($(0,0)+( -40:{\x} and {\y})$) {};
		
		\draw [rotate around={-10:(0,0)}] ($(u1)+(0,-0.75)$) ellipse (0.65 and 0.75);
		\draw [rotate around={ -2:(0,0)}] ($(u2)+(0,-0.75)$) ellipse (0.65 and 0.75);
		\draw [rotate around={  2:(0,0)}] ($(u3)+(0,-0.75)$) ellipse (0.65 and 0.75);
		\draw [rotate around={ 10:(0,0)}] ($(u4)+(0,-0.6 )$) ellipse (0.5  and 0.6);
		
		\draw (u1) -- ++(  35: .2) edge [densely dotted] ++(  35:.16);
		\draw (u1) -- ++( -65: .2) edge [densely dotted] ++( -65:.16);
		\draw (u1) -- ++(-135: .2) edge [densely dotted] ++(-135:.16);
		
		\draw (u2) -- ++(  55:.15) edge [densely dotted] ++(  55:.16);
		\draw (u2) -- ++( -58: .2) edge [densely dotted] ++( -58:.16);
		\draw (u2) -- ++(-128: .2) edge [densely dotted] ++(-128:.16);
		
		\draw (u3) -- ++( 127:.15) edge [densely dotted] ++( 127:.12);
		\draw (u3) -- ++( -53: .2) edge [densely dotted] ++( -53:.16);
		\draw (u3) -- ++(-123: .2) edge [densely dotted] ++(-123:.16);
		
		\draw (u4) -- ++( 145: .2) edge [densely dotted] ++( 145:.16);
		\draw (u4) -- ++( -45: .2) edge [densely dotted] ++( -45:.16);
		\draw (u4) -- ++(-115: .2) edge [densely dotted] ++(-115:.16);
		
		\node () at               (0,-1.7) {$\cdots$};
		\node () at                  (0,0) {$T_0$};
		\node () at ($(u1)+(-0.15,-0.80)$) {\small$T_{1\ph}$};
		\node () at ($(u2)+(-0.04,-0.80)$) {\small$T_{2\ph}$};
		\node () at ($(u3)+( 0.04,-0.80)$) {\small$T_{\ceil{n/\basedeg}-1}$};
		\node () at ($(u4)+( 0.15,-0.65)$) {\small$T_{\ceil{n/\basedeg}}$};
	\end{tikzpicture}
	\caption{Construction of Theorem~\ref{thm:cantor-tree}.}
	\label{fig:tree-construction}
\end{figure}
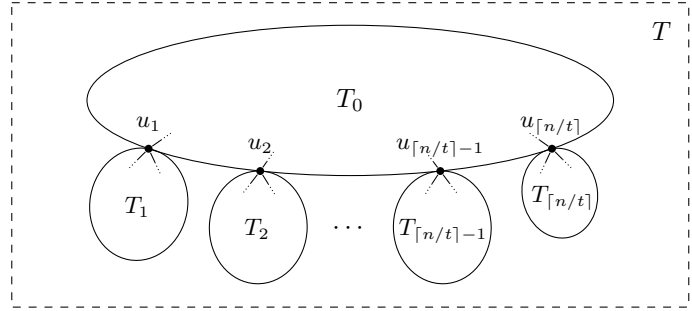
\end{theorem}

Theorem~\ref{thm:cantor-tree} provides greater freedom than Theorem~\ref{thm:sufficient} by not requiring the inequality $\dg(v_\delta)\leq\dg(v)$ to hold for $v\in V$ that are initially internal vertices in the tree $T_0$. Proposition~\ref{prop:trivial-reduction-tree} guarantees the existence of trees $T_1,\dotsc,T_{\ceil{n/t}}$ to use in the construction. We can of course provide a better selection for these trees if the methods of Sections~\ref{sec:tower} and~\ref{sec:quadratic} are used to construct the basis $\{\vartheta_0,\dotsc,\vartheta_{\basedeg-1}\}$.

The remainder of the section is dedicated to the proof of Theorem~\ref{thm:cantor-tree}.

\begin{lemma}\label{lem:gen-cantor} The following hold:
\begin{enumerate}
	\item\label{gen-cantor-sum} $\beta_i=\sum^j_{r=0}\binom{j}{r}\beta^{q^r}_{i+j\basedeg}$	for $i\in\{0,\dotsc,(2^m-j)\basedeg-1\}$ and $j\in\{0,\dotsc,2^m-1\}$,
	\item\label{gen-cantor-delta} $\beta_i=\beta^{q^{2^k}}_{i+2^k\basedeg}-\beta_{i+2^k\basedeg}$ for $i\in\{0,\dotsc,(2^m-2^k)\basedeg-1\}$ and $k\in\{0,\dotsc,m-1\}$,
	\item\label{gen-cantor-theta} $\beta_i=\vartheta_i$ for $i\in\{0,\dotsc,\basedeg-1\}$,
	\item\label{gen-cantor-subfields} $\beta_0,\dotsc,\beta_{2^k\basedeg-1}\in\F_{q^{2^k}}$ for $k\in\{0,\dotsc,m-1\}$, and
	\item\label{gen-cantor-li} $\beta_0,\dotsc,\beta_{2^m\basedeg-1}$ are linearly independent over $\F_2$.
\end{enumerate}
\end{lemma}
Our proof of Lemma~\ref{lem:gen-cantor} generalises arguments found in Cantor's paper~\cite{cantor1989} and the paper of Gao and Mateer~\cite{gao2010}.

\begin{proof} We prove~\eqref{gen-cantor-sum} by induction on $j$. It is clear that~\eqref{gen-cantor-sum} holds if $j=0$, regardless of the value of $i$. Therefore, suppose that~\eqref{gen-cantor-sum} holds for some $j\in\{0,\dotsc,2^m-2\}$ and each $i\in\{0,\dotsc,(2^m-j)\basedeg-1\}$. Then
\begin{equation*}
	\beta_{i+\basedeg}
	=\sum^j_{r=0}\binom{j}{r}\beta^{q^r}_{i+\basedeg+j\basedeg}
	=\sum^j_{r=0}\binom{j}{r}\beta^{q^r}_{i+(j+1)\basedeg}
\end{equation*}
for $i\in\{0,\dots,(2^m-j-1)\basedeg-1\}$. As $(2^m-j-1)\basedeg-1\leq (2^m-1)\basedeg-1$, it follows that
\begin{align*}
	\beta_i
	&=\beta^q_{i+\basedeg}-\beta_{i+\basedeg}\\
	&=\sum^j_{r=0}\binom{j}{r}\left(\beta^{q^{r+1}}_{i+(j+1)\basedeg}-\beta^{q^r}_{i+(j+1)\basedeg}\right)\\
	&=\beta_{i+(j+1)\basedeg}+\sum^{j+1}_{r=1}\left(\binom{j}{r-1}+\binom{j}{r}\right)\beta^{q^r}_{i+(j+1)\basedeg}\\
	&=\beta_{i+(j+1)\basedeg}+\sum^{j+1}_{r=1}\binom{j+1}{r}\beta^{q^r}_{i+(j+1)\basedeg}\\
	&=\sum^{j+1}_{r=0}\binom{j+1}{r}\beta^{q^r}_{i+(j+1)\basedeg}
\end{align*}
for $i\in\{0,\dotsc,(2^m-j-1)\basedeg-1\}$. Thus, property~\eqref{gen-cantor-sum} holds.

For $i,j\in\N$, Lucas' lemma~\cite[p.~230]{lucas1878} (see also~\cite{fine1947}) implies that $\tbinom{i}{j}\equiv 1\pmod{2}$ if and only if $[j]_k\leq[i]_k$ for $k\in\N$. Using the lemma, property~\eqref{gen-cantor-delta} follows from property~\eqref{gen-cantor-sum} by setting $j=2^k$. Similarly, Lucas' lemma and property~\eqref{gen-cantor-sum} with $j=(2^{m-k}-1)2^k$ implies that
\begin{equation*}
	\begin{aligned}
		\beta_{(2^k-1)\basedeg+i}
		&=
		\sum^{(2^{m-k}-1)2^k}_{r=0}
		\binom{(2^{m-k}-1)2^k}{r}
		\beta^{q^r}_{(2^m-1)\basedeg+i}\\
		&=\sum^{2^{m-k}-1}_{r=0}
		\beta^{q^{2^kr}}_{(2^m-1)\basedeg+i}\\
		&=\trace_{\F_{q^{2^m}}/\F_{q^{2^k}}}
		\left(\beta_{(2^m-1)\basedeg+i}\right)
	\end{aligned}
\end{equation*}
for $i\in\{0,\dots,\basedeg-1\}$ and $k\in\{0,\dotsc,m-1\}$. Setting $k=0$, property~\eqref{gen-cantor-theta} follows by the choice of $\beta_{(2^m-1)\basedeg},\dotsc,\beta_{2^m\basedeg-1}$. Moreover, the trace formula implies that $\beta_{(2^k-1)\basedeg},\dotsc,\beta_{2^k\basedeg-1}\in\F_{q^{2^k}}$ for $k\in\{0,\dotsc,m-1\}$, after-which the recursive definition of $\beta_0,\dotsc,\beta_{(2^m-1)\basedeg-1}$ implies that property~\eqref{gen-cantor-subfields} holds.

Property~\eqref{gen-cantor-theta} implies that $\beta_0,\dotsc,\beta_{\basedeg-1}$ are linearly independent over $\F_2$, and belong to the kernel of the $\F_2$-linear map $\varphi:\F\rightarrow\F$ given by $\omega\mapsto\omega^q-\omega$. Thus, for $i\in\{2,\dotsc,2^m\}$, any nontrivial $\F_2$-linear relation amongst $\beta_0,\dotsc,\beta_{i\basedeg-1}$, which necessarily involves at least one of $\beta_{(i-1)\basedeg},\dotsc,\beta_{i\basedeg-1}$, translates under $\varphi$ to a nontrivial relation amongst $\beta_0,\dotsc,\beta_{(i-1)\basedeg-1}$. It follows that property~\eqref{gen-cantor-li} holds by induction on~$i$. 
\end{proof}

Lemma~\ref{lem:gen-cantor} provides generalisations of the properties of Cantor bases given in Lemma~\ref{lem:cantor}. The following lemma may be viewed as a partial generalisation Proposition~\ref{prop:cantor-trees}.
              
\begin{lemma}\label{lem:cantor-tree} Let $(V,E)$ be a full binary tree with $n\leq 2^m\basedeg$ leaves that satisfies the following conditions:
\begin{enumerate}
	\item\label{cantor-tree-i} if $v\in V$ such that $\left|\leaves_v\right|>\basedeg$, then $\dg(v)\in\{2^k\basedeg\mid k<\ceil{\log_2\ceil{n/\basedeg}}\}$, and
	\item\label{cantor-tree-ii} if $v\in V$ such that $\left|\leaves_v\right|\leq\basedeg$, and $v$ is either the root of the tree or a child of a vertex $v'\in V$ with $\left|\leaves_{v'}\right|>\basedeg$, then the subtree of $(V,E)$ rooted on $v$ is a reduction tree for $(\beta_0,\dotsc,\beta_{\left|\leaves_v\right|-1})$.
\end{enumerate}
Then $(V,E)$ is a reduction tree for $(\beta_0,\dotsc,\beta_{n-1})$.
\end{lemma}

The following technical lemma is required for the proof of Lemma~\ref{lem:cantor-tree}.

\begin{lemma}\label{lem:scalar-invariance} Let $\mu\in\F^n$ have linearly independent entries over $\F_2$, $(V,E)$ be a reduction tree for $\mu$, and $\omega\in\F$ be nonzero. Then $(V,E)$ is a reduction tree for $\omega\mu$.
\end{lemma}
\begin{proof} We prove the lemma by induction on $n$. The lemma holds trivially for $n=1$. Therefore, let $n\geq 2$ and suppose that the lemma holds for all smaller values of $n$. Let $\mu=(\mu_0,\dotsc,\mu_{n-1})\in\F^n$ have linearly independent entries over~$\F_2$, $(V,E)$ be a reduction tree for $\mu$, $r\in V$ be the root vertex of the tree, and $\omega\in\F$ be nonzero. Then $\mu_i/\mu_0=\omega\mu_i/(\omega\mu_0)\in\F_{2^{\dg(r)}}$ for $i\in\{0,\dotsc,\dg(r)-1\}$, the induction hypothesis implies that the subtree rooted on $r_\alpha$ is a reduction tree for $\omega\alpha(\mu,\dg(r))=\alpha(\omega\mu,\dg(r))$, and the subtree rooted on $r_\delta$ is a reduction tree for $\delta(\mu,\dg(r))=\delta(\omega\mu,\dg(r))$. Therefore, $(V,E)$ is a reduction tree for $\omega\mu$. Hence, the lemma follows by induction.
\end{proof}

\begin{proof}[Proof of Lemma~\ref{lem:cantor-tree}] We prove the lemma by induction of $n$. Condition~\eqref{cantor-tree-ii} implies that the lemma holds trivially if $n\leq\basedeg$. Therefore, let $n\in\{\basedeg+1,\dotsc,2^m\basedeg\}$ and suppose that the lemma is true for all smaller values of $n$. Let $(V,E)$ be a full binary tree with $n$ leaves that satisfies conditions~\eqref{cantor-tree-i} and~\eqref{cantor-tree-ii} of the lemma. Let $\beta=(\beta_0,\dotsc,\beta_{n-1})$ and $r\in V$ be the root vertex of the tree. Then $\left|\leaves_r\right|=n>\basedeg$. Thus, \eqref{cantor-tree-i} implies that $\dg(r)=2^k\basedeg$ for some $k<\ceil{\log_2\ceil{n/\basedeg}}\leq m$. Therefore, property~\eqref{gen-cantor-subfields} of Lemma~\ref{lem:gen-cantor} implies that $\beta_i/\beta_0\in\F_{2^{\dg(r)}}$ for $i\in\{0,\dotsc,\dg(r)-1\}$. Moreover, properties~\eqref{gen-cantor-delta} and~\eqref{gen-cantor-subfields} of the lemma imply that
\begin{equation}\label{eqn:gen-cantor-alpha-delta}
	\alpha\left(\beta,\dg(r)\right)
	=\left(\beta_0,\dotsc,\beta_{\dg(r)-1}\right)
	\quad\text{and}\quad
	\delta\left(\beta,\dg(r)\right)
	=\frac{1}{\beta_0}\left(
		\beta_0,
		\dotsc,
		\beta_{n-\dg(r)-1}
	\right).
\end{equation}
The subtrees of $(V,E)$ rooted on $r_\alpha$ and $r_\delta$ satisfy the conditions of the lemma through inheritance. Thus, the induction hypothesis and~\eqref{eqn:gen-cantor-alpha-delta} imply that the subtree rooted on $r_\alpha$ is a reduction tree for $\alpha(\beta,\dg(r))$. Similarly, the induction hypothesis, Lemma~\ref{lem:scalar-invariance} and~\eqref{eqn:gen-cantor-alpha-delta} imply that the subtree root on $r_\delta$ is a reduction tree for $\delta(\beta,\dg(r))$. Therefore, $(V,E)$ is a reduction tree for $\beta$. Hence, the lemma follows by induction.
\end{proof}


We now complete the proof of Theorem~\ref{thm:cantor-tree} by showing that its construction produces binary trees that satisfy the conditions of Lemma~\ref{lem:cantor-tree}.

\begin{proof}[Proof of Theorem~\ref{thm:cantor-tree}] For $i\in\{1,\dotsc,\ceil{n/\basedeg}\}$, $(V_i,E_i)$ is a full binary tree with $\min(n-(i-1)\basedeg,\basedeg)$ leaves. Therefore, it is clear that $(V,E)$ is a full binary tree with
\begin{equation*}
	\sum^{\ceil{n/\basedeg}}_{i=1}
	\min\left(n-(i-1)\basedeg,\basedeg\right)
	=n
\end{equation*}
leaves. We show that $(V,E)$ satisfies the conditions of Lemma~\ref{lem:cantor-tree}.

Suppose there exists a vertex $v\in V$ such that $\left|\leaves_v\right|>\basedeg$. Then $v$ is not descended from or equal to $u_i$ for $i\in\{1,\dotsc,\ceil{n/\basedeg}\}$, since $(V_i,E_i)$ has at most $\basedeg$ leaves. By the choice of $(V_0,E_0)$, it follows that $2^k$ of the vertices $u_i$ are descended from~$v_\alpha$ for some $k<\ceil{\log_2\ceil{n/\basedeg}}$. Let $i_1,\dotsc,i_{2^k}$ be the indices of these vertices. Then $i_1,\dotsc,i_{2^k}<\ceil{n/\basedeg}$, since the ordering of the vertices $u_i$ implies that $u_{\ceil{n/\basedeg}}$ must be equal to $v_\delta$ or one of its descendants. It follows that the subtrees of $(V,E)$ rooted on $u_{i_1},\dotsc,u_{i_{2^k}}$ each have $\basedeg$ leaves. Thus, $\dg(v)=2^k\basedeg$ for some $k<\ceil{\log_2\ceil{n/\basedeg}}$. Therefore, $(V,E)$ satisfies condition~\eqref{cantor-tree-i} of Lemma~\ref{lem:cantor-tree}.

Suppose there exists a vertex $v\in V$ such that $\left|\leaves_v\right|\leq\basedeg$, and $v$ is either the root of the tree or the child of a vertex $v'\in V$ with $\left|\leaves_{v'}\right|>\basedeg$. Then $v$ is descended from or equal to $u_i$ for some $i\in\{1,\dotsc,\ceil{n/\basedeg}\}$, since the subtrees rooted on $u_1,\dotsc,u_{\ceil{n/t}-1}$ each have $\basedeg$ leaves, while the subtree rooted on $u_{\ceil{n/t}}$ has at least one leaf. If $v$ is the root of $(V,E)$, then $(V,E)=(V_i,E_i)$. If $v$ is the child of a vertex $v'\in V$ such that $\left|\leaves_{v'}\right|>\basedeg$, and thus $\left|\leaves_{v'}\right|>\left|\leaves_{u_i}\right|$, then $u_i$ is a descendant of $v'$. In either case, $v$ is equal to $u_i$. Thus, the choice of $(V_i,E_i)$ and property~\eqref{gen-cantor-theta} of Lemma~\ref{lem:gen-cantor} imply that the subtree of $(V,E)$ rooted on $v$ is a reduction tree for $(\beta_0,\dotsc,\beta_{\left|\leaves_v\right|-1})$. Hence, $(V,E)$ satisfies condition~\eqref{cantor-tree-ii} of Lemma~\ref{lem:cantor-tree}.
\end{proof}

\bibliographystyle{amsplain}
\providecommand{\bysame}{\leavevmode\hbox to3em{\hrulefill}\thinspace}
\providecommand{\MR}{\relax\ifhmode\unskip\space\fi MR }
\providecommand{\MRhref}[2]{%
  \href{http://www.ams.org/mathscinet-getitem?mr=#1}{#2}
}
\providecommand{\href}[2]{#2}

\end{document}